\documentclass[11pt]{article}

\usepackage[margin=1in]{geometry}
\usepackage{amsmath,amsthm,amsfonts,amssymb,latexsym}
\usepackage{nicefrac,xspace}
\usepackage{graphicx}
\usepackage[pdftex,usenames,dvipsnames]{xcolor}
\usepackage{multirow}
\usepackage{enumitem}
\usepackage{soul}

\newtheorem{theorem}{Theorem}[section]
\newtheorem{corollary}{Corollary}[section]
\newtheorem{lemma}{Lemma}[section]
\newtheorem{remark}{Remark}[section]
\newtheorem{observation}{Observation}[section]

\DeclareMathVersion{lifts}
\SetSymbolFont{letters}{lifts}{U}{eur}{m}{n}
\DeclareMathAlphabet{\eu}{U}{eur}{m}{n}
\DeclareSymbolFont{euler}{U}{eur}{m}{n}
\DeclareMathSymbol \euetasi \mathalpha {euler} {"20}
\DeclareMathSymbol \euetahi \mathalpha {euler} {"1E}
\DeclareMathSymbol \eualfa \mathalpha {euler} {"0B}
\DeclareMathSymbol \eubeta \mathalpha {euler} {"0C}
\DeclareMathSymbol \eueta \mathalpha {euler} {"11}
\DeclareMathSymbol \euomega \mathalpha {euler} {"21}

\newcommand{\set}[1]{\left\{#1\right\}}

\newcommand{\reals}{\mathbb{R}}
\newcommand{\posreals}{\reals_+}
\newcommand{\nnreals}{\overline{\reals}_+}
\newcommand{\nats}{\mathbb{N}}

\newcommand{\veczero}{\boldsymbol{0}}
\newcommand{\vecbeta}{\boldsymbol{\beta}}
\newcommand{\veca}{\boldsymbol{a}}

\newcommand{\vecu}{\boldsymbol{u}}
\newcommand{\vecv}{\boldsymbol{v}}
\newcommand{\vecx}{\boldsymbol{x}}

\newcommand{\vecA}{\boldsymbol{A}}
\newcommand{\vecE}{\boldsymbol{E}}
\newcommand{\vecH}{\boldsymbol{H}}
\newcommand{\vecJ}{\boldsymbol{J}}
\newcommand{\vecL}{\boldsymbol{L}}
\newcommand{\vecU}{\boldsymbol{U}}
\newcommand{\vecX}{\boldsymbol{X}}
\newcommand{\vecY}{\boldsymbol{Y}}

\newcommand{\vecDelta}{\boldsymbol{\Delta}}
\newcommand{\vecb}{\boldsymbol{b}}
\newcommand{\z}[2]{z_{#1}^{(#2)}}
\newcommand{\zhat}[2]{\widehat{z}_{#1}^{(#2)}}
\newcommand{\psip}[1]{\psi^{(#1)}}

\newcommand{\calD}{\mathcal{D}}
\newcommand{\calE}{\mathcal{E}}
\newcommand{\calG}{\mathcal{G}}

\newcommand{\calI}{\mathcal{I}}

\newcommand{\calL}{\mathcal{L}}
\newcommand{\calM}{\mathcal{M}}

\newcommand{\calR}{\mathcal{R}}
\newcommand{\calS}{\mathcal{S}}

\newcommand{\E}{\mathbf{E}}
\newcommand{\hatG}{\widehat{\calG}}
\newcommand{\rhohat}{\widehat{\rho}}
\newcommand{\euf}{\eu{f}}
\newcommand{\eug}{\eu{g}}
\newcommand{\tP}{\textrm{P}}
\newcommand{\de}{{\textrm{def}}}
\newcommand{\du}{{\textrm{dup}}}
\newcommand{\bad}{{\textrm{bad}}}

\newcommand{\half}{\tfrac12}
\newcommand{\aas}{\textsf{a.a.s.}\xspace}
\newcommand{\dotprod}{\boldsymbol{\cdot}}
\newcommand{\kfr}{\nicefrac{1}{k}}

\newcommand{\deriv}[2][f]{\frac{\partial #1}{\partial #2}}

\newcommand{\offdiag}[3][f]{\frac{\partial^2 #1}{\partial #2\partial #3}}

\newcommand{\onept}{\hspace{1pt}}
\newcommand{\halfpt}{\hspace{0.5pt}}
\newcommand{\xqedhere}[2]{%
  \rlap{\hbox to#1{\hfil\llap{\ensuremath{#2}}}}}

\newcommand{\brac}[1]{\left(#1\right)}

\newcommand{\beq}[1]{\begin{equation}\label{#1}}
\newcommand{\eeq}{\end{equation}}

\begin{document}

\title{On the chromatic number of a random hypergraph}
\author{Martin Dyer%
\thanks{%
School of Computing,
University of Leeds, Leeds LS2 9JT, UK
({\tt m.e.dyer@leeds.ac.uk}).
Supported by EPSRC Research Grant EP/I012087/1.}%
\and
Alan Frieze
\thanks{%
Department of Mathematics, Carnegie Mellon University,
Pittsburgh PA15213, USA
({\tt alan@random.math.cmu.edu}).
Partially supported by NSF Grant ccf1013110.}%
\and
Catherine Greenhill%
\thanks{%
School of Mathematics and Statistics,
UNSW Australia,
Sydney NSW 2052, Australia  ({\tt c.greenhill@unsw.edu.au}).
Research supported by the Australian Research Council grant DP120100197  and performed during the
author's sabbatical at Durham University, UK.}%
}%

\date{\today}

\maketitle

\begin{abstract}
We consider the problem of $k$-colouring a random $r$-uniform hypergraph with $n$ vertices and $cn$ edges,
where $k,\,r,\,c$ remain constant as $n\to\infty$. Achlioptas and Naor showed that the chromatic number of
a random graph in this setting, the case $r=2$, must have one of two easily computable values as $n\to\infty$.
We give a complete generalisation of this result to random uniform hypergraphs.
\end{abstract}

\allowdisplaybreaks
\section{Introduction}\label{sec:intro}
We study the problem of $k$-colouring a random $r$-uniform hypergraph with $n$ vertices and $cn$ edges, where $k,\,r$ and
$c$ are considered to be constant as $n\to\infty$. We generalise a theorem of Achlioptas and Naor~\cite{AchNao05} for
$k$-colouring a random graph ($2$-uniform hypergraph) on $n$ vertices.

Their theorem specifies the two possible values for the chromatic
number of the random graph as $n\to\infty$. We give a complete
generalisation of the result of~\cite{AchNao05}. We broadly follow the
approach of Achlioptas and Naor~\cite{AchNao05}, although they rely on
simplifications which are available only in the case $r=2$. We show
that these simplifications can be replaced by more general techniques,
valid for all $k,\,r\geq 2$ except $k=r=2$.

There is an extensive literature on this problem in the case $r=2$,
colouring random graphs. In the setting we consider here, this culminates
with the results of Achlioptas and Naor~\cite{AchNao05}, though these do
not give a complete answer to the problem. Our results here include those
of~\cite{AchNao05}.

There is also a literature for the case $k=2$, random hypergraph
$2$-colouring. Achlioptas, Kim, Krivelevich and Tetali~\cite{AcKiKT02} gave a
constructive approach, but their results were substantially improved by
Achlioptas and Moore~\cite{AchMoo06}, using non-constructive methods. The
results of~\cite{AchMoo06} are asymptotic in~$r$. Our results here include
those of~\cite{AchMoo06}, but we also give a non-asymptotic treatment.
Recently, Coja-Oghlan and Zdeborov{\'a}~\cite{CojZde12} have given a small
qualitative improvement of the result of~\cite{AchMoo06}, which goes beyond
what can be proved here. See these papers, and their references, for further
information.

Finally, we note that Krivelevich and
Sudakov~\cite{KriSud98} studied a  wide range of random hypergraph colouring
problems, and some of their results  were recently improved by Kupavskii and
Shabanov~\cite{KupSha11}. But, in the setting of this paper, these results are
much less precise than those we establish here.

\begin{remark}
After preparing this paper, we learnt
of related work by Coja-Oghlan and his coauthors in the case $r=2$.
Coja-Oghlan and Vilenchik~\cite{CojVil13}
improved the upper bound on the $k$-colourability threshold,
restricting the sharp threshold for $k$-colourability to an
interval of constant width, compared with logarithmic width
in~\cite{AchNao05}. (See also Remark~\ref{chrom:rem004} below.)
A small improvement in the lower bound was obtained by Coja-Oghlan~\cite{Coj13}.
Additionally, Coja-Oghlan, Efthymiou and Hetterich~\cite{CojEftHet13}
adapted the methods from~\cite{CojVil13} to study $k$-colourability
of random regular graphs.
\end{remark}

\subsection{Hypergraphs}\label{sec:hyper}

Let $[n]=\{ 1,2,\ldots, n\}$.  Unless otherwise stated,
the asymptotic results in this paper are as $n\to\infty$.
Consider the set $\Omega(n,r,m)$ of
$r$-uniform hypergraphs on the vertex set $[n]$ with $m$ edges.  Such a
hypergraph is defined by its edge set $\calE$, which consists of $m$ distinct
$r$-subsets of $n$. Let $N=\binom{n}{r}$ denote the total number of
$r$-subsets.

Now let $\calG(n,r,m)$ denote the uniform model of a random $r$-regular hypergraph
with $m$ edges.  So $\calG(n,r,m)$ consists of
the set $\Omega(n,r,m)$ equipped with the uniform
probability distribution.  We write $G\in\calG(n,r,m)$ for a random hypergraph
chosen uniformly from $\Omega(n,r,m)$.  The edge set $\calE$ of this random
hypergraph may be viewed as a sample of size
$m$ chosen uniformly, without replacement, from the set of $N$ possible edges.

Although our main focus is the uniform model $\calG$, it is
simpler for many calculations to work with an alternative model.
Let $\Omega^*(n,r,m)$
denote the set of all $r$-uniform multi-hypergraphs
on $[n]$, defined as follows: each element of $\Omega^*(n,r,m)$ consists of vertex set $[n]$ and a multiset of edges, where each edge is now a multiset
of $r$ vertices (not necessarily distinct). We can generate a random element of $G$ of $\Omega^*(n,r,m)$ using the following simple procedure:  choose
$\vecv=(v_1,v_2,\ldots,v_{rm})\in [n]^{rm}$ uniformly at random and
let the edge multiset of $G$ be $\{ e_1,\ldots, e_m\}$, where
$e_i=\{v_{r(i-1)+1},\ldots,v_{ri}\}$ for $i\in[m]$.
Let $\calG^*(n,r,m)$ denote the probability space on $\Omega^*(n,r,m)$
which arises from this procedure, and write
$G\in\calG^*(n,r,m)$ for a hypergraph $G$ generated in this fashion.

Observe that an element $G\in\Omega^*(n,r,m)$ may not satisfy the definition of $r$-uniform
hypergraph given above, for two reasons. First, an edge of $G$ may contain repeated
vertices, which $\Omega(n,r,m)$ does not permit. We call such an edge
\emph{defective}. Second, an edge of $G\in\Omega^*(n,r,m)$ may be identical to some other edge, which again $\Omega(n,r,m)$ does not permit. We call such an edge a \emph{duplicate}.

Say an edge is \emph{bad} if it is a defective or duplicate edge.
Note that $\calG^*(n,r,m)$ is not the uniform
probability space over $\Omega^*(n,r,m)$, but that all $G$ without bad edges
are equiprobable. Thus $\calG^*(n,r,m)$, conditional
on there being no bad edges, is identical to $\calG(n,r,m)$.

If $E_n$ is a sequence of events, we say that $E_n$ occurs ``asymptotically
almost surely'' (\aas) if $\Pr(E_n)\to 1$ as $n\to\infty$. In this paper, the
event $E_n$ usually concerns $G\in\calG^*(n,r,m)$, where $m(n)=\lfloor
cn\rfloor$, for some constant $c$.
The difference between $cn$ and $\lfloor cn\rfloor$ is usually negligible, and we follow~\cite{AchNao05} in disregarding it unless the distinction is important. Thus we will write the model simply as $G\in\calG^*(n,r,cn)$, and similarly for the other models we consider.

\begin{lemma}\label{nobad}
Let $c$ be a positive constant.
For $G\in\calG^*(n,r,cn)$,
\[
  \Pr\big(G\mbox{ has no bad edge}\big)\,\sim\,
  \begin{cases}\ e^{-c(c+1)}&\mbox{if }r=2,\\ \ e^{-cr(r-1)/2}&\mbox{if }r>2.\end{cases}
\]
Furthermore, for $G\in\calG^*(n,r,cn)$, \aas\  $G$ has at most $2\ln n$ bad edges.
\end{lemma}

\begin{proof}
Throughout this proof, all probabilities are calculated in $\calG^*(n,r,cn)$.
For any edge $e\in\calE$,
\begin{equation*}
\Pr(e\mbox{ is defective})=\,1-\frac{n(n-1)\cdots(n-r+1)}{n^r}
\,=\,1-\exp\left(-\frac{r(r-1)}{2n}+O\Big(\frac{1}{n^2}\Big)\right)\,
\sim\,\frac{r(r-1)}{2n}\,.
\end{equation*}
Since this is true independently for each $e\in\calE$,
we have
\begin{equation}
\Pr(\mbox{no defective edge})\,\sim\,
\exp\Big(-\frac{r(r-1)m}{2n}\Big)\, \sim \, e^{-cr(r-1)/2} \label{model:eq001}
\end{equation}
as $m\sim cn$.
Next note that, conditional on there being no defective edges, $\calE$
is a uniform
sample of size $m$ chosen, with replacement, from the $N$  possible
$r$-subsets of $[n]$. Thus
\begin{align}
\Pr\big(\mbox{no duplicate edge}\mid\mbox{no defective edge}\big)\,&=\,\frac{N(N-1)\cdots(N-m+1)}{N^m}\notag\\
&\sim\,\exp\Big(-\frac{m(m-1)}{2N}\Big)\,\sim\,\begin{cases}\ e^{-c^2}&\mbox{if }r=2,\\\ 1&\mbox{if }r>2.\end{cases}
\label{model:eq002}
\end{align}
 Combining~\eqref{model:eq001} and~\eqref{model:eq002}
proves the first statement.

Now let $m_\de$ ($m_\du$, $m_\bad$, respectively) denote the number of
defective edges (duplicate edges, bad edges, respectively)
in $G\in\calG^*(n,r,m)$ (counting multiplicities).
For the second statement,
note that $m_\de$ has distribution Bin($m,p_\de$), and so
$\E[m_\de]\sim cr(r-1)/2$ as $n\to\infty$.
Hence Chernoff's bound~\cite[Corollary 2.4]{JaLuRu00} gives,
for large enough $n$,
\begin{equation}\label{model:eq001a}
\Pr(m_\de\geq\ln n) \,\leq\,e^{-\ln n}\,=\,1/n.
\end{equation}
Therefore \aas $G\in\calG^*(n,r,cn)$ has at most $\ln n$ defective edges.
Next, note that each edge in $\calE$ has at most $(m-1)/n^r$ duplicates
in expectation, and so
\[ \E[m_\du]\leq \frac{m(m-1)}{n^r}\leq c^2\]
for large $n$.
(Indeed, if $r > 2$ then $\E[m_\du]=o(1)$,
but we do not exploit this.)
Thus, using Markov's inequality~\cite[(1.3)]{JaLuRu00},
\begin{equation}\label{model:eq002a}
  \Pr(m_\du\geq \ln n)\, \leq\, \frac{c^2}{\ln n}.
\end{equation}
Combining this with \eqref{model:eq001a} proves the
second statement, since $m_\bad\leq m_\de + m_\du$.
\end{proof}

As already stated, conditional on there being no bad edges,
$\calG^*(n,r,cn)$ is identical to $\calG(n,r,cn)$.
By the first statement of Lemma~\ref{nobad},
$G$ has no bad edges with probability $\Omega(1)$ as $n\to\infty$.
This implies that any event occurring \aas in $\calG^*(n,r,cn)$
occurs \aas in $\calG(n,r,cn)$.
In Lemma~\ref{b''impliesb'} we use the second statement of
Lemma~\ref{nobad} to show that $\calG(n,r,cn)$ and $\calG^*(n,r,cn)$
are essentially equivalent, for our purposes.

We also make use of the following simple property of $\calG^*(n,r,m)$.
A vertex $i\in[n]$ of $G\in\Omega^*(n,r,m)$ is \emph{isolated} if it appears in no edge.
Note that a vertex is isolated if and only if it is absent from
the vector $\vecv\in [n]^{rm}$ defined above.
The following simply restates this property.

\begin{observation}\label{model:obs001}
For $S\subseteq[n]$, let $\calI_S$ be the event that all vertices in $S$ are
isolated in $G\in\calG^*(n,r,m)$. Let $G'$ be $G$ conditional on $\calI_S$
and let $G''$ be obtained from $G'$ by deleting all vertices in $S$
and relabelling the remaining vertices by $[n-|S|]$, respecting the original
ordering.
Then $G''\in\calG^*(n-|S|,r,m)$.
\end{observation}

We show, in the proof of Lemma~\ref{app:lem007}, that $G\in\calG^*(n,r,cn)$
has $\Omega(n)$ isolated vertices \aas and hence $G$ has many disconnected
components.

A further model of random hypergraphs is often used, which we will denote by
$\hatG(n,r,p)$. In this, the edge set $\calE$ of $G$ is chosen by Bernoulli
sampling. Each of the $N$ possible $r$-subsets of $[n]$ is included in $\calE$ independently with probability $p$. Essentially, this is $\calG(n,r,m)$
where $m$ is a binomial random variable Bin($N,p$). We show in
Section~\ref{sec:colouring} below that $\hatG(n,r,cn/N)$ and $\calG(n,r,cn)$
are equivalent for our problem.

\subsection{Hypergraph colouring}\label{sec:colouring}

Let $\nats$ denote the set of positive integers and define
$\nats_0=\nats\cup\set{0}$. A function $\sigma: [n]\to[k]$ is called a
$k$-partition of $[n]$, the blocks of the partition being the sets
$\sigma^{-1}(i)$, with sizes $n_i=|\sigma^{-1}(i)|$ ($i\in[k]$). Let $\Pi_k$
denote the set of $k$-partitions of $[n]$, so $|\Pi_k|=k^n$.
%
A $k$-colouring of a hypergraph $H=([n],\calE)$ is a $k$-partition $\sigma$
such that for each edge $e\in\calE$, the set $\sigma(e)$ satisfies
$|\sigma(e)|>1$. (We use the notation $H$ for fixed hypergraphs and $G$ for
random hypergraphs.) We say an edge $e\in\calE$ is \emph{monochromatic} in
$\sigma$ if $|\sigma(e)|=1$, so a $k$-partition is a colouring if no edge is
monochromatic. The \emph{chromatic number} $\chi(H)$ is the smallest $k$ such
that there exists a $k$-colouring of $H$.

Note that what we study here is sometimes called the \emph{weak} chromatic
number of the hypergraph. The \emph{strong} chromatic number is defined
similarly in terms of strong colourings, which are $k$-partitions $\sigma$
such that $|\sigma(e)|=|e|$ for each edge $e\in\calE$. Even more general
notions of colouring may be defined. See, for example,~\cite{KriSud98}. We
will not consider this further here, though it seems probable that the methods
we use would be applicable.

The principal objective of the paper will be to prove the following result.
\begin{theorem}\label{thm:chrom001}
Define $u_{r,k}=k^{r-1}\ln k$ for integers $r\geq 2$ and $k\geq 1$.
Suppose that $r\geq 2$, $k\geq 1$, and let $c$ be a positive constant.
Then for $G\in \calG(n,r,cn)$,
\begin{enumerate}[topsep=0pt,itemsep=6pt,label=(\alph*)]
\addtolength{\itemsep}{-0.5\baselineskip}
  \item If $c\geq u_{r,k}$ then \aas  $\chi(G)>k$. \label{thm:chrom001:a}
  \item If $k\geq 2$ and 
           $\max\{r,k\}\geq 3$
  then there exists a constant $c_{r,k}\in (u_{r,k-1},\,u_{r,k})$ such
that if $c<c_{r,k}$ is a positive constant
then \aas  $\chi(G)\leq k$.\label{thm:chrom001:b}
  \end{enumerate}
\end{theorem}
Now the following theorem, which is a complete generalisation of the result of \cite{AchNao05} to uniform hypergraphs, follows easily.
Note that the lower bound $u_{r,k-1}\leq c $ from Theorem~\ref{thm:chrom002}
is trivial when $k=2$, since $u_{r,1}=0$
for all $r\geq 2$.

\begin{theorem}\label{thm:chrom002}
For all $r,\,k\geq 2$, if $c\in [u_{r,k-1},\,u_{r,k})$ is a positive constant
then \aas the chromatic number of $G\in \calG(n,r,cn)$ is either $k$ or $k+1$.
Indeed, if $\max\{ r,k\} \geq 3$ and
$c\in [u_{r,k-1},\, c_{r,k})$, where $c_{r,k}$ is a constant
satisfying the conditions of~Theorem~\ref{thm:chrom001}\ref{thm:chrom001:b},
then \aas $\chi(G)=k$ for $G\in\calG(n,r,cn)$.
\end{theorem}

\begin{proof}
Let $G\in\calG(n,r,cn)$ and suppose that $u_{r,k-1} \leq  c < u_{r,k}$.
By Theorem~\ref{thm:chrom001}\ref{thm:chrom001:a}, we know that $\chi(G)\geq k$
\aas, and by Theorem~\ref{thm:chrom001}\ref{thm:chrom001:b} we know that
$\chi(G) \leq k+1$ \aas, since $c < u_{r,k} < c_{r,k+1}$.  This proves the first statement.
Furthermore, if $\max\{ r,k\} \geq 3$ and $c < c_{r,k}$ then
$\chi(G)\leq k$ \aas, by Theorem~\ref{thm:chrom001}\ref{thm:chrom001:b},
proving the final statement.
\end{proof}

For all but a few small values of $(r,k)$ we will see that $c_{r,k}$
is much closer to $u_{r,k}$ than to $u_{r,k-1}$, so that for most
values of $c$, the chromatic number of $G\in \calG(n,r,cn)$
is \aas uniquely determined.
For more detail see Remark~\ref{chrom:rem004}.

Part~\ref{thm:chrom001:a} of Theorem~\ref{thm:chrom001} is easy, and is
proved in Lemma~\ref{moment:lem001}.
As in~\cite{AchNao05},
part~\ref{thm:chrom001:b} will be proved using the second moment method~\cite[p.54]{JaLuRu00}.
If $Z$ is a random variable defined on $\nats_0$, this method applies the
inequalities
\begin{equation}\label{intro:eqn001}
  \frac{\E[\halfpt Z\halfpt]^2}{\E[Z^2]}\,\leq\,\Pr(Z>0)\,\leq\,\E[Z]\,.
\end{equation}
Although based on a rather simple idea, the second moment method is often very laborious to apply, and our analysis will be no exception.

A $k$-partition
is called \emph{balanced} if $\lfloor n/k\rfloor\leq n_i\leq\lceil n/k\rceil$
for $i=1,\ldots, k$.
A \emph{balanced} $k$-colouring of a
$H$ is a balanced $k$-partition which is also a $k$-colouring of $H$.
For convenience, we will assume that $k$ divides $n$, so
in a balanced colouring, each colour class has precisely $n/k$ vertices.
Since we suppose $k$ to be constant, the
effects of this assumption are asymptotically negligible as $n\to\infty$.
(This is proved in Lemma~\ref{b''impliesb'} below.)
Following~\cite{AchNao05},  our analysis will be carried out mainly in
terms of balanced colourings.  Indeed, we will apply~\eqref{intro:eqn001}
to the random variable $Z$ which is the number of balanced $k$-colourings
(defined formally in Section~\ref{sec:first moment}).

Clearly, if $Z>0$ then a $k$-colouring exists.
However, the analysis in Section~\ref{sec:moments}
will only allow us to conclude that $c<c_{r,k}$ implies
that $\liminf_{n\to\infty}\Pr(Z>0)>0$. Thus, we first prove a weaker statement
about $G\in\calG(n,r,cn)$:
\textit{\begin{enumerate}[label=(b$'$)]
  \item If $r,\,k\geq 2$ then there exists a constant $c_{r,k}\in (u_{r,k-1},u_{r,k})$
such that if $c<c_{r,k}$ is a positive constant
then \[ \liminf_{n\to\infty}\Pr(\chi(G)\leq k)>0.\]
\label{thm:chrom001:b'}
\end{enumerate}}
Then part~\ref{thm:chrom001:b} of Theorem~\ref{thm:chrom001} will follow from
the fact that there is a \emph{sharp threshold} for $k$-colourability of a
random hypergraph (see Lemma~\ref{b'impliesb}, below). Achlioptas and Naor~\cite{AchNao05} used a result
of~Achlioptas and Friedgut~\cite{AchFri99} which established that random
graph $k$-colourability has a sharp threshold. We will use instead the
following, more general, result.

Hatami and Molloy~\cite{HatMol08} studied the problem of the existence of a
homomorphism from a random hypergraph to a fixed hypergraph $H$.
They used the Bernoulli random hypergraph model $\hatG(n,r,p)$,
defined at the end of Section~\ref{sec:hyper}.

Given a fixed hypergraph $H=([\nu],\calE_{H})\in\Omega^*(\nu,r,\mu)$,
Hatami and Molloy
considered the threshold $p$ for the existence of a homomorphism from
$G=([n],\calE_{G})\in \hatG(n,r,p)$ to $H$. A homomorphism from $G$ to $H$ is
a function $\sigma:[n]\to [\nu]$ such that
$\sigma(e)\in \calE_{H}$ for all $e\in \calE_{G}$. If $H'$ is formed from $H$ by deleting duplicate
edges then the homomorphisms from $G$ to $H'$ are identical to those from $G$
to $H$, so we may assume that $H$ has no duplicate edges. A \emph{loop} in $H$
is an edge $e\in\calE_{H}$ for which the underlying set is a singleton. A
triangle in $H$ is a sequence $(v_1,e_1,v_2,e_2,v_3,e_3)$ of distinct vertices
$v_i\in[\nu]$ and edges $e_i\in\calE_{H}$ ($i\in[3]$), such that $v_1,\,v_2\in
e_1$, $v_2,\,v_3\in e_2$ and $v_1,\,v_3\in e_3$. The following was proved
in~\cite{HatMol08} (with minor changes of notation):

\begin{theorem}[Hatami and Molloy]\label{thm:HatMol08}
Let $H$ be a connected undirected loopless $r$-uniform hypergraph with at least
one edge. Then the $H$-homomorphism problem has a sharp threshold iff
(i)  $r\geq 3$\, or\, (ii) $r = 2$ and $H$ contains a triangle.
\end{theorem}

Here a sharp threshold means that there exists a function $p(n)$ taking values
in $[0,1]$ for all sufficiently large $n$ such that, for all $0<\varepsilon<1$, $G\in\hatG(n,r,(1-\varepsilon)p)$ has a homomorphism to $H$ \aas,
and $G\in \hatG(n,r,(1+\varepsilon)p)$ has no homomorphism to $H$ \aas

\begin{observation}\label{model:obs003}
The property of having an $H$-homomorphism is a \emph{monotone decreasing} property of $G$, that is, an $H$-homomorphism cannot be destroyed by deleting arbitrary edges of $G$. This fact will be used later.
\end{observation}

\begin{observation}\label{model:obs004}
A random hypergraph in
$G\in \hatG(n,r,cn/\binom{n}{r})$ \aas
has $cn\left(1 + \Theta\left(n^{-1/4}\right)\right)$ edges (see \eqref{bernoulli},
and $G\in \hatG(n,r,cn/\binom{n}{r})$ is uniformly random conditioned on
the number of edges it contains.
Hence if an existence problem
has a sharp threshold (with respect to $p$)  for $\hatG(n,r,p)$
then it has a sharp threshold (with respect to $c$) for $\calG(n,r,cn)$.
In this setting, existence of a sharp threshold means that there
exists a function $c(n)=\Theta(1)$ such that, for all $0<\varepsilon<1$, $G\in\calG(n,r,(1-\varepsilon)cn)$ has a homomorphism to $H$ \aas,
and $G\in \calG(n,r,(1+\varepsilon)cn)$ has no homomorphism to $H$ \aas
\end{observation}

\begin{lemma}\label{sharp-threshold}
Suppose that $r, k\geq 2$ with $\max\{k,r\}\geq 3$, and let $c$
be a positive constant.
Then the problem of $k$-colouring $G\in\calG(n,r,cn)$ has a sharp
threshold.
\end{lemma}

\begin{proof}
Take $K=([k],\calE_{K})\in\Omega^*(k,r,\mu)$ to be such that $\calE_{K}$
contains all $r$-multisets with elements in $[k]$,
except for the $k$ possible loops. Then $\mu=\binom{k+r-1}{r}-k$.
It is easy to see that the homomorphisms from a graph $G$
to $K$
are precisely the $k$-colourings of $G$. If $r=2$ and $k\geq 3$ then $K$
contains a triangle. (We may take $v_i=i\pmod 3 + 1$
and $e_i$ to be an edge with underlying set
$[3]\setminus\{i\}$, for $i\in[3]$.)
Thus it follows from
Theorem~\ref{thm:HatMol08} that the problem of $k$-colouring
$G\in \hatG(n,r,p)$ has a sharp threshold unless
$k=r=2$. Hence, by Observation~\ref{model:obs004}, the problem of $k$-colouring
$G\in \calG(n,r,cn)$ has a sharp threshold unless $k=r=2$.
\end{proof}

In the excluded case, which is the question of whether a random graph is $2$-colourable, it is known that there is no sharp threshold (see~\cite[Corollary 7]{FlKnPi89}).

We now use Lemma~\ref{sharp-threshold} to prove the following.

\begin{lemma}\label{b'impliesb}
Suppose that $k\geq 2$ and
$\max\{r,k\}\geq 3$. Then
\ref{thm:chrom001:b'} implies~\ref{thm:chrom001:b}.
\end{lemma}

\begin{proof}
From part~\ref{thm:chrom001:b'} of Theorem~\ref{thm:chrom001}, we have a constant
$c_{r,k}\in (u_{r,k-1},u_{r,k})$ such that
for $G\in \calG(n,r,cn)$,
\[ \liminf_{n\to\infty}\Pr(\chi(G)\leq k)>0\]
whenever $c < c_{r,k}$ is a positive constant.
Then Lemma~\ref{sharp-threshold}
implies that the threshold function $c(n)$ satisfies
$\liminf_{n\to\infty} c(n)\geq c_{r,k}$.
Thus for any $c<c_{r,k}$ we have \aas
$\chi(G)\leq k$, proving part~\ref{thm:chrom001:b} of Theorem~\ref{thm:chrom001}.
\end{proof}

In fact, we will prove an even weaker statement than~\ref{thm:chrom001:b'}.
\textit{\begin{enumerate}[label=(b$''$)]
  \item If $r,\,k\geq 2$ then there exists a constant $c_{r,k}\in (u_{r,k-1},u_{r,k})$
such that for any positive constant $c < c_{r,k}$, the random hypergraph
$G\in\calG^*(kt,r,ckt)$ satisfies
$\liminf_{t\to\infty}\Pr(\chi(G)\leq k)>0$.
\label{thm:chrom001:b''}
\end{enumerate}}

Observe that, in addition to restricting $n$ to multiples of $k$, the
random hypergraph model for~\ref{thm:chrom001:b''} is different from
that used in~\ref{thm:chrom001:b'}.
We now show why~\ref{thm:chrom001:b''} is sufficient.

\begin{lemma}\label{b''impliesb'}
If $r,k\geq 2$ then
\ref{thm:chrom001:b''} implies~\ref{thm:chrom001:b'}.
\end{lemma}

\begin{proof}
Let $\tP^*(n,m)=\Pr(\chi(G)\leq k)$, where $G\in\calG^*(n,r,m)$, and let $\delta(c) = \liminf_{t\to\infty}\tP^*(kt,ckt)$. Then~\ref{thm:chrom001:b''}
is the statement that there exists a constant $c_{r,k}\in (u_{r,k-1},u_{r,k})$
such that $\delta(c)>0$ for all positive $c<c_{r,k}$.
Assume that~\ref{thm:chrom001:b''} holds.

Given $n$ and $c< c_{r,k}$, let $t=\lfloor n/k\rfloor$ and let $c'$ be such that $c < c'  < c_{r,k}$. We show in Lemma~\ref{app:lem007} that
$G\in \calG^*(n,r,cn)$ has at least $k-1$ isolated vertices \aas .
Let $I$ be a set of
$n-kt\leq k-1$ isolated vertices in $G$, chosen randomly from the
set of isolated vertices in $G$.
Form $G'$ from $G$ by deleting the set $I$ of isolated vertices
and relabelling the vertices in $G'$ with $[kt]$, respecting
the relative ordering.  By symmetry, each set of size $n-kt$
is equally likely to be the chosen set $I$.  Hence
$G'\in
\calG^*(kt,r,cn)$, by Observation~\ref{model:obs001},
since $G$ can be uniquely reconstructed from $G'$ and $I$.
So $\tP^*(n,cn) =\tP^*(kt,cn)-o(1)$.
Next, if $n\geq c'k/(c'-c)$ then $c'kt > c'(n-k) \geq cn$.
Therefore, since $k$-colourability is a monotone decreasing
property (Observation~\ref{model:obs003}), it follows that
$\tP^*(kt,cn) \geq \tP^*(kt,c'kt)$.

Finally, since $c'<c_{r,k}$,~\ref{thm:chrom001:b''} implies that
$\tP^*(kt,c'kt)> \delta(c')-o(1)$, with $\delta(c')>0$. Hence we have
\[ \tP^*(n,cn)\,\geq\, \tP^*(kt,cn)-o(1)\, \geq\, \tP^*(kt,c'kt)-o(1)\,\geq\, \delta(c')-o(1),\]
which implies that
\begin{equation}
\label{halfway}
\liminf_{n\to\infty}\tP^*(n,cn)\geq\delta(c')>0.
\end{equation}

By Lemma~\ref{nobad}, \aas
$G'\in\calG^*(n,r,c'n)$ has at most $2\ln n$ bad edges.
Denote the set of bad edges in $G'$ by $B(G')$.
Let $G'$ be a uniformly chosen element of $\Omega^*(n,r,c'n)$ with at most
$2\ln n$ bad edges,
and form the random hypergraph $\varphi(G')$
as follows:
delete $B(G')$ and a set of $(c'-c)n-|B(G')|$ randomly chosen
good edges from $G'$.
(If $n$ is sufficiently large then $2\ln n \leq (c'-c)n$,
making this procedure possible.)
The resulting hypergraph $\varphi(G')$ belongs
to $\Omega(n,r,cn)$, and, by symmetry, it is a uniformly random
element of $\Omega(n,r,cn)$.
That is, that $\varphi(G')$ has the same distribution
as $G\in \mathcal{G}(n,r,cn)$ when $G'$ is chosen uniformly
from those elements of $\Omega^*(n,r,c'n)$ with at most $2\ln n$
bad edges.

Now choose a constant $c''$ with $c' < c'' < c_{r,k}$.
Then by \eqref{halfway} applied to $c'$, we have
$\tP^*(n,c'n)\geq\delta(c'')>0$.
It follows that for $G'\in\calG^*(n,r,c'n)$,
\[ \Pr(\chi(G')\leq k \text{ and } G' \text{ has at most $2\ln n$ bad edges})
   \geq \delta(c'') - o(1).\]
By monotonicity (Observation~\ref{model:obs003}), since $\varphi(G')$ has
fewer edges than $G'$,
we conclude that
\[ \Pr(\chi(\varphi(G'))\leq k \, \mid\, G' \text{ has at most $2\ln n$ bad
       edges}) \geq \delta(c'') - o(1).\]
Hence using the second statement of Lemma~\ref{nobad},
$\Pr(\chi(G)\leq k) \geq\delta(c'')-o(1)$
for $G\in\calG(n,r,cn)$.
This shows that~\ref{thm:chrom001:b'} holds, completing the proof.
\end{proof}

The remainder of the paper will be devoted to proving
Theorem~\ref{thm:chrom001}, with part~\ref{thm:chrom001:b} weakened
to~\ref{thm:chrom001:b''}. First we obtain expressions for $\E[Z]$ and
$\E[Z^2]$ in Sections~\ref{sec:first moment} and~\ref{sec:second moment},
respectively. The expression for $\E[Z^2]$ is analysed using Laplace's method, under the assumption that
constants $c_{r,k}\in (u_{r,k-1},u_{r,k})$ exist which satisfy some
other useful conditions (see Lemma~\ref{moment:lem002}).
This is established in
Section~\ref{sec:opt}, completing the proof.  Some remarks about asymptotics
are made in Section~\ref{sec:asymptotics}.

The analysis of Section~\ref{sec:opt} will require many technical lemmas, some merely verifying inequalities. These inequalities are obvious for large $r$
and $k$ but, since $r$ and $k$ are constants, we need to establish precise
conditions under which they are true.  We relegate the proofs of most
technical lemmas to the appendix, since they complicate what are fairly
natural and straightforward arguments. Therefore, whenever we use a lemma
without proof, the proof can be found in the appendix.

To complete this section, we prove the result corresponding to
Theorem~\ref{thm:chrom002}
for the Bernoulli random hypergraph model $\hatG(n,r,p)$.
Recall that $u_{r,k} = k^{r-1}\ln k$ and $N=\binom{n}{r}$.
\begin{corollary}
Let $r\geq 2$. 
 Given a positive constant $c$, let $k(c,r)$ be the
smallest integer $k$
such that $c \leq u_{r,k}$.  (Note, $u_{r,k} > 0$ by definition.)
If $G\in\hatG(n,r,cn/N)$ then $\chi(G)\in\set{k(c,r),\, k(c,r) + 1}$ \aas
\end{corollary}
\begin{proof}
Let $G\in\hatG(n,r,cn/N)$, and let $m$ be its (random) number of edges. Then Chernoff's bound~\cite[Corollary 2.3]{JaLuRu00} gives
\begin{equation}
\label{bernoulli}
  \Pr\big(| m - cn | \geq c n^{\nicefrac34}\big)\,
        \leq\,2e^{-c\sqrt{n}/3}\,.
\end{equation}
Therefore  $cn(1-n^{-\nicefrac14})\leq m\leq cn(1+n^{-\nicefrac14})$ \aas, and hence $c'n<m<c''n$ \aas for any positive constants $c'$, $c''$ such that $c'<c<c''$.

Let $k=k(r,c)$, so $u_{r,k-1}< c \leq u_{r,k}$. Choose $c'\in(u_{r,k-1},c)$,
so $m>c'n$ \aas  Now, conditional on $m>c'n$, $c'>u_{r,k-1}$ implies
$\chi(G)\geq k$ \aas, by Theorem~\ref{thm:chrom002} and monotonicity
(Observation~\ref{model:obs003}).

Similarly, choose $c''\in(c,c_{r,k+1})$, so $m<c''n$ \aas  Then, conditional
on $m<c''n$, $c''<c_{r,k+1}$ implies $\chi(G)\leq k+1$ \aas, by
Theorem~\ref{thm:chrom002} and Observation~\ref{model:obs003}. Thus
$\chi(G)\in\set{k,\,k+1}$ \aas
\end{proof}

\begin{remark}
We have shown the equivalence of various models for our problem when
$\max\{k,r\}\geq 3$. We note that this equivalence does not hold for the case
$k=r=2$, where the non-existence of a $2$-colouring is equivalent to the
appearance of an odd cycle in a random graph. This is due to the absence of a
sharp threshold for this appearance~\cite[Corollary 7]{FlKnPi89}.
Fortunately, this has little impact on our results.
\end{remark}

\section{Moment calculations}\label{sec:moments}

\subsection{First moment}\label{sec:first moment}

\begin{lemma}\label{moment:lem001}
Let $r\geq 2$, $k\geq 1$ and recall that $u_{r,k}=k^{r-1}\ln k$. Suppose that $c\geq
u_{r,k}$ is a positive constant and let $G\in \calG^*(n,r,cn)$.  Then \aas $\chi(G)>k$.
\end{lemma}

\begin{proof}
First suppose that $k=1$. Since $c>0$, the hypergraph $G$ has at least one
edge, so $\chi(G) > 1$ with probability 1.

For the rest of the proof, assume that $k\geq 2$.
Consider any $k$-partition $\sigma\in\Pi_k$ with block sizes $n_i$
$(i\in[k])$. Given $\sigma$, a random edge $e\in\calE$ is monochromatic with
probability
\[ \sum_{i=1}^k (n_i/n)^r\geq k(1/k)^r= 1/k^{r-1},\]
using Jensen's
inequality~\cite{HaLiPo88} with the convex function $x^r$. Since the edges in
$\calE$ are chosen independently, the probability that $\sigma$ is a
$k$-colouring of $G$ is
at most $(1-1/k^{r-1})^{cn}$. Let $X$ be the number of
$k$-colourings of $G$. Using \eqref{intro:eqn001} and the fact that
$|\Pi_k|=k^n$, we conclude that $\Pr(X>0) \leq\E[X]\leq \left(k\,
(1-1/k^{r-1})^{c}\right)^n$. If $c\geq u_{r,k}$ then
$c>(k^{r-1}-\nicefrac12)\ln k$, and hence
\[k\Big(1-\frac{1}{k^{r-1}}\Big)^c\,=\,\exp\Big(\ln k +c\ln\Big(1-\frac{1}{k^{r-1}}\Big)\Big) \,
\leq \, \exp\Big(\ln k -\frac{c}{k^{r-1}-\nicefrac12}\Big)\,<\,1,\]
where we have used Lemma~\ref{app:lem004} in the penultimate inequality. It follows that
$\Pr(X>0)\to 0$
as $n\to\infty$ when $c> u_{r,k}$.
\end{proof}

\begin{remark}\label{chrom:rem001}
We have proved the slightly stronger bound $(k^{r-1}-\nicefrac12)\ln k$.
This is used in~\cite{AchMoo06}, and noted, but not used, in~\cite{AchNao05}.
Since the difference is small, we mainly use the simpler bound $k^{r-1}\ln k$.
\end{remark}

In the remainder of the paper, we will assume that $k$ divides $n$,
unless stated otherwise. Recall that $Z$ is the number of balanced colourings of $G\in \calG^*(n,r,cn)$.
Let $\Xi_k$ denote the set of all balanced $k$-partitions of $[n]$.
For any balanced partition $\sigma\in\Xi_k$ and any $e\subseteq [n]$, let $M_e(\sigma)$ be the event that $|\sigma(e)|=1$. If $e$ is an edge of $G\in \calG^*(n,r,cn)$ then clearly
$\Pr(\overline{M_e(\sigma)})=1-1/k^{r-1}$, and these events are independent
for $e\in\calE$. Thus, since $|\Xi_k|=n!/\big((n/k)!\big)^k$,
\begin{equation}\label{moment:eqn001}
  \E[Z]\,=\,\frac{n!}{\big((n/k)!\big)^k}\Big(1-\frac{1}{k^{r-1}}\Big)^{cn}\,\sim\,\frac{k^{k/2}}{(2\pi n)^{(k-1)/2}}
\Big( k\Big(1-\frac{1}{k^{r-1}}\Big)^c\Big)^{n}.
\end{equation}
We have suppressed the discretisation error $cn-\lfloor cn\rfloor$. This would apparently give an additional $O(1)$ factor in $\E[Z]$ here, and in $\E[Z^2]$ below. This is of no consequence for two reasons:
\begin{enumerate}[topsep=0pt,itemsep=0pt,label=(\roman*)]
\item We need only prove that $\liminf_{n\to\infty}\E[Z^2]/\E[Z]=\Omega(1)$, so the correction is unimportant.
\item The asymptotic value for $\E[Z^2]/\E[Z]$ we obtain is independent of $n$, so using the sequence $c_n=\lfloor cn\rfloor/n$  gives the same asymptotic approximation as that given by using $c$.
\end{enumerate}

\subsection{Second moment}\label{sec:second moment}

Using the notation of Section~\ref{sec:first moment}, let $\sigma,\tau\in\Xi_k$ be balanced partitions.
Then $\overline{M_e(\sigma)} \cap \overline{M_e(\tau)}$ is the event that the edge $e$ is not monochromatic in either $\sigma$ or $\tau$. For $i,j\in [k]$, define
\[ \ell_{ij}\,=\,|\{v\in[n]: \sigma(v)=i,\,\tau(v)=j\}|.\]
Let $\vecL$ be the $k\times k$ matrix $(\ell_{ij})$.  Then $\vecL\in\calD$, where
\[\textstyle  \calD\,=\,\{\vecL\in \nats_0^{k\times k}\,\, :\,\, \sum_{i=1}^k\ell_{ij}=\sum_{j=1}^k\ell_{ij}=n/k\}.\]
There are exactly $n!/\big(\prod_{i,j=1}^k\ell_{ij}!\big)$ pairs $\sigma,\,\tau\in \Xi_k$ which share the same matrix
$\vecL\in\calD$.

Now $\Pr(M_e(\sigma))=\Pr(M_e(\tau))=1/k^{r-1}$, and
\begin{align}
  &\Pr(M_e(\sigma)\cap M_e(\tau))\, =\,\sum_{i=1}^k\sum_{j=1}^k \Big(\frac{\ell_{ij}}{n}\Big)^r. \notag
\intertext{Thus by inclusion-exclusion,}
  \Pr(\overline{M_e(\sigma)}\cap \overline{M_e(\tau)})\,&=\,
      1-\Pr(M_e(\sigma))-\Pr(M_e(\tau))+ \Pr(M_e(\sigma)\cap M_e(\tau))\\
        &=\,1-\frac{2}{k^{r-1}}+\sum_{i=1}^k\sum_{j=1}^k\Big(\frac{\ell_{ij}}{n}\Big)^r.
         \notag
\intertext{Therefore}
  \E[Z^2]\,&= \,\sum_{\sigma,\tau\in\Xi_k}\Big(1-\frac{2}{k^{r-1}}+
\sum_{i,j=1}^k\, \Big(\frac{\ell_{ij}}{n}\Big)^r\Big)^{cn} \notag \\
&=\, \sum_{\vecL\in\calD}\frac{n!}{\prod_{i,j=1}^k \ell_{ij}!}
\left(1-\frac{2}{k^{r-1}}+\sum_{i,j=1}^k \Big(\frac{\ell_{ij}}{n}\Big)^r\right)^{cn}.
\label{secondmoment}
\end{align}

Let $\posreals=\{x\in\reals:x>0\}$ and $\nnreals=\{ x\in\reals : x\geq 0\}$. Then, for
$\vecX = (x_{ij})\in\posreals^{k\times k}$, define the functions
\begin{align}
F(\vecX)\,&=\, - \sum_{i=1}^k \sum_{j=1}^k x_{ij}\ln x_{ij}+
c\halfpt \ln\Big(1-\frac{2}{k^{r-1}}+ \sum_{i=1}^k\sum_{j=1}^k x_{ij}^r\Big)\,,
\label{moment:eqn003}\\
G(\vecX)\,&=\, \,\big(2\pi n\big)^{-(k^2-1)/2}\,
\big({\textstyle\prod_{i,j=1}^k x_{ij}}\big)^{-\nicefrac12}.\notag
\end{align}
We can extend $F$ to $\nnreals^{k\times k}$ by continuity, setting
$x\ln x = 0$ when $x=0$.

We now apply Stirling's inequality in the form
\[ p! = \sqrt{2\pi (p\wedge 1)}\, \left(\frac{p}{e}\right)^{p}\,
   \left( 1 + O\left(\frac{1}{p+1}\right) \right),\]
valid for all integers $p\geq 0$, where $p\wedge 1 = \max\{ p,1\}$.
If all the $\ell_{ij}$ are positive then the summand of \eqref{secondmoment} becomes
\begin{align}
  \frac{n!}{\prod_{i,j=1}^k\ell_{ij}!}\, &
\left(1-\frac{2}{k^{r-1}}+\sum_{i,j=1}^k\,
     \Big(\frac{\ell_{ij}}{n}\Big)^r\right)^{cn}
                                 = G(\vecL/n) \, e^{nF(\vecL/n)}\, \left(1 + O\left(\frac{1}{\min \ell_{ij} + 1}\right)\right).
\label{moment:eqn002}
\end{align}
If any of the $\ell_{ij}$ equal zero then the above expression still holds with the corresponding
argument $x_{ij}$ of $G$ replaced by $1/n$, for all such $i,j$ (and treating $n$ as fixed).

Let $\vecJ_0$ be the $k\times k$ matrix with all entries equal to $1/k^2$. Then
\begin{align}
  F(\vecJ_0)\,&=\,2\ln k +2c\ln\big(1-1/k^{r-1}\big)\,=\,\ln\Big(k\Big(1-\frac{1}{k^{r-1}}\Big)^c\Big)^2\,,\label{moment:eqn004}\\
  G(\vecJ_0)\, &=  (2\pi n)^{-(k^2-1)/2}\, k^{k^2}.
     \label{moment:eqn004b}
\end{align}
Hence the term of \eqref{secondmoment} corresponding to $\vecL = n\vecJ_0$ is
asymptotically equal to
\[\frac{k^{k^2}}{\big(2\pi n\big)^{(k^2-1)/2}}\Big(k\Big(1-\frac{1}{k^{r-1}}\Big)^c\Big)^{2n}.\]
Observe from~\eqref{moment:eqn001} that this term is smaller than $\E[Z]^2$
by a factor which is polynomial in $n$. We will find a
positive constant $c_{r,k}$
such that when $c \in (0, c_{r,k})$, the function $F(\vecX)$ has a unique
maximum at  $\vecX = \vecJ_0$.
This will allow us to
apply the following theorem of Greenhill, Janson and Ruci\'{n}ski~\cite{GrJaRu10}
to estimate $\E[Z^2]$ in the region where $c < c_{r,k}$.
(See that paper for background and definitions.)

\mathversion{lifts}

\begin{theorem}[Greenhill, Janson and Ruci{\' n}ski~\cite{GrJaRu10}]\label{thm:GrJaRu10}
Suppose the following:\vspace{-0.5\baselineskip}
  \begin{enumerate}[topsep=10pt,itemsep=6pt,label=(\roman*)]
  \addtolength{\itemsep}{-0.5\baselineskip}
  \item $\mathcal{L}\subset\reals^N$ is a lattice with rank $r$.
\item $V\subseteq\reals^N$ is the
  $r$-dimensional subspace spanned by $\mathcal{L}$.
\item $W=V+w$ is an affine subspace parallel to $V$, for some $w\in\reals^N$.
\item $K\subset\reals^N$ is a compact convex set with non-empty interior $K^\circ$.
\item $\phi:K\to\reals$ is a continuous function and the restriction of $\phi$ to $K\cap W$ has a unique
maximum at some point $x_0\in K^\circ\cap W$.
\item $\phi$ is twice continuously differentiable in a neighbourhood of $x_0$
and $H:=D^2\phi(x_0)$ is its Hessian at $x_0$.
\item $\psi:K_1\to\reals$ is a continuous function on some neighbourhood
$K_1\subseteq K$ of $x_0$ with $\psi(x_0)>0$.
\item \label{GJR:TA2ln}
For each positive integer $n$ there is a vector
$\ell_n\in \reals^N$ with $\ell_n/n\in W$,
\item \label{GJR:TA2a}
For each positive integer $n$,
there is a positive real number $b_n$, and a function\\[3pt]
\mbox{$a_n: (\mathcal{L} + \ell_n)\cap nK \to\reals$} such that,
as $n\to\infty$,
\begin{align*}
  a_n(\ell)&=O\big(b_ne^{n\phi(\ell /n)+o(n)}\big), &&
              \ell\in (\mathcal{L}+ \ell_n)\cap nK,\\\mbox{and}\quad
  a_n(\ell)&=b_n\big(\psi(\ell /n)+o(1)\big)e^{n\phi(\ell /n)},
    && \ell\in (\mathcal{L} + \ell_n) \cap nK_1,
\notag
\end{align*}
uniformly for $\ell $ in the indicated sets.\vspace{-6pt}
\end{enumerate}
Then provided $\det(-H|_V)\neq0$, as $n\to\infty$,
\begin{equation*} \sum_{\ell\in (\mathcal{L}+\ell_n)\cap nK} a_n(\ell )
\sim \frac{(2\pi)^{r/2}\psi(x_0)} {\operatorname{det}(\mathcal{L})
   \det(-H|_V)^{1/2}} b_{n} n^{r/2} e^{n\phi(x_0)}. \xqedhere{102pt}{\qed}
\end{equation*}
\end{theorem}

As remarked in~\cite{GrJaRu10}, the asymptotic approximation given by this theorem remains valid for $n\in I$, where $I\subset\nats$ is infinite, provided~\ref{GJR:TA2ln} and~\ref{GJR:TA2a} hold for all $n\in I$. The conclusion of the theorem then holds for $n\in I$ as $n\to\infty$.

\mathversion{normal}

We will use this observation with $\eu{I}=\set{kt:t\in\nats}$, since we require only the weaker statement~\ref{thm:chrom001:b''} in Theorem~\ref{thm:chrom001}.

We must relate the quantities in Theorem~\ref{thm:GrJaRu10} to our notation and analysis.
We let $\eu{n}$ be $n$, restricted to positive integers divisible by $k$.
Denote by $\reals^{k\times k}$ the set of real $k\times k$ matrices, which we will view as $k^2$-vectors in the space
$\reals^{k^2}$.
Then $\eu{N}=k^2$ in Theorem~\ref{thm:GrJaRu10}.
Next, $\eu{V}$ in Theorem~\ref{thm:GrJaRu10} will be the subspace
$\calM$ of
$\reals^{k\times k}$ containing all matrices $\vecX$ such that all row and column
sums are zero, i.e.
\[ \textstyle\sum_{i=1}^k x_{ij}\,=\,\sum_{i=1}^k x_{ji}\,=\,0\qquad(j\in[k])\,,\]
and the affine subspace $\eu{W}$ will consist of the matrices $\vecX$
such that all row and column sums are $1/k$, i.e.
\[ \textstyle\sum_{i=1}^k x_{ij}\,=\,\sum_{i=1}^k x_{ji}\,=\,\kfr                \qquad(j\in[k])\,.\]
The point $\eu{w}\in\eu{W}$ will be $\vecJ_0$.

The lattice $\eu{\calL}$  in Theorem~\ref{thm:GrJaRu10}
will be the set of integer matrices in $\calM$: that is,
the set of all $k\times k$ integer matrices $\vecL = (\ell_{ij})$ such that
\[ \textstyle\sum_{i=1}^k \ell_{ij}\,=\,\sum_{i=1}^k \ell_{ji}\,=\,0\qquad(j\in[k]).\,\]
Let
$\eu{\ell}_n$ equal the $k\times k$ diagonal matrix with all diagonal entries equal to $n/k$.
Then $\eu{\ell}_n/n\in\eu{W}$
and $\eu{\ell}_n$ is an integer matrix, since we assume that $n$ is divisible by $k$.

The compact convex set $\eu{K}$ will be the subset of $\reals^{k\times k}$ such that
$0\leq x_{ij}\leq \kfr$ $(i,j\in[k])$,
which has non-empty interior $\eu{K_0}=\{x_{ij}: 0< x_{ij}<\kfr\}$.
Define $\eu{a_n}(\vecL)$ to be the summand of \eqref{secondmoment}; that is,
\[ \eu{a_n}(\vecL) = \frac{n!}{\prod_{i,j=1}^k \ell_{ij}!}
\left(1-\frac{2}{k^{r-1}}+\sum_{i,j=1}^k \Big(\frac{\ell_{ij}}{n}\Big)^r\right)^{cn}.
\]
We wish to calculate  $\E[Z^2]$, which by \eqref{secondmoment} equals
\beq{AMF3}
\sum_{\vecL\in (\eu{\calL} + \eu{\ell}_n)\cap n K} \eu{a_n}(\vecL).
\eeq
In Section~\ref{sec:opt} we will prove the following result.

\begin{lemma}
\label{moment:lem002}
Recall that $u_{r,k} = k^{r-1}\ln k$ for $r\geq 2, k\geq 1$.
Now fix $r,k\geq 2$. There exists a positive constant $c_{r,k}\in (u_{r,k-1},\, u_{r,k})$
which satisfies
\[
  \quad c_{r,k} \leq \frac{(k^{r-1}-1)^2}{r(r-1)},\]
such that $F$ has a unique maximum in $\eu{K\cap W}$
at the point $\vecJ_0\in \eu{K_0\cap W}$ whenever $c\in (0,c_{r,k})$.
\end{lemma}

Throughout this section we assume that Lemma~\ref{moment:lem002} holds.
Then $\vecJ_0$ is the unique maximum
of $F$ within $\eu{K\cap W}$, so we set $\euetahi:=F$ and $\eu{x_0}:=\vecJ_0$.
Note that $F$ is analytic in a neighbourhood of $\vecJ_0$.

Let $\eu{K_1}$ be any neighbourhood of $\vecJ_0$
whose closure is contained within $\eu{K_0}$.
The function $\euetasi$ in Theorem~\ref{thm:GrJaRu10} will be defined by
$\euetasi(\vecX) = \prod_{i,j=1}^k x_{ij}^{-\nicefrac12}$.
So $\euetasi$ is positive and analytic on $\eu{K_1}$.
We let $\eu{b_n}$  equal $\big(2\pi n\big)^{-(k^2-1)/2}$.
By \eqref{moment:eqn002},
the quality of approximations required by~\ref{GJR:TA2a} of
Theorem~\ref{thm:GrJaRu10} hold.  To see this, observe that
the relative error in~\eqref{moment:eqn002} is always $O(1)$
and that $G(\vecL/n) = e^{o(n)}$ for all $\vecL\in (\calL + \eu{\ell}_n)\cap n\eu{K}$.
This proves the first statement in~\ref{GJR:TA2a}.
However, if $\vecL\in (\calL + \eu{\ell}_n)\cap n\eu{K_1}$
then all $\ell_{ij} = \Theta(n)$ and hence
the relative error in~\eqref{moment:eqn002} is $1 + O(1/n)$.
Since then $\euetasi(\vecL/n)(1 + O(1/n)) =  \euetasi(\vecL/n) + o(1)$,
the second statement in~\ref{GJR:TA2a} holds.

Next, observe that $\calL$ and $\calM$
respectively have rank and
dimension $(k-1)^2$, since we may specify $\ell_{ij}$ or
$x_{ij}$ ($i,j\in[k-1]$) arbitrarily,
and then all $\ell_{ij}$ or $x_{ij}$ ($i,j\in[k]$) are determined.
Thus $\eu{r}=(k-1)^2$ in Theorem~\ref{thm:GrJaRu10}.

We now calculate the determinants required in Theorem~\ref{thm:GrJaRu10}.
Let $H$ be the Hessian of $F$ at the point $\vecJ_0$.
This matrix can be regarded as a quadratic form on $\reals^{k\times k}$.
In Theorem~\ref{thm:GrJaRu10} we need the determinant of $-H|_{\calM}$,
which denotes the
quadratic form $-H$ restricted to the subspace $\mathcal{M}$ of
$\reals^{k\times k}$.
This can be calculated by
\begin{equation}\label{detBV}
 \det(-H|_{\calM})= \frac{\det\, U^T(-H)U}{\det\,U^TU}
\end{equation}
for any $k^2\times (k-1)^2$ matrix $U$ whose columns form a basis of $\mathcal{M}$.

\begin{lemma}
\label{est:lem001}
Suppose that $r, k\geq 2$ and $0 < c < c_{r,k}$, where
$c_{r,k}$ satisfies Lemma~\ref{moment:lem002}.
Then the determinant of $\calL$ is $\det\calL=k^{k-1}$
and the determinant of $-H|_{\calM}$ is
$(k^2\eualfa)^{(k-1)^2}$, where
\[ \eualfa=1- \frac{c\halfpt r(r-1)}{(k^{r-1}-1)^2}.\]
\end{lemma}

\begin{proof}
Let $\delta_{ij}$ be the Kronecker delta, and define
the matrices  $\vecE_{ij}$ by $\big(\vecE_{ij}\big)_{i'j'}=
\delta_{ii'}\delta_{jj'}$.
Then $\{ \vecE_{ij} : i,j\in [k]\}$
forms a basis for $\reals^{k\times k}$.
Let $\vecE_{i*}$ be such that $\big(\vecE_{i*}\big)_{i'j'}=\delta_{ii'}$, and $\vecE_{* j}$ be such that
$\big(\vecE_{* j}\big)_{i'j'}=\delta_{jj'}$.
Then $\calM$ is the subspace of $\reals^{k\times k}$ which is orthogonal to
$\{ \vecE_{i*}$,\, $\vecE_{* i}: i\in[k]\}$.

We claim that the vectors $\vecU_{ij}=\vecE_{ij}-\vecE_{ik}-\vecE_{kj}+\vecE_{kk}$ $(i,j\in[k-1])$ form a basis for $\calM$.
To show this, consider elements of $\reals^{k\times k}$ as
vectors in $\reals^{k^2}$ (under the lexicographical ordering of
the indices $(i,j)$, say).  Then for
$i'\in[k],\,i,j\in[k-1]$,
taking dot products in $\reals^{k^2}$ gives
\begin{align}
  \vecE_{i'*}\dotprod \vecU_{ij} &=
    \sum_{\ell=1}^k\sum_{\ell'=1}^k (\vecE_{i'*})_{\ell\ell'}\, (\vecU_{ij})_{\ell\ell'}\notag\\
    &= \sum_{\ell=1}^k\sum_{\ell'=1}^k\delta_{i'\ell}\big(\delta_{i\ell}
\delta_{j\ell'}-\delta_{i\ell}\delta_{k\ell'}-\delta_{k\ell}\delta_{j\ell'}+\delta_{k\ell}\delta_{k\ell'}\big)
   = \delta_{ii'}-\delta_{ii'}-\delta_{k i'}+\delta_{k i'} = 0,\label{est:eqn001}
\intertext{and, similarly,}
   \vecE_{* i'}\dotprod \vecU_{ij} &= \sum_{\ell=1}^k\sum_{\ell'=1}^k\delta_{i'\ell'}\big(\delta_{i\ell}
\delta_{j\ell'}-\delta_{i\ell}\delta_{k\ell'}-\delta_{k\ell}\delta_{j\ell'}+\delta_{k\ell}\delta_{k\ell'}\big)
   = \delta_{i'j}-\delta_{i' k}-\delta_{i'j}+\delta_{i' k} = 0.\label{est:eqn002}
\end{align}
Thus, from \eqref{est:eqn001} and \eqref{est:eqn002}, the $(k-1)^2$ vectors $\vecU_{ij}$ lie in $\calM$, so we need only show that they are linearly independent. We will do this by
computing the determinant of the corresponding
 $(k-1)^2\times (k-1)^2$ Gram matrix $M$. Let $U$ be the
$k^2\times (k-1)^2$ matrix with
columns $\vecU_{ij}$ ($i,j\in[k-1]$). Then $M=(m_{ij,i'j'}) = U^TU$, and
we calculate (taking dot products in $\reals^{k^2}$),
\begin{align*}
m_{ij,i'j'} &= \vecU_{ij}\dotprod \vecU_{i'j'}\\
 &=\sum_{\ell=1}^k\sum_{\ell'=1}^k\big(\delta_{i\ell}
\delta_{j\ell'} -\delta_{i\ell}\delta_{k\ell'} -\delta_{k\ell}\delta_{j\ell'}+\delta_{k\ell}\delta_{k\ell'}\big)
\big(\delta_{i'\ell} \delta_{j'\ell'} -\delta_{i'\ell}\delta_{k\ell'}-\delta_{k\ell}\delta_{j'\ell'}
+\delta_{k\ell}\delta_{k\ell'}\big)\\
&=\delta_{ii'} \delta_{jj'}+\delta_{ii'}+ \delta_{jj'}+1.
\end{align*}
It follows that $M$ is a $(k-1)\times(k-1)$ block matrix, with blocks of size $(k-1)\times(k-1)$, such that
\[ M=\begin{bmatrix} 2B & B &\cdots & B\\ B & 2B &\cdots & B\\
 B & B &\ddots & B\\B & B &\cdots & 2B
 \end{bmatrix},\ \ \mbox{where}\ \
  B=\begin{bmatrix} 2 & 1 &\cdots & 1\\ 1 & 2 &\cdots & 1\\
 1 & 1 &\ddots & 1\\1 & 1 &\cdots & 2
 \end{bmatrix}\,.\]
We compute the determinant of matrices of this form in Lemma~\ref{app:lem018}.
Taking $p=q=k-1$
in Lemma~\ref{app:lem018}, we have
$\det M = k^{k-1}k^{k-1}=k^{2(k-1)}$. In particular, since the determinant is nonzero, it follows that the $\vecU_{ij}$
$(i,j\in[k-1])$
give a basis for $\calM$. Also note that, after permuting its rows,
\[ \renewcommand{\arraystretch}{1.2}
U\,=\,\left[ \begin{array}{c}I_{k-1} \\\hline U' \end{array}\right],\]
where $I_{k-1}$ is the $(k-1)^2\times(k-1)^2$ identity matrix, and $U'$ is a
$(2k-1)\times(k-1)^2$ integer matrix with entries in $\set{-1,0,1}$.
Therefore
for $\vecX\in\reals^{k^2}$ and $\vecY\in\reals^{(k-1)^2}$, we have $\vecX = \vecU \vecY$ if and only
if $\vecY_{ij} = \vecX_{ij}$ for $i,j\in [k-1]$.
It follows that $\{\vecU_{ij} : i,j\in[k-1]\}$ is a basis for the lattice $\calL$
and hence the determinant of $\calL$ is $\det\calL=\sqrt{\det M}=k^{k-1}$.

We also require the determinant of $-H|_{\calM}$.
For $\vecX\in\calM$, let
\[ F_1(\vecX) = - \sum_{i=1}^k\sum_{j=1}^k x_{ij}\ln x_{ij},\qquad
  F_2(\vecX) = 1 - \frac{2}{k^{r-1}} + \sum_{i=1}^k\sum_{j=1}^k x_{ij}^r.
  \]
Then $H = \big(h_{ij,i'j'}\big)$ has entries
\[ h_{ij,i'j'} =
   \left[ \frac{\partial^2 F_1}{\partial x_{ij}\partial x_{i'j'}}
          \right]_{\vecJ_0}
      + c\left[ \frac{\partial^2 \ln F_2}{\partial x_{ij}
             \partial x_{i'j'}}\right]_{\vecJ_0}.
             \]
Now
\[
  \left[\offdiag[F_1]{x_{ij}}{x_{i'j'}}\right]_{\vecJ_0}\,  =\,
    \left[-\frac{1}{x_{ij}}\right]_{\vecJ_0}
\delta_{ii'}\delta_{jj'}\,=\, -k^2\, \delta_{ii'} \delta_{jj'}.
\]
Next,
\[ \deriv[\ln F_2]{x_{ij}}  \,=\,
     \frac{1}{F_2}\, \deriv[F_2]{x_{ij}} \quad \text{ and }
      \quad \frac{\partial F_2}{\partial x_{ij}} = rx_{ij}^{r-1}. \]
Hence
\begin{align*}
\left[\offdiag[\ln F_2]{x_{ij}}{x_{i'j'}}\right]_{\vecJ_0}\,&=\,
  \left[\frac{1}{F_2}\offdiag[F_2]{x_{ij}}{x_{i'j'}}\, \delta_{ii'}\delta_{jj'}
-\frac{1}{F_2^2}\deriv[F_2]{x_{ij}}
  \deriv[F_2]{x_{i'j'}}\right]_{\vecJ_0}\\[0.5ex]
&= \left[ \frac{r(r-1)x_{ij}^{r-2}}{F_2}\delta_{ii'}\delta_{jj'} -
    \frac{r^2 x_{ij}^{r-1} x_{i'j'}^{r-1}}{F_2^2}\right]_{\vecJ_0}\\
&=\, \frac{r(r-1)}{k^{2(r-2)}(1 - 1/k^{r-1})^2}\, \delta_{ii'}\delta_{jj'}
  -\frac{r^2}
  {k^{4(r-1)}(1-1/k^{r-1})^4}
    \\[0.5ex]
&=\,\frac{k^2 r(r-1)}{(k^{r-1}-1)^2}\, \delta_{ii'}\delta_{jj'}
-\frac{r^2}{(k^{r-1}-1)^4}.
\end{align*}
Here we have used the fact that $F_2(\vecJ_0) = (1 - 1/k^{r-1})^2$.

These calculations show that
\[ -H =\, k^2\, \eualfa I_k + \eubeta J\]
where $I_k$ is the $k^2\times k^2$ identity matrix, $J$ is the $k^2\times k^2$ matrix
with all entries equal to 1,
\[  \eualfa\,=\,1- \frac{c\halfpt r(r-1)}{(k^{r-1}-1)^2}\,,\hspace*{5mm}  \hspace*{5mm}
\eubeta\,=\,\frac{c\halfpt r^2}{(k^{r-1}-1)^4}.
\]
By~\eqref{detBV}, the determinant of $-H|_{\calM}$ equals
\[ \frac{\det U^T(-H)U}{\det U^TU} = \frac{\det U^T(k^2\eualfa I_k+\eubeta J)U}{\det U^TU} =
\frac{\det (k^2\eualfa) U^TU }{\det U^TU} =
   \frac{(k^2\eualfa)^{(k-1)^2}\det U^TU}{\det U^TU}
= (k^2\eualfa)^{(k-1)^2}. \]
Here we have used the fact that $JU=0$, which follows since
every column of $U$ is an element of $\calM$ and hence has zero sum.
This completes the proof.
\end{proof}

Note that $\euetasi(\vecJ_0) = k^{k^2}$,
while~\eqref{moment:eqn004} gives
$\euetahi(\vecJ_0) = F(\vecJ_0)=2\ln\big(k(1-1/k^{r-1})^c\big)$.
Now $\eualfa$ is positive when $c\in (0,c_{r,k})$,
using Lemma~\ref{moment:lem002}.
Hence Lemma~\ref{est:lem001} guarantees that $\det(-H|_{\calM})\neq 0$.
Therefore we can apply Theorem~\ref{thm:GrJaRu10} to~\eqref{secondmoment}, giving
\[ \E[Z^2]\,\sim\,
\frac{k^{k}}{(2\pi n)^{k-1}\halfpt \eualfa^{(k-1)^2/2}}\,
     \Big(k\Big(1-\frac{1}{k^{r-1}}\Big)^c\Big)^{2n}.\]
Thus, from \eqref{moment:eqn001}, for all $r,\,k\geq 2$ we have
\[\Pr(Z>0)\,\geq\,\frac{\,\E[\halfpt Z \halfpt]^2}{\E[Z^2]} \,\sim\,
    \eualfa^{(k-1)^2/2}\,,\]
which is a positive constant. So $\liminf_{n\to\infty}\Pr(Z>0)>0$
and we have established part~\ref{thm:chrom001:b''} of Theorem~\ref{thm:chrom001},
under the assumption that Lemma~\ref{moment:lem002} holds.

It remains to prove Lemma~\ref{moment:lem002}, which is the focus
of the next section.

\section{Optimisation}\label{sec:opt}

We now consider maximising the function $F$ in \eqref{moment:eqn003}, and develop
conditions under which
this function has a unique maximum at $\vecJ_0$.
In doing so, we will determine suitable constants $c_{r,k}$ and prove
that Lemma~\ref{moment:lem002} holds.  This will complete the proof
of Theorem~\ref{thm:chrom001}.

Our initial goal will be to reduce the maximisation of $F$ to a univariate optimisation problem.
This reduction is performed in several stages,
presented in Sections~\ref{sec:subproblem}--\ref{sec:varyrho}.
We analyse the resulting univariate problem in Sections~\ref{sec:min eta}--\ref{sec:general}.
For more detail on our optimisation strategy, see Section~\ref{s:strategy}.
Finally, we consider
a simplified asymptotic treatment of the univariate optimisation problem in Section~\ref{sec:asymptotics}.

As is common when working with convex functions, we define
$x\ln x = +\infty$ for all $x < 0$.


It will be convenient to rescale the variables, letting $\vecA = (a_{ij})$ be
the $k\times k$ matrix defined by $\vecA = k\vecX$,
so $a_{ij} = k x_{ij}$ for all $i,j\in [k]$.
Substituting into~\eqref{moment:eqn003}, we can write
\[ F(\vecX) = \ln k - \frac{1}{k}\sum_{i=1}^k\sum_{j=1}^k a_{ij}\ln a_{ij}
  + c\ln\left(1 - \frac{2}{k^{r-1}} + \frac{\rho}{k^{2r-2}}\right)
  \]
  where
\[ \rho\,=\,k^{r-2}\sum_{i=1}^k\sum_{j=1}^k a_{ij}^r.\]
Letting $z=F(\vecX)-\ln k$, we consider the optimisation problem
\begin{subequations}
\label{opt:eqn001}
\begin{align}
    \textrm{maximise}\  \ z= -\frac{1}{k}\sum_{i=1}^k \sum_{j=1}^ka_{ij}& \ln a_{ij} + c\halfpt \ln\left(1-\frac{2}{k^{r-1}}+
\frac{\rho}{k^{2r-2}}\right)
\label{opt:eqn002}\\
\textrm{subject to}\hspace*{14mm} \sum_{i=1}^k \sum_{j=1}^k a_{ij}^r &= \frac{\rho}{k^{r-2}}, \label{opt:eqn003}\\
    \sum_{j=1}^k a_{ij} &=1\qquad(i\in[k]),\label{opt:eqn004}\\
    \sum_{i=1}^k a_{ij} &=1\qquad(j\in[k]),\label{opt:eqn005}\\
    a_{ij}&\geq 0\qquad(i,j\in[k]).\label{opt:eqn006}
\end{align}
\end{subequations}
In any feasible solution to \eqref{opt:eqn001}
we have
\begin{align}
    \rho\ &=\ k^{r-2}\sum_{i=1}^k \sum_{j=1}^k a_{ij}^r\ \leq\ k^{r-2}\sum_{i=1}^k \Big(\sum_{j=1}^k a_{ij}\Big)^r\ =\ k^{r-1}
\label{opt:eqn007}\\
    \intertext{and}
    \rho\ &=\ k^{r-2}\sum_{i=1}^k \sum_{j=1}^k a_{ij}^r\ \geq\ \frac{k^{r-2}\Big(\sum_{i=1}^k
\sum_{j=1}^k a_{ij}\Big)^r}{\Big(\sum_{i=1}^k
\sum_{j=1}^k 1\Big)^{r-1}}\ =\ \frac{k^{r-2}k^r}{k^{2(r-1)}}\ =\ 1,\label{opt:eqn008}
\end{align}
where we have used H\"older's inequality~\cite{HaLiPo88} in \eqref{opt:eqn008}.
Hence the system \eqref{opt:eqn003}--\eqref{opt:eqn006} is infeasible if $\rho\not\in[1,k^{r-1}]$,
in which case we set $\max z = -\infty$. Conversely, it is easy to show that the system
\eqref{opt:eqn003}--\eqref{opt:eqn006} is feasible for all $\rho\in [1,k^{r-1}]$. A formal proof of this is given in Lemma~\ref{app:feasible}.

We wish to determine the structure of the maximising solutions in the optimisation problem \eqref{opt:eqn001}.
Following~\cite{AchNao05}, we will relax the constraints \eqref{opt:eqn005},
%
and write \eqref{opt:eqn003} as
\begin{equation}
\sum_{j=1}^k a_{ij}^r = \frac{\varrho_i}{k^{r-2}},\qquad \sum_{i=1}^k \varrho_i =\rho.\label{opt:eqn009}
\end{equation}
By the same method as for Lemma~\ref{app:feasible}, we can show that
the system given by \eqref{opt:eqn004},~\eqref{opt:eqn006} and \eqref{opt:eqn009} is feasible if and only if $\sum_{j=1}^k \varrho_i =\rho$ and $\varrho_i\in[1/k,k^{r-2}]$
for $i\in[k]$.
Note that \eqref{opt:eqn007} and \eqref{opt:eqn008} assume $a_{ij}\geq 0$, but the relaxation of \eqref{opt:eqn006}
will be unimportant. Since $z = -\infty$ whenever some $a_{ij} < 0$, these conditions must be satisfied automatically at any finite optimum.

\begin{remark}
\label{rem3}
For a fixed value of $\rho\in [1,k^{r-1}]$, we can ignore the term $c\halfpt\ln(1-2/k^{r-1}+\rho/k^{2r-2})$ in the objective function $z$ in (\ref{opt:eqn002}), leading to the
maximisation problem considered in Lemma~\ref{app:boundary}
with $\ell=k$ and $\rhohat=\rho/k^{r-2}$.
Then Lemma~\ref{app:boundary} shows that this objective function
cannot be maximised on the boundary of the region determined
by \eqref{opt:eqn004} and \eqref{opt:eqn006}, unless $\rho =k^{r-1}$. Thus we can
also omit the constraints \eqref{opt:eqn006} in the optimisation, and consider the
reduced system \eqref{opt:eqn002}--\eqref{opt:eqn004}.  
Allowing $\rho$ to vary clearly cannot introduce a local maximum 
on the boundary, since Lemma~\ref{app:boundary} holds for all 
$\rho\in[1,k^{r-1}]$.
In fact, it is not difficult to see that there can be no local maximum on the boundary even for the complete system \eqref{opt:eqn001}. This can be proved formally in a similar way to Lemma~\ref{app:boundary}. We omit the proof since we do not use this observation.
\end{remark}

\subsection{Outline of the optimisation strategy}\label{s:strategy}

Having dropped the constraints \eqref{opt:eqn005}, we can break \eqref{opt:eqn001} up into $k$ independent simpler one-row subproblems, if we specify the values
$(\varrho_1,\varrho_2,\ldots, \varrho_k)$ such that $\rho=\varrho_1+\varrho_2+\cdots+\varrho_k$. This is done in Section~\ref{sec:subproblem}. Later in Section~\ref{sec:combined} we will determine the optimal values of $\varrho_1,\varrho_2,\ldots,\varrho_k$ for a fixed value of $\rho$. Then in Section~\ref{sec:varyrho} we will allow $\rho$ to vary.

Section~\ref{sec:subproblem} reduces the one-row problem \eqref{opt:eqn010} to an essentially one-variable problem \eqref{zx:aa}. For this problem there is a unique optimum value of $\beta=\beta(\varrho)$ that is given in \eqref{amf1}. This value of $\beta$ maximises the objective function, now expressed as $-\euf(\beta)$: see \eqref{opt:eqn023}. The nonlinear constraint \eqref{opt:eqn012} will now have been replaced by an equation $\eug(\beta)=k^{2-r}\varrho-k^{1-r}$, see \eqref{opt:eqn023}.

In Section~\ref{sec:combined} we try to find values $\varrho_1,\varrho_2,\ldots,\varrho_k$ that sum to a fixed $\rho$ and and associated values $\beta_1,\beta_2,\ldots,\beta_k$ that minimise $\euf(\beta_1)+\euf(\beta_2)+\cdots+\euf(\beta_k)$. The constraints become $\eug(\beta_1)+\eug(\beta_2)+\cdots+\eug(\beta_k)=k^{2-r}(\rho-1)$, see \eqref{opt:eqn024}. We show that the $\beta_i$ take one of at most two values, $\gamma_1\leq \gamma_2$. These values are the solutions to 
$\euf'(\beta)/\eug'(\beta)=\lambda$ where $\lambda$ is a Lagrange multiplier, to be optimised over. This reduces the optimisation to \eqref{opt:eqn032}. Here $t_j$ is the number of $\beta_i$ taking the value $\gamma_j$, and $h_j=\euf(\gamma_j)-\lambda \eug(\gamma_j)$ for $j=1,2$. We relax the equation in \eqref{opt:eqn032} to an inequality and argue that $t_2=0$ in an optimal solution. At this point we can think of the optimisation as being over $\lambda$ or, equivalently, over $\beta_*=\gamma_1$. Choosing the latter we end up with the optimisation problem \eqref{opt:eqn034}.

We now have to optimise over $\rho$ and find a bound on $c$ that ensures that $a_{ij}=1/k$ optimises the relaxed problem (\ref{opt:eqn002})--(\ref{opt:eqn004}). Using the relaxation \eqref{opt:eqn034}, we see that it is sufficient to satisfy \eqref{amf3}. This leads to an inequality \eqref{opt:eqn037} for $c$. Making the right hand side of this inequality as small as possible leads to a univariate optimisation problem that is dealt with in Sections~\ref{sec:min eta}--\ref{sec:general}.

\subsection{The subproblem corresponding to one row}\label{sec:subproblem}

Now, consider any \emph{fixed} feasible values of the $\varrho_i$ $(i\in[k])$ such that
$\sum_{j=1}^k \varrho_i=\rho$. Then the problem
decomposes into $k$ independent maximisation subproblems.
In this subsection we use Lagrange multipliers to perform the optimisation on these subproblems.

As already mentioned, when the $\varrho_i$ are fixed (and hence $\rho$ is fixed), the term
$c\ln\left( 1 - 2/k^{r-1} + \rho/k^{2r-2}\right)$ in the objective function $z$ of (\ref{opt:eqn001}) is
constant.  Hence we omit this term from the optimisation problems we consider until we
once again allow $\rho$ to vary, in Section~\ref{sec:varyrho}.

We temporarily suppress the subscript $i$, to write $\veca=(a_1,a_2,\ldots,a_k)$ for the $i$th row of $A$.
For a fixed value of $\varrho$, the subproblem is then 
\begin{subequations}
\label{opt:eqn010}
\begin{align}
    \textrm{maximise}\     \z{1}{\varrho}
=-\sum_{j=1}^k a_j  & \ln a_j\label{opt:eqn011}\\
    \textrm{subject to} \hspace*{12mm}
     \sum_{j=1}^k a_j^r &= \frac{\varrho}{k^{r-2}},   \label{opt:eqn012}\\
    \sum_{j=1}^k a_j &=1,\label{opt:eqn013}\\
    a_j&\geq 0.\label{opt:eqn014}
\end{align}
\end{subequations}
We assume that $1/k\leq \varrho\leq k^{r-2}$,
so that the problem is feasible.
\begin{remark}\label{rem1}
When $\varrho = 1/k$ or $\varrho = k^{r-2}$
the optimization is trivial.
If  $\varrho = 1/k$ then there is a unique optimal solution, which satisfies
$a_j = 1/k$ for all $j\in [k]$ and gives $\z{1}{\varrho}=\ln k$. This is the value
of $\varrho$ that gives the global optimum to our problem. If $\varrho = k^{r-2}$ then there are
$k$ distinct optimal solutions, each with $a_j=1$ for exactly one value of
$j$, and $a_j=0$ otherwise, each giving $\z{1}{\varrho}=0$.
For ease of exposition, we include these cases in our argument below, though the analysis is unnecessary in these cases.
\end{remark}

\begin{remark}\label{rem2}
Applying Lemma~\ref{app:boundary} with $\ell=1$ and $\rhohat=\varrho/k^{r-2}$
shows that $\z{1}{\varrho}$ cannot be maximised at any  point which lies in the boundary of the region determined by \eqref{opt:eqn013}--\eqref{opt:eqn014}, unless $\varrho =k^{r-2}$. Thus we can omit the constraints \eqref{opt:eqn014} in the optimisation of \eqref{opt:eqn010}.
\end{remark}

Introducing the multiplier $\lambda$ for \eqref{opt:eqn012} and
$\mu$ for \eqref{opt:eqn013}, the Lagrangian is 
\begin{equation}\label{opt:eqn015}
    L(\veca,\lambda,\mu) \ =\ -\sum_{j=1}^k a_j \ln a_j + \lambda\, \Big(\sum_{j=1}^k a_j^r-
\frac{\varrho}{k^{r-2}}\Big)+\mu\Big(\sum_{j=1}^k a_j-1\Big).
\end{equation}
The maximisation of $L$
gives \eqref{opt:eqn012} and \eqref{opt:eqn013},
together with the equations
\begin{equation}\label{opt:eqn016}
    \varphi(a_j) = 0 \qquad (j\in [k]), \quad \text{ where }
  \quad \varphi(x) = \varphi_{\lambda,\mu}(x) =  -1-\ln x+\lambda rx^{r-1}+\mu.
\end{equation}
If the equation $\varphi(x)=0$  has only one root then all the $a_j$ equal
this root and hence, from \eqref{opt:eqn013}, $a_j=1/k$
for all $j\in [k]$. In this case $\varrho=1/k$.
It follows from Lemma~\ref{app:feasible} that $\varphi$ has at least one root.

Now suppose that the equation $\varphi(x)=0$ has more than one root
(that is, $1/k <  \varrho \leq  k^{r-2}$), and let $\alpha$ be the largest.
If $\veca$ satisfies $a_j\neq \alpha$ for some $j\in [k]$ then subtracting the corresponding equations in \eqref{opt:eqn016}
gives
\[
    \ln \alpha -\ln a_j-\lambda r(\alpha^{r-1}-a_j^{r-1})=0.\]
That is,
\begin{equation}
\label{lagrangian}
  \lambda=\frac{\ln \alpha -\ln a_j}{r(\alpha^{r-1}-a_j^{r-1})}>0.
\end{equation}
Hence, since $-\ln x$ and $x^{r-1}$ are both convex on $x>0$ and $\lambda$ is positive, $\varphi(x)$ is a
  strictly convex function. It follows that the equation $\varphi(x)=0$ has at most two roots in $(0,\infty)$.
Let the roots of $\varphi(x)=0$ be $\alpha$
and $\beta$, where we assume that $\alpha>\beta$. We have $a_j\in\set{\alpha,\beta}$  for all $j\in[k]$. But we still need to determine how many of the $a_j$ equal $\alpha$ and how many equal $\beta$.

Consider any stationary point $(\veca_\ast,\lambda_\ast,\mu_\ast)$ of $L$.  Then $\veca_\ast$
and $\lambda_\ast$ satisfy \eqref{lagrangian}.
Suppose without loss of
generality that for some $1\leq t\leq k-1$ we have
$a_1,\ldots,a_t=\alpha$,
$a_{t+1},\ldots,a_k=\beta$, where $\veca_\ast = (a_1,\ldots, a_k)$.
The Hessian $\vecH = \vecH_{\lambda_\ast,\mu_\ast}$ of the Lagrangian
$L_{\lambda_\ast,\mu_\ast} = L(\, \cdot\, , \lambda_\ast, \mu_\ast)$, considered as a function of $\veca$ only,
is a $k\times k$ diagonal matrix with
diagonal entries
\[
 h_{jj} = \begin{cases}
    \varphi'(\alpha)\,=\,-\frac{1}{\alpha}+\lambda r(r-1)\alpha^{r-2}
  & (j=1,\ldots, t),\\
    \varphi'(\beta)\,=\,-\frac{1}{\beta}+\lambda r(r-1)\beta^{r-2}
     & (j=t+1,\ldots,k).
   \end{cases}
\]
Since $\varphi$ is strictly convex with zeros $\beta < \alpha$, we know that
$\varphi'(\beta) < 0 < \varphi'(\alpha)$.
The quadratic form determined by the Hessian at $\veca_\ast$ is
\begin{equation}\label{opt:eqn017}
\vecx^T \vecH \halfpt\vecx\,=\,\varphi'(\alpha)\sum_{j=1}^t x_j^2+\varphi'(\beta)\sum_{j=t+1}^k x_j^2.
\end{equation}
To determine the nature of the stationary point $\veca_\ast$, we restrict the quadratic form to $\vecx$
lying in the tangent space at $\veca_\ast$.
This means that $\vecx$ satisfies linear equations determined by the gradient vectors of the constraint
functions at $\veca_\ast$.
See, for example,~\cite{Shutle95}. In our case, these equations are
\begin{align*}
    \alpha^{r-1}\sum_{j=1}^tx_j+\beta^{r-1}\sum_{j=t+1}^k x_j&=0,\\
    \sum_{j=1}^tx_j+\sum_{j=t+1}^kx_j&=0.
\end{align*}
These equations are linearly independent since $\alpha>\beta$.
Rearranging these equations allows us to express $x_1$ and
$x_k$ in terms of $x_2,\ldots, x_{k-1}$ as follows:
\[ x_1=-\sum_{j=2}^tx_j,\qquad \,x_k=-\sum_{j=t+1}^{k-1}x_j.\]
Substituting these into \eqref{opt:eqn017} gives
\begin{equation}\label{opt:eqn018}
\vecx^T \vecH \halfpt\vecx\,=\,\varphi'(\alpha)\Big(\sum_{j=2}^t x_j\Big)^2+\varphi'(\alpha)
\sum_{j=2}^t x_j^2+\varphi'(\beta)\sum_{j=t+1}^{k-1} x_j^2+\varphi'(\beta)\Big(\sum_{j=t+1}^{k-1} x_j\Big)^2.
\end{equation}
For $\veca_\ast$ to be a strict local
maximum, the right hand side of \eqref{opt:eqn018} must be negative for all $x_2$, $x_3$, \ldots, $x_{k-1}$
such that $\vecx\neq\veczero$. Since $\varphi'(\alpha)>0$, $\varphi'(\beta)<0$, this will be true if and only if $t=1$, when the terms with coefficient $\varphi'(\alpha)$ in \eqref{opt:eqn018} are absent. This local maximum is clearly unique up to the choice of $j\in[k]$ such that $a_j=\alpha$. Hence it is global, since $\z{1}{\varrho}$ is bounded on the compact region determined by \eqref{opt:eqn012}--\eqref{opt:eqn014},
and has no local maxima on the boundary when $\varrho < k^{r-2}$ (see
Remark~\ref{rem2}).
Thus there are $k$ global maxima, given by choosing $p\in[k]$ and setting
$a_p=\alpha$, $a_j=\beta$ $(j\in[k], j\neq p)$,
where $(\alpha,\beta)$
is the unique solution such that $\alpha\geq\beta\geq 0$ to
the equations
\begin{equation}\label{amf1}
  \alpha^r+(k-1)\beta^r  = k^{2-r}\varrho, \qquad
  \alpha+(k-1)\beta = 1.
\end{equation}
The fact that there is at least one solution to these equations
follows from Lemma~\ref{app:feasible}.  Next, note that
the derivative of the function
\[ \left( 1- (k-1)\beta\right)^r + (k-1)\beta^r\]
is zero at $\beta = 1/k$ and negative for
$\beta\in (0,\nicefrac{1}{k})$.  Hence there can be at most one
solution to these equations which satisfies $0\leq \beta\leq \nicefrac{1}{k}$,
or equivalently, $0\leq \beta\leq \alpha$.

Note that the relaxation of the constraints~\eqref{opt:eqn006} proves to be unimportant, since the optimised values of the $a_{ij}\in\{\alpha,\beta\}$ are positive.
Thus the optimisation \eqref{opt:eqn010} results in the system
\begin{subequations}
\label{zx:aa}
\begin{align}
  \textrm{maximise}\  \z{1}{\varrho}\,=\,-\alpha\ln \alpha  - (k&-1)\beta\ln \beta
\label{aa:eqn1}\\
   \textrm{subject to}\,\,\,  \alpha^r+(k-1)\beta^r\, &=\, k^{2-r}\varrho,\\
    \alpha+(k-1)\beta\, &=\, 1,\\
     \beta\leq \nicefrac{1}{k}.\label{aa:eqn2}
\end{align}
\end{subequations}

We have omitted the constraint $0\leq \beta$ here, but this will be
enforced in any optimal solution since $\z{1}{\varrho} = -\infty$ if $\beta < 0$.
The maximisation problem is trivial since there is only one feasible solution
which satisfies $0\leq \beta \leq \nicefrac{1}{k}$, and no other feasible
solution can be a maximum.

When $\varrho = 1/k$ we have $\alpha=\beta= 1/k$, while if $\varrho = k^{r-2}$
then $\alpha = 1$ and $\beta = 0$.  When $1/k < \varrho < k^{r-2}$ we
have $0 < \beta < 1/k < \alpha < 1$.

\subsection{The combined problem, for a fixed value of $\rho$}\label{sec:combined}

We now combine these subproblems (one for each row) to give an optimisation problem
corresponding to a \emph{fixed} value of $\rho\in [1,k^{r-1}]$, as follows:
\begin{subequations}
\label{opt:eqn019}
\begin{align}
    \textrm{maximise} \ \z{2}{\rho}
\ =\ -\sum_{i=1}^k\big(\alpha_i\ln \alpha_i +&(k-1)\beta_i\ln \beta_i\big)\label{opt:eqn020}\\[-2.5ex]
 \textrm{subject to}\ \
\sum_{i=1}^k\big(\alpha_i^r+(k-1)\beta_i^r\big)\ &=\ k^{2-r}\rho,
\label{opt:eqn021}\\
    \alpha_i+(k-1)\beta_i\ &=\ 1\qquad(i\in[k]),\label{opt:eqn022}\\
    \beta_i &\leq \nicefrac{1}{k} \qquad (i\in [k]).\notag
\end{align}
\end{subequations}

As before, the objective function ensures that $\beta_i\geq 0$ for $i\in[k]$
at any finite optimum.
Recall from Remark~\ref{rem3} that $\z{2}{\rho}$ has no maximum
on the boundary  when $\rho < k^{r-1}$.

For $\beta\in\mathbb{R}$, write
\begin{equation}\label{opt:eqn023}
\euf(\beta)\,=\,\ln k+\alpha\ln \alpha+(k-1)\beta\ln \beta\,,\quad
\eug(\beta)\,=\,\alpha^r+(k-1)\beta^r-\frac{1}{k^{r-1}}\, ,\quad
\end{equation}
where $\alpha$ is defined as $1-(k-1)\beta$ and hence $\textrm{d}\alpha/\textrm{d}\beta=-(k-1)$.
We use the notation $\euf$ and $\eug$ here, and reserve the symbols
$f$ and $g$ for transformed versions of these functions which
will be introduced in Section~\ref{sec:min eta}.

Now $\euf(\beta) = +\infty$ if $\beta < 0$ or
$\beta > 1/(k-1)$.
Also
\[ \euf(0)=\ln k, \qquad \eug(0)=1-1/k^{r-1} \quad \text{ and }
 \quad \euf(1/k)=\eug(1/k)=0. \]
Note further that
\begin{equation*}
  \euf'(\beta)\,=\,-(k-1)(\ln\alpha-\ln\beta)<0,\quad
  \eug'(\beta)\,=\,-r(k-1)(\alpha^{r-1}-\beta^{r-1})<0\quad(\beta\in[0,\kfr))\,,
\end{equation*}
so both $\euf(\beta)$ and $\eug(\beta)$ are positive and decreasing for
$\beta\in [0,\kfr)$.

Letting 
$\zhat{2}{\rho}= k\ln k-\z{2}{\rho}$, \eqref{opt:eqn019} can now be rewritten as
\begin{subequations}
\label{opt:eqn024}
\begin{align}
    \textrm{minimise}\,\,\, & \zhat{2}{\rho}(\vecbeta)\,=\, \sum_{i=1}^k \euf(\beta_i)
\label{opt:eqn025}\\
    \textrm{subject to} \,\,\,&\sum_{i=1}^k \eug(\beta_i)\ =\ k^{2-r}(\rho-1)\,,
\label{opt:eqn026}\\
  & \qquad \beta_i \ \leq 1/k \qquad (i\in [k]).\label{opt:eqn026b}
\end{align}
\end{subequations}
We remark that
\begin{equation}\label{zero}
\zhat{2}{\rho}(0)=k\ln k>0=\zhat{2}{\rho}(\kfr,\ldots,\kfr).
\end{equation}
We therefore ignore $\vecbeta=\veczero$ in our search for the minimum in \eqref{opt:eqn024}: see Remark \ref{rem1}.

We also ignore \eqref{opt:eqn026b} and
apply the Lagrangian method to \eqref{opt:eqn025} and \eqref{opt:eqn026},
using the multiplier
$-\lambda$ for \eqref{opt:eqn026}. The Lagrangian optimisation will be to minimise the function 
\begin{equation}
\label{psidef}
 \psip{\rho}(\vecbeta,\lambda) = \sum_{i=1}^k \euf(\beta_i)
 - \lambda\left(\Big(\sum_{i=1}^k \eug(\beta_i) \Big) - k^{2-r}(\rho-1)\right).
\end{equation}
The stationary points of the Lagrangian $\psip{\rho}$ are given by \eqref{opt:eqn021}, \eqref{opt:eqn022} and the equations
\begin{equation}\label{opt:eqn027}
  \frac{-1}{k-1}\halfpt\deriv[\psip{\rho}]{\beta_i}\, =\,
  \frac{-1}{k-1}\big(\euf'(\beta_i)-\lambda \eug'(\beta_i)\big)\, =\, (\ln \alpha_i -\ln
\beta_i)-\lambda r(\alpha_i^{r-1}-\beta_i^{r-1}) \, =\, 0\qquad(i\in[k])\,.
\end{equation}
We will concentrate on those stationary points of the Lagrangian that can give us the optimum for \eqref{opt:eqn024}.

Let $B=[0, \kfr]$.
We define
\[
\eueta(\beta)\,=\,\frac{\euf(\beta)}{\eug(\beta)}, \qquad
\euomega(\beta) \,=\, \frac{\euf'(\beta)}{\eug'(\beta)}
\]
for $\beta\in [0,\kfr)$, and extend by continuity to give
\[ \eueta(\kfr) = \euomega(\kfr) = \frac{k^{r-1}}{r(r-1)}.\]
Again, we reserve the notation $\eta$ and $\omega$ for transformed
versions of these functions, introduced in Section~\ref{sec:min eta}.
(The values of $\eueta(\kfr)$ and $\euomega(\kfr)$ are established in
Lemma~\ref{app:lem008}, in terms of the transformed functions.)

Now suppose that $(\vecbeta_*,\lambda_*)$ is a stationary point
of the
Lagrangian $\psip{\rho}$ which satisfies $\vecbeta_* = (\beta_{1*},\ldots, \beta_{k*})\in B^k$.
Then $\zhat{2}{\rho}(\vecbeta_*) = \psip{\rho}(\vecbeta_*,\lambda_*)$.

Suppose that there exists $i$ such that
$\veczero<\beta_{i*} < 1/k$. Then $\beta_{i*}< 1/k<\alpha_{i*}=1-(k-1)\beta_{i*}$,
and $(\beta,\lambda)= (\beta_{i*},\lambda_*)$ must
satisfy the equation
\begin{equation}\label{opt:eqn028}
\lambda\ =\euomega(\beta)\ =\ \frac{\euf'(\beta)}{\eug'(\beta)} \ =\  \frac{\ln \alpha -\ln \beta}{r(\alpha^{r-1}-\beta^{r-1})}>0,
\end{equation}
where $\alpha=1-(k-1)\beta$.
This shows that $0 < \lambda_* = \euomega(\vecbeta_*)<\infty$. Furthermore, (\ref{opt:eqn028}) implies
that $\beta_{j*}>0$ for all $j$, since $\euomega(0)=\infty$.
Thus in any stationary point $(\vecbeta_*,\lambda_*)$ of $\psip{\rho}$ with
$\veczero\neq \vecbeta_*\in B^k$, for each $i\in [k]$, either $\beta_{i*}= 1/k$
(in which case $\alpha_{i*} = 1/k$ and \eqref{opt:eqn027} holds), or
$\beta_{i*} \in (0,\kfr)$ and $(\beta_{i*},\lambda_*)$
is a solution to~\eqref{opt:eqn028}. 

We now assume that $\vecbeta_*\neq (\kfr,\ldots, \kfr)$ (see Remark \ref{rem1})
and rewrite 
 (\ref{psidef}) 
as
\begin{align}
 \psip{\rho}(\vecbeta,\lambda) &=
\lambda\, k^{2-r}(\rho-1) +
\sum_{\beta_i\neq 1/k} \eug(\beta_i)(\eueta(\beta_i)-\lambda).
\label{opt:eqn028b}
\end{align}
First suppose that $\lambda_* > \max_B \, \eueta(\beta)$.
Since
\[ \eueta'(\beta) = \frac{g'(\beta)}{g(\beta)}\,\left( \euomega(\beta) - \eueta(\beta)\right)
\]
for all $\beta\in (0,\kfr)$, it follows from (\ref{opt:eqn028}) that $\eueta'(\beta_{i*}) < 0$
for any $i\in [k]$ with $\beta_{i*}\neq\kfr$.  But this shows that $(\vecbeta_*,\lambda_*)$ is not
a local minimum of $\psi^{(\rho)}$, as we can decrease the value of $\psi^{(\rho)}$ by
increasing $\beta_{i*}$ infinitesimally, while holding all other values of $\beta_{j*}$ and $\lambda_*$
steady.  Since we wish to minimise $\psi^{(\rho)}$ (and hence $\zhat{2}{\rho}$), we now assume that
$\lambda_* \leq \max_B\, \eueta(\beta)$.

Next, suppose that
$\lambda_*<\min_B\, \eueta(\beta)$.
As $\vecbeta_* \neq (\kfr,\ldots, \kfr)$, by \eqref{opt:eqn028b} we conclude that
$$ \zhat{2}{\rho}(\vecbeta_*) = \psip{\rho}(\vecbeta_*,\lambda_*)> \lambda_*\,  k^{2-r} (\rho-1)\geq 0.$$
Hence $(\vecbeta_*,\lambda_*)$ cannot minimise $\zhat{2}{\rho}$
if $\lambda_* < \min_{\beta\in B} \eueta(\beta)$.


Therefore (see Remark \ref{rem1}) we may now assume that $(\vecbeta_*,\lambda_*)$
is a stationary point of (\ref{opt:eqn024}) with
\[ \vecbeta_*\in B^k\setminus \{ \veczero,\, (\kfr,\ldots, \kfr)\},\]
where
$\lambda_* = \euomega(\beta_{i*})$ for any $i$ such that $\beta_{i*}\neq \kfr$, and
such that $\lambda = \lambda_*$  satisfies
\begin{equation}\label{opt:eqn029}
 \min_{\beta\in B}\, \eueta(\beta)\ \leq\ \lambda\ \leq\ \max_{\beta\in B}\,    \eueta(\beta).
\end{equation}
We will prove the following in Section~\ref{sec:min eta} below.
\begin{lemma}\label{opt:lem001}
The function $\euomega(\beta)$ has a unique minimum in $(0,\kfr)$.
Furthermore, if \eqref{opt:eqn029} holds for some $\lambda > 0$ then the
equation $\euomega(\beta)=\lambda$ has at most two distinct solutions $\beta\in B$.
\end{lemma}

Now consider the case that the equation $\euomega(\beta)=\lambda$ has
exactly two distinct roots $\gamma_1 > \gamma_2$.
Define $\gamma_0=\kfr $. Let $t_i$ $(i=1,2)$ be the multiplicity of $\gamma_i$ amongst the $\beta_j$ ($j\in[k]$). We write $\euf_i$ for $\euf(\gamma_i)$ $(i=0,1,2)$, and similarly for $\eug$, $\eueta$ and $\euomega$. For $i=0,1,2$ we define $h_i=\euf_i-\lambda \eug_i$. Since $\gamma_1>\gamma_2$ and $\eug'(\beta)<0$ for $\beta\in[0,\kfr)$, we have $\eug_1<\eug_2$.
Now $\euomega$ is continuous on $B$ and has a unique minimum in $(0,\kfr)$,
by Lemma~\ref{opt:lem001}. Hence, for any $\beta$ strictly between the two
solutions $\gamma_1,\gamma_2$ of $\euomega(\beta)=\lambda$,
it follows that
$\euomega(\beta)<\lambda$. 
Therefore
\begin{equation}\label{opt:eqn030}
h_1 - h_2 \ = \ (\euf_1 - \lambda \eug_1) - (\euf_2 - \lambda \eug_2) \ = \
\int_{\gamma_2}^{\gamma_1} \big(\euf'(\beta)-\lambda \eug'(\beta)\big)\mathrm{d}\beta
 \ <\ 0.
\end{equation}
Hence $h_1<h_2$. Also,
as $\euf_0=\eug_0=0$ we have
\begin{equation}\label{opt:eqn031}
 - h_1 \ = \ (\euf_0 - \lambda \eug_0) - (\euf_1 - \lambda \eug_1) \ = \   \int_{\gamma_1}^{\kfr} \big(\euf'(\beta)-\lambda \eug'(\beta)\big)\mathrm{d}\beta
 \ >\ 0,
\end{equation}
where the final inequality holds since $\euomega(\beta)>\lambda$ for
$\gamma_1<\beta<\kfr$, by Lemma~\ref{opt:lem001} (noting that
$\gamma_1$ is the larger of the two solutions of $\euomega(\beta)=\lambda$).
Hence $h_1<0$.
Now 
the minimum of (\ref{opt:eqn024}) is bounded below by the solution of the following problem:
\begin{align}
\mbox{minimise} \quad &\ \, 
           t_1 h_1+t_2h_2 
             \notag\\
\mbox{where} \quad \,\, & \, t_1 \eug_1+t_2 \eug_2 \,=\, k^{2-r}(\rho-1),\ \,
t_1 +t_2 \,\leq\, k, \ \,
t_1,\,t_2\,\in\,\nats_0.
\label{opt:eqn032}
\end{align}
We relax the equality constraint in \eqref{opt:eqn032} to give
\begin{align}
\mbox{minimise}\quad & \, 
         t_1 h_1+t_2h_2 
                        \notag\\
\mbox{where}\quad \,\, & \, t_1 \eug_1+t_2 \eug_2 \,\leq\, k^{2-r}(\rho-1),\ \,
t_1 +t_2 \,\leq\, k, \ \,
t_1,\,t_2\,\in\,\nats_0.
\label{opt:eqn033}
\end{align}
It follows that we must have $t_2=0$ in the optimal solution to \eqref{opt:eqn033}. To see this,
suppose the optimal solution is $t_1=\tau_1$, $t_2=\tau_2>0$. Consider the solution
$t_1=\tau_1+\tau_2$, $t_2=0$. This clearly satisfies the second and third constraint of \eqref{opt:eqn033}.
Since $\eug$ is decreasing on $[0,\kfr]$ we have
$\eug_1<\eug_2$, so
\[ (\tau_1+\tau_2)\eug_1< \tau_1\eug_1+\tau_2\eug_2\leq k^{2-r}(\rho-1).\]
Hence the solution $t_1 = \tau_1 + \tau_2$, 
$t_2=0$ also satisfies the first
constraint. Now $(\tau_1+\tau_2)h_1< \tau_1h_1+\tau_2h_2$ by \eqref{opt:eqn030}, contradicting the optimality of
$t_1=\tau_1, t_2=\tau_2$. Therefore, we will simply write $\beta_\ast$ for $\gamma_1$ and $t$ for $t_1$ from this point.

The rest of the argument also holds when $\euomega(\beta) =\lambda$ has
only one solution $\beta_\ast$, so this case re-enters the argument now.
By \eqref{opt:eqn031}, we must choose $t$ to be as large as possible subject to the constraints $t\leq k$ and $t\, \eug_1 \leq k^{2-r}(\rho-1)$. Therefore $t$ must be the smaller of $\lfloor k^{2-r}(\rho-1)/\eug(\beta_\ast)\rfloor$ and  $k$. We will usually relax the constraint $t\leq k$ below, since we are mainly interested in small values of $t$. In any case, this relaxation can only worsen the objective function. Recalling that
$\z{2}{\rho} = k\ln k - \zhat{2}{\rho}$, the objective function of the system  \eqref{opt:eqn019}
can be bounded above by 
\begin{equation}\label{opt:eqn034}
   \max\ \z{2}{\rho}\ \leq\ 
     k\ln k- \, t\, \euf(\beta_\ast),\quad  
  \textrm{where}\ \, t \,= \, \lfloor k^{2-r}(\rho-1)/\eug(\beta_\ast)\rfloor.
\end{equation}
Note that $\lambda$ has now been removed from consideration, as this upper bound on $\z{2}{\rho}$
depends on $\beta_\ast = \beta_\ast(\rho)$ only.

\subsection{Allowing $\rho$ to vary}\label{sec:varyrho}

Here we return to the relaxed problem given by (\ref{opt:eqn001}--\ref{opt:eqn003}), now
allowing the value of $\rho$ to vary.

As $\rho$ increases from 1 to $k^{r-1}$, the bound in \eqref{opt:eqn034} changes only at integral values of $k^{2-r}(\rho-1)/\eug(\beta_\ast)$. Thus the only relevant values of $\rho$ of are those for which $k^{2-r}(\rho-1)/\eug(\beta_\ast)$ is an integer. Then we may write \eqref{opt:eqn034} simply as 
\begin{equation}\label{opt:eqn035}
   \max\ \z{2}{\rho}\,\leq \,  k\ln k- \, t\, \euf(\beta_\ast),\quad
  \textrm{where}\ \, t \,= \, k^{2-r}(\rho-1)/\eug(\beta_\ast),\ \, t\,\in\,\nats_0.
\end{equation}

Let $\vecJ$ be the $k\times k$ matrix with all entries $\kfr$, and note
that $\vecJ/k = \vecJ_0$.
We wish to find conditions on $c$ which guarantee that $F(\vecA/k) < F(\vecJ/k)$
for all $\vecA\neq\vecJ$ which satisfy \eqref{opt:eqn003},~\eqref{opt:eqn004}.
From the above, and \eqref{opt:eqn001}, this will be true
when
\begin{equation}\label{amf3}
\frac{k\ln k - t\euf(\beta_\ast)}{k}+c\halfpt \ln\left(1-\frac{2}{k^{r-1}}+\frac{\rho}{k^{2r-2}}\right)\, <\,\ln k+2c\halfpt\ln\left(1-\frac{1}{k^{r-1}}\right),
\end{equation}
that is, when
\begin{equation}\label{opt:eqn036}
c\halfpt \ln\left(1+\frac{\rho-1}{(k^{r-1}-1)^2}\right)\,
   <\,\frac{t}{k}\, \euf(\beta_\ast).
\end{equation}
Next, from~\eqref{opt:eqn035} we have
$(\rho-1)=t\, \eug(\beta_\ast)k^{r-2}$.
Substituting this into~\eqref{opt:eqn036} gives
\begin{equation}\label{opt:eqn037}
c\halfpt \ln\left(1+\frac{k^{r-1}}{(k^{r-1}-1)^2}\halfpt\frac{t \eug(\beta_\ast)}{k}\right)\,
  <\,\frac{t \eug(\beta_\ast)}{k}\onept\, \eueta(\beta_\ast).
\end{equation}
Define
\begin{equation}\label{opt:eqn038}
\vartheta(\beta)\,=\,\ln\Big(1+\frac{k^{r-1}}{(k^{r-1}-1)^2}\halfpt\frac{t \eug(\beta)}{k}\Big),\quad
\mbox{and}\quad C(\beta)=\frac{(k^{r-1}-1)^2}{k^{r-1}}\halfpt\eueta(\beta).
\end{equation}
Then~\eqref{opt:eqn037} can be written as
\begin{equation}\label{opt:eqn039}
  c\,<\,C(\beta_\ast)\halfpt\frac{e^\vartheta-1}{\vartheta}\,=\,C(\beta_\ast)\sum_{i=0}^{\infty}\frac{\vartheta^i}{(i+1)!}.
\end{equation}
Now the right side of \eqref{opt:eqn039} is clearly minimised when $\vartheta$ is as small as possible.
From~\eqref{opt:eqn038}, this is when $t$ is as small as possible. If we relax the integrality constraint on $t$ and allow $t\to 0$,  then $\vartheta\to 0$ and~\eqref{opt:eqn039} becomes
$c<C(\beta_\ast)$. Thus we can estimate $c_{r,k}$ by minimising $C(\beta)$ over $\beta \in [0,\nicefrac{1}{k}]$. Then the computation of $c_{r,k}$ reduces to
minimising the function
\begin{equation}\label{opt:eqn040}
\eueta(\beta)\,=\,\frac{\ln k+(1-(k-1)\beta)\ln(1-(k-1)\beta) +(k-1)\beta\ln\beta}{(1-(k-1)\beta)^r+(k-1)\beta^r-1/k^{r-1}} \qquad (0\leq \beta \leq 1/k).
\end{equation}
Therefore we may take
\begin{equation}\label{opt:eqn041}
c_{r,k}\,=\,\frac{(k^{r-1}-1)^2}{k^{r-1}}\min_{\beta\in B}\halfpt\eueta(\beta)\,.
\end{equation}
Then, whenever $c\in (0,c_{r,k})$, we know that \eqref{opt:eqn039} holds, and
hence that $\vecJ$ is the unique maximum of $F$ over all doubly stochastic matrices.

\begin{remark}\label{chrom:rem002}
We have taken $\vartheta=0$ in~\eqref{opt:eqn039}, when the smallest value possible for $\vartheta$ is clearly larger. In \eqref{opt:eqn037}, $t\in\nats_0$ is the number of rows of $\vecA$ whose entries are not all $\kfr$. We wish to estimate $c_{r,k}$, which is the largest value of $c$ such that $\vecA\neq\vecJ$, so we must clearly have $t\geq 1$. However, we cannot have $t=1$, since~\eqref{opt:eqn004}--\eqref{opt:eqn006} imply that $\vecA$ cannot have a single row  whose entries are not all $\kfr$. Thus we may assume that $t\geq 2$.
Since $t$ should be as small as possible, we may take $t=2$. Then~\eqref{opt:eqn037} becomes
\begin{equation}\label{opt:eqn042}
c\, \leq\,\frac{2\euf(\beta)}{k\halfpt \ln\big(1+2k^{r-2}\eug(\beta)/(k^{r-1}-1)^2\big)}.
\end{equation}
We could use~\eqref{opt:eqn042} directly to improve the estimate of $c_{r,k}$. This is done in~\cite{AchMoo06} for $k=2$, giving a small improvement in $c_{r,2}$, though~\cite{AchNao05} uses only~\eqref{opt:eqn041} for $r=2$.
In the main, we will also use~\eqref{opt:eqn041}, which corresponds to allowing $t\to 0$.
However, we show in Section~\ref{sec:asymptotics} that the increment in $c_{r,k}$ which results from using~\eqref{opt:eqn042} is small, and can be obtained indirectly from~\eqref{opt:eqn039}.
\end{remark}

\begin{remark}\label{chrom:rem005}
We might improve the estimate of $c_{r,k}$ further by avoiding the
relaxation of~\eqref{opt:eqn005} in the optimisation. We note that taking $t=k$
in~\eqref{opt:eqn037} results in a local maximum of~\eqref{opt:eqn001},
as follows. Let $p$ be any permutation of $[k]$, and set
$a_{i p(i)}=\alpha$, $a_{ij}=\beta$ ($j\neq  p(i),\,i\in[k]$). This gives $k!$
local maxima of~\eqref{opt:eqn001}. We conjecture
that these solutions are the global maxima, but we are unable to prove this.
The inclusion of~\eqref{opt:eqn005} gives conditions for the local maxima
which may have solutions yielding larger values of $z$ in~\eqref{opt:eqn001}.
The local maxima seem rather difficult to describe explicitly, so
we leave this as an open question. However, we show in Section~\ref{sec:asymptotics}
that including~\eqref{opt:eqn005} cannot result in a large
improvement in $c_{r,k}$.
\end{remark}

\subsection{The univariate optimisation}\label{sec:min eta}

We have now achieved the objective of reducing the problem to a univariate optimisation, namely,
minimising the function $\eueta$. To carry out this
minimisation, we will first make a substitution $x=(k-1)\beta$ in \eqref{opt:eqn040}, so that
\begin{equation}\label{opt:eqn043}
\eta(x) = \eueta\left( \frac{x}{k-1}\right)\,=\,\frac{\ln k - x\ln(k-1)+(1-x)\ln(1-x) +x\ln x}
{(1-x)^r+x^r/(k-1)^{r-1}-1/k^{r-1}}\,=\,\frac{f(x)}{g(x)}\qquad(x\in[0,1-1/k]),
\end{equation}
where, using \eqref{opt:eqn023}, we let
\begin{align*}
  f(x) = \euf\left(\frac{x}{k-1}\right)\,&=\,\ln k - x\ln(k-1)+(1-x)\ln(1-x) +x\ln x,\\
  g(x) = \eug\left(\frac{x}{k-1}\right)\,&=\,(1-x)^r+x^r/(k-1)^{r-1}-1/k^{r-1}.
\end{align*}
Figure~\ref{fig:eta} gives a plot of the function $\eta$ when $k=4$
and $r=3$.
\begin{figure}[ht]
\begin{center}
\includegraphics[scale=0.5]{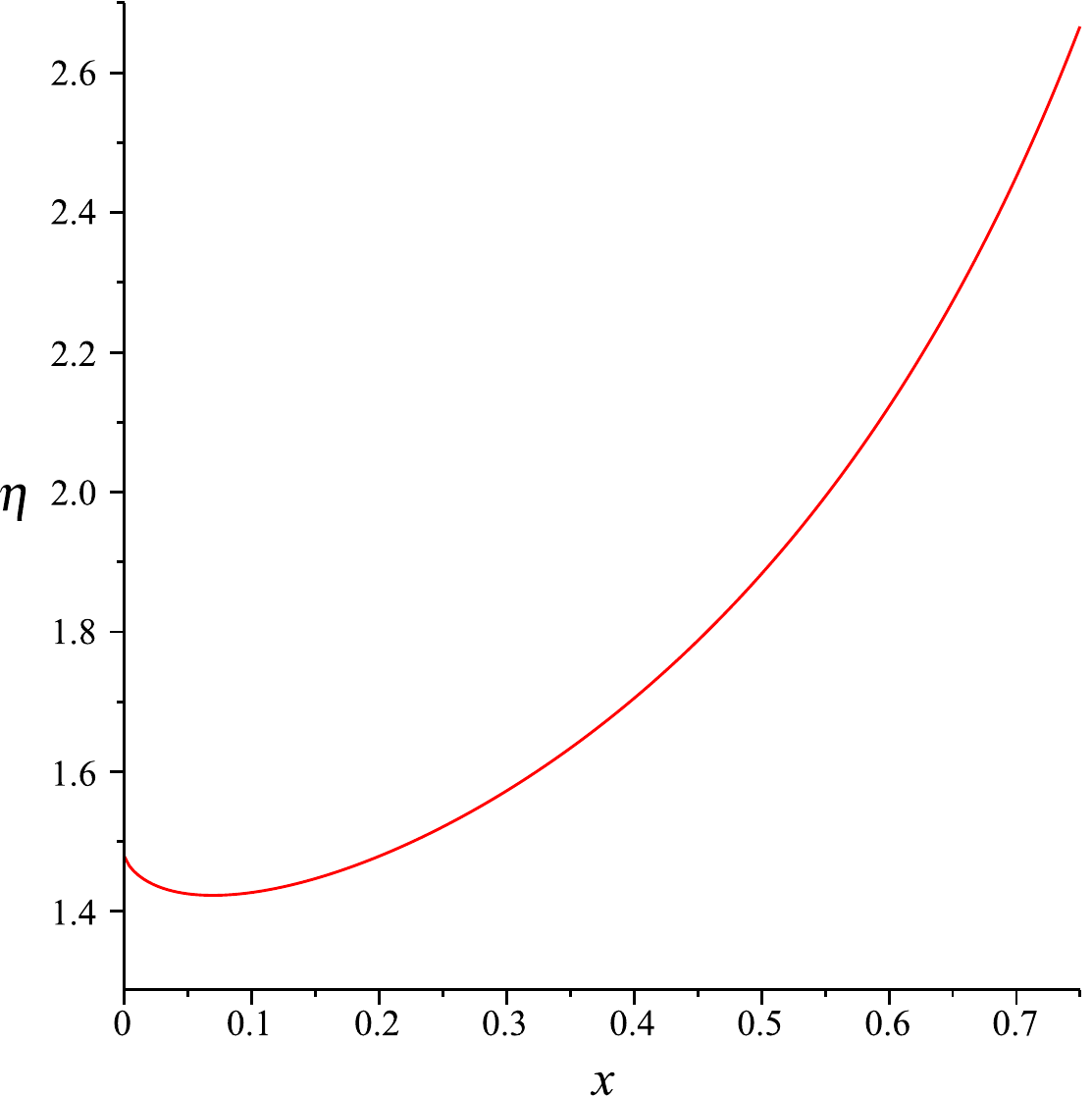}
\label{fig:eta}
\caption{The function $\eta$ when $k=4$ and $r=3$.}
\end{center}
\end{figure}
Now
\[  f(x) =\,\ln\big(k(1-x)\big) - x\ln\big((k-1)(1-x)/x\big)\]
for $x\in(0,1-\kfr)$,
and at the boundaries we have $f(0)=\ln k$ and $f(1-\kfr)=0$.
Differentiating gives
\begin{align}
  f'(x)\,&=\,- \ln(k-1)-\ln(1-x) +\ln x\label{opt:eqn044} \\
  &=\,-\ln\big((k-1)(1-x)/x\big)< 0\quad\mbox{ for }x\in(0,1-\kfr).\notag
\end{align}
Therefore $f(x) > f(1-\kfr) = 0$ for all $x\in (0,1-\kfr)$.
Also $\lim_{x\to0}f'(x)=-\infty$ while
 $f'(1-\kfr)=0$.
Note, using \eqref{opt:eqn044}, that
\begin{align}
  f(x)\,&=\,\ln\big(k(1-x)\big)+xf'(x).\label{opt:eqn045}\\[-\baselineskip]
  \intertext{Also,\vspace{-\baselineskip}}
  f''(x)\,&=\,\frac{1}{1-x} +\frac{1}{x}\,=\,\frac{1}{x(1-x)}>0\quad\mbox{ for }x\in(0,1-\kfr], \label{opt:eqn046}\\
  f'''(x)\,&=\,\frac{1}{(1-x)^2} -\frac{1}{x^2}.
\end{align}
  We note that $f''(1-\kfr)=k^2/(k-1)$ and
$f'''(1-\kfr)=k^3(k-2)/(k-1)^2$. Now we turn our attention
to the function $g$,
which satisfies $g(0)=1-1/k^{r-1}$ and
$g(1-\kfr)=0$. Differentiating gives
\begin{align}
  g'(x)\,&=\,-r\big((1-x)^{r-1}-x^{r-1}/(k-1)^{r-1}\big) < 0 \quad
 \mbox{ for } x\in(0,1-\kfr),\notag
  \intertext{which shows that $g(x) > g(1-\kfr)=0$ for $x\in (0,1-\kfr)$.
Also $g'(0)=-r$ and $g'(1-\kfr)=0$. Finally,}
  g''(x)\,&=\,r(r-1)\big((1-x)^{r-2}+x^{r-2}/(k-1)^{r-1}\big)>0 \quad
  \mbox{ for } x\in(0,1-\kfr],\label{opt:eqn047}\\
  g'''(x)\,&=\,-r(r-1)(r-2)\big((1-x)^{r-3}-x^{r-3}/(k-1)^{r-1}\big).\notag
\intertext{Note that, when $r=2$,
$g''$ is constant and $g'''$ is identically zero. Also, in particular,}
g''(1-\kfr)\,&=\,\frac{r(r-1)}{(k-1)k^{r-3}},\quad
g'''(1-\kfr)\,=\,-\frac{r(r-1)(r-2)(k-2)}{(k-1)^2k^{r-4}}.  \notag
\end{align}
Hence $f(x)$ and $g(x)$ are positive, strictly decreasing and strictly convex functions on $(0,1-\kfr)$.

Returning to the function $\eta$ defined in \eqref{opt:eqn043},
in Lemma~\ref{app:lem008} we show that
\begin{align}
\lim_{x\to 1-1/k}\eta(x)\,=\,\frac{k^{r-1}}{r(r-1)},\quad\,\,\,
\lim_{x\to0}\eta'(x)\,&=\,-\infty,\quad\,\,\,
\lim_{x\to 1-1/k}\eta'(x)\,=\,\frac{(k-2)k^r}{r(k-1)}\,\geq\,0,
\label{opt:eqn048}
\end{align}
and we will take these limits as defining
$\eta(1-\kfr)$, $\eta'(0)$ and
$\eta'(1-\kfr)$, respectively. Note also that
$\eta(0) = k^{r-1}\ln k/(k^{r-1}-1)$.

If $k=2$ then $\eta$ has a stationary point at $x=1-\kfr= \nicefrac12$. Otherwise,
$\eta$ has an interior minimum in $(0,1-\kfr)$,
since $\eta'(0)<0$ and $\eta'(1-\kfr)>0$.
We first show that this is the unique stationary point of $\eta$ in
$(0,1-\kfr)$.   This is not straightforward, since $\eta$ is not
convex, as observed in~\cite{AchNao05} for the case $r=2$.
Furthermore, the approach of~\cite{AchNao05}, making a nonlinear
substitution in $\eta$, does not generalise beyond $r=2$. Hence our
arguments here are very different from those in~\cite{AchNao05}.

To determine the nature of the stationary points of $\eta$, we consider the
function $h(x)=f(x)-\lambda g(x)$ on $(0,1-\kfr]$,
for fixed $\lambda>0$. Then $h$ is analytic, and its zeros contain the points
at which $\eta(x)=\lambda$ in $(0,1-\kfr]$. We will apply Rolle's Theorem~\cite{Spivak06} to $h$. The zeros
of $h$ are separated by zeros of $h'$, and these are separated by zeros of $h''$.
Since $f(1-\kfr)=g(1-\kfr)=0$ and
$f'(1-\kfr)=g'(1-\kfr)=0$, we conclude that
$h'$ has a zero at $x=1-\kfr$
for all $\lambda$, and $h$ has a double zero at $x=1-\kfr$. Now, from
\eqref{opt:eqn046} and \eqref{opt:eqn047}, the zeros of $h''(x)=f''(x)-\lambda
g''(x)$ in $(0,1-\kfr]$ are the solutions of
\begin{equation}\label{opt:eqn049}
 x(1-x)^{r-1}+\frac{(1-x)x^{r-1}}{(k-1)^{r-1}}\,=\, \frac{1}{\lambda r(r-1)}.
\end{equation}
In Lemma~\ref{app:lem009} we show that if $r\leq 2k$ then \eqref{opt:eqn049}
has at most two solutions in $[0,1]$, while if $r\geq 2k+1$ then \eqref{opt:eqn049}
has  at most two solutions in $[0,1-\kfr]$ whenever $\lambda <
\lambda_0$, where
\begin{equation*}
\frac{1}{\lambda_0}\ =\ r(r-1)\left(\frac{(r-2)2^{r-1}}{r^r}+\frac{1}{k^r}\right).
\end{equation*}
(Here, as elsewhere in the paper, we have $r,k\geq 2$.)
For uniformity, we set $\lambda_0=\infty$ if $r\leq 2k$ and define
\[ \Lambda=\set{x\in[0,1-\kfr]:\eta(x)<\lambda_0},\quad
\Lambda'=\Lambda\cap (0,1-\kfr).\]
Then $\Lambda'$ is a union of open intervals. We show
in Lemma~\ref{app:lem011} that $\eta(0) < \eta(1-\kfr) < \lambda_0$,
which implies that $0,\,1-\kfr\in\Lambda$. Hence
$\Lambda=\Lambda'\cup\set{0,1-\kfr}$, which shows that $\Lambda$
is nonempty. Now $\eta'(0)<0,\,\eta'(1-\kfr)\geq0$ imply that $\Lambda'$ is nonempty. Our search for a value of $x$ making $\eta$ small will be restricted to $\Lambda'$.
We have shown that $h''$ has at most two zeros in $\Lambda$, and
hence $h$ has at most four zeros in $\Lambda$.
Since there is a double zero of $h$ at
$x=1-\kfr\in \Lambda\setminus\Lambda'$, it follows that
there are at most two zeros of $h$ in $\Lambda'$. Thus
$\eta(x)=\lambda$ at most twice in $\Lambda'$. Since $\eta'(0)<0,\,
\eta'(1-\kfr)\geq0$, we know that $\eta$ has a
local minimum in $\Lambda'$. Then $\eta$ has at most one local minimum
$\xi\in\Lambda'$. To see this, suppose there are two local minima
$\xi_1,\,\xi_2 \in\Lambda'$ with $\eta(\xi_1)\leq \eta(\xi_2) = \lambda<\lambda_0$.
If $\eta(\xi_1) = \lambda$ then $\eta(x)=\lambda$ has at least four roots
in $\Lambda'$, with double roots at both $\xi_1$ and
$\xi_2$. If $\eta(\xi_1)<\lambda$ then $\eta(x)=\lambda$
has at least three roots in $\Lambda'$, with a double root at $\xi_2$
and, by continuity, a root strictly between $\xi_1$ and $\xi_2$. In either case, we have a contradiction. It
also follows that $\Lambda$ is connected. Otherwise, since $\eta'(0)<0,\, \eta'(1-\kfr)\geq0$, each maximal
interval of $\Lambda'$ must contain a local minimum, a contradiction.
Thus $\Lambda=[0,1-\kfr]$.  In other words,
\begin{equation}
\label{needed}
\eta(x) < \lambda_0 \quad \text{ for all } x\in [0,1-\nicefrac{1}{k}].
\end{equation}

We have proved that $\eta$ has exactly one local minimum
in $(0,1-\kfr)$, and we will denote this minimum point by
$\xi\in (0,1-\kfr)$.   It also follows that
there are no local maxima of $\eta$ in $[0,1-\kfr]$, as we now prove.
If there were a local
maximum $\xi'\in[0,\xi)$ then $\eta'(0)<0$ would imply that there is a local minimum
in $(0,\xi')$, a contradiction. The same argument
applies to the interval $(\xi,1-\kfr]$, for $k>2$. If $k=2$, it is possible that $x=\nicefrac{1}{2}$
is a local maximum, but it still follows
that there can be no local maximum in $(\xi,\nicefrac12)$.

To summarise: if $k>2$ then $\eta$ has
exactly one stationary point $\xi\in (0,1-\kfr)$,
a local minimum. If $k=2$ then there is a unique local minimum
$\xi\in (0,\nicefrac12]$ but,
if $\xi\neq \nicefrac{1}{2}$, then $\nicefrac{1}{2}$ may
be a local maximum. In either case, $\xi$ is the global minimum.

We now prove Lemma~\ref{opt:lem001}, using the same method but
working with the transformed function $\omega$ defined by
\[ \omega(x) = \euomega\left(\frac{x}{k-1}\right) =
  \frac{\ln(k-1) + \ln(1-x) - \ln(x)}{r\left((1-x)^{r-1} - x^{r-1}/(k-1)^{r-1}
        \right)}
 = \frac{f'(x)}{g'(x)}.
\]
Figure~\ref{fig:omega} gives a plot of the function $\omega$
when $k=4$ and $r=3$.
\begin{figure}[ht]
\begin{center}
\includegraphics[scale=0.5]{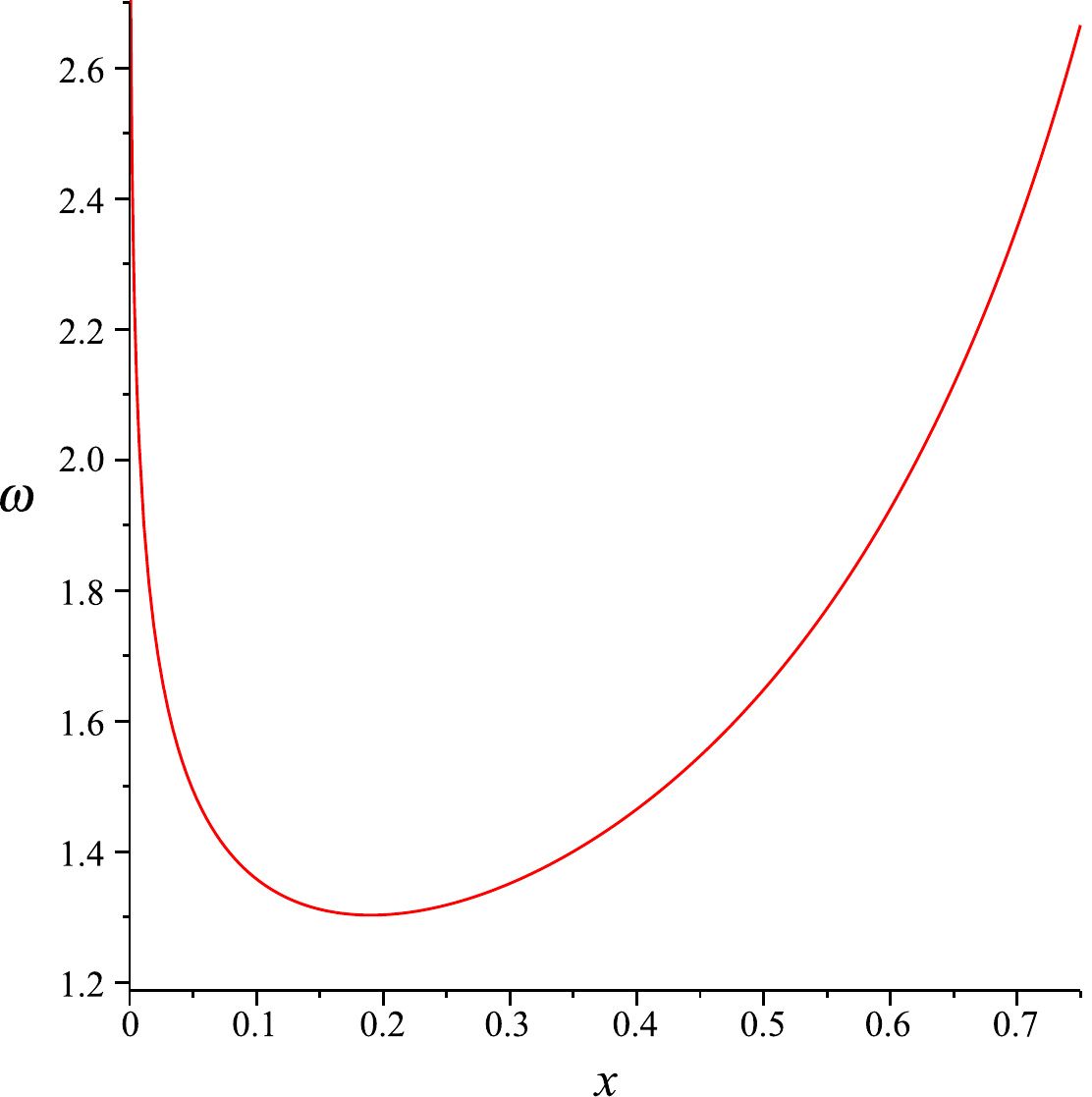}
\label{fig:omega}
\caption{The function $\omega$ when $k=4$ and $r=3$.}
\end{center}
\end{figure}

\begin{proof}[Proof of Lemma~\ref{opt:lem001}]
Let $\lambda$ be a real number which satisfies \eqref{opt:eqn029}.
Note that the
solutions to $\euomega(\beta)=\lambda$ in~\eqref{opt:eqn028} correspond to
the zeros of $h'(x)$, where $h(x)$ is the function defined above.
Combining~\eqref{opt:eqn029} and \eqref{needed},
we see that $\lambda < \lambda_0$.
Therefore by Lemma~\ref{app:lem009} and Lemma~\ref{app:lem011},
we may conclude that
$h''(x)$ has at most two zeros in $[0,1-\kfr ]$.
Hence $h'(x)$ has at most three zeros, and we know that $h'(1-\kfr)=0$.
Thus there can be at most two zeros of $h'(x)$ in $[0,1-\kfr)$. Therefore
$\omega(x)=f'(x)/g'(x)$ can take the value
$\lambda$ at most twice in $[0,1-\kfr)$.  Since $\omega$ is analytic on $(0,1-\kfr)$, by the arguments above, $\omega$ can have at most one stationary point in $(0,1-\kfr)$.

Now $\omega(0)=+\infty$ since $f'(0)=-\infty$ and $g'(0)=-r$.
By Lemma~\ref{app:lem008} 
$\omega(1-\kfr)=\eta(1-\kfr)=k^{r-1}/r(r-1)$ and that $\omega(\xi)=\eta(\xi)$,
where $\xi$ denotes the point which minimises $\eta$.
Since $\omega(\xi)=\eta(\xi)<+\infty=\omega(0)$ and $\omega(\xi)=\eta(\xi)< \eta(1-\kfr)=\omega(1-\kfr)$, $\omega$ must have a unique minimum in $(0,1-\kfr)$,
completing the proof.
\end{proof}

It remains to identify the local minimum $\xi$ of $\eta$ to a close enough approximation.
Using~\eqref{opt:eqn045}, the condition that $\eta'(x)\leq 0$ is
\[
  g'(x)\big(\ln(k(1-x))+f'(x)(x-g(x)/g'(x))\big) \geq 0,\qquad x\in(0,1-\kfr).
\]
We have shown that $f'(x) < 0$ and $g'(x)<0$ for $x\in(0,1-\kfr)$, so the condition $\eta'(x)\leq 0$ is equivalent to
\begin{equation}\label{opt:eqn050}
 x-\frac{g(x)}{g'(x)}\,\geq\, \frac{\ln\big(k(1-x)\big)}{-f'(x)}=
\frac{\ln\big(k(1-x)\big)}{\ln\big((k-1)(1-x)/x\big)},\qquad x\in(0,1-\kfr).
\end{equation}
We will now use \eqref{opt:eqn050} to show that $\xi$ is approximately
$1/k^{r-1}$,  except for the cases $k=2,\,r=3,\,4$.
(If $r=2$ then $\xi=1/k^{r-1}$ exactly.)
This will enable us to determine the value of $c_{r,k}$ and
establish that Lemma~\ref{moment:lem002} holds.

\subsection{The case $k=2$}\label{sec:k=2}

We will first examine the case $k=2$ in more detail. We must determine whether
$x=\nicefrac12$ is a local minimum or maximum of $\eta$. If it is a
local minimum, then it is the global minimum. Otherwise, there is a unique local
minimum $\xi\in(0,\nicefrac12)$. To resolve this, we must examine
$\eta$ in the neighbourhood of $x=\nicefrac12$. We show in Lemma~\ref{app:lem012}
that $\nicefrac12$ is a local minimum of $\eta$ for $2\leq r\leq 4$,
but is a local maximum if $r\geq 5$. Thus, for $r=2,\,3,\,4$, the global minimum is
$\xi=\nicefrac12$. (Note that we include the case $r=k=2$ here,
though ultimately it plays no part in our analysis.)  Hence from~\eqref{opt:eqn041} and~\eqref{opt:eqn048} we have
that for $r=2,\,3,\,4$,
\begin{equation}\label{opt:eqn051}
  c_{r,2}\,=\,\frac{(2^{r-1}-1)^2}{2^{r-1}}\frac{2^{r-1}}{r(r-1)}
\,=\,\frac{(2^{r-1}-1)^2}{r(r-1)}.
\end{equation}
Specifically,
\begin{equation}\label{opt:eqn052}
  c_{2,2}\,=\nicefrac12\,=\,0.5,\quad c_{3,2}\,=\,\nicefrac32\,=\,1.5,
\quad c_{4,2}\,=\,\nicefrac{49}{12}\,\simeq\,4.0833.
\end{equation}
Now $u_{r,1}=0$ for all $r$, and $u_{r,2}=2^{r-1}\ln2$, so
\begin{equation}\label{opt:eqn053}
u_{2,2}\,\simeq\,1.3863,\quad u_{3,2}\,\simeq\,2.7726,\quad
   u_{4,2}\,\simeq\,5.5452.
\end{equation}
It follows that
\begin{equation}\label{opt:eqn054}
u_{r,1} < c_{r,2} < u_{r,2} \quad \text{ for $r=2,3,4$,}
\end{equation}
as required.
(We cannot use this result in Theorem~\ref{thm:chrom001} when $k=r=2$, since there is no sharp threshold in this case.)

In the cases $k=2,\,r\geq 5$, there is a local minimum $\xi\in(0,\nicefrac12)$, so
the optimisation has similar characteristics to $k\geq 3$. We
consider these cases in Section~\ref{sec:general} below.

\subsection{The case $r=2$}\label{sec:r=2}

We will consider the case $r=2$ separately, since $\eta$ can be minimised
exactly in this case. The results given in this section were obtained by Achlioptas and Naor in~\cite{AchNao05}, by making a nonlinear substitution in $\eta$. We can derive their results more simply, since we know that $\eta$ has a unique minimum. We have
\begin{align*}
g(x)\,&=\,(1-x)^2+\frac{x^2}{k-1}-\frac{1}{k}\,=\,
\frac{k}{k-1}\Big(x-\frac{k-1}{k}\Big)^2,\\
\text{ so }\, g'(x)\,&=\,\frac{2k}{k-1}\Big(x-\frac{k-1}{k}\Big).\\
\intertext{It follows that}
x-\frac{g(x)}{g'(x)}\,&=\,
   \half \big(x+(k-1)/k\big).
\end{align*}
Hence \eqref{opt:eqn050} implies that $x$ minimises $\eta$ if and only if
\[
x+\frac{k-1}{k}\ =\ \frac{2\halfpt\ln\big(k(1-x)\big)}{\ln\big((k-1)(1-x)/x\big)}.
\]
It is easily verified that $x=1/k$ satisfies this equation, and hence is the unique minimum of $\eta$ in $[0,1-\kfr]$.

We have dealt with the case $k=2$ in the previous section, so we now
assume that $k\geq 3$.  Then
\[
\min_{x\in[0,1-\kfr]}\eta(x)\,=\,\frac{f(1/k)}{g(1/k)}\,=\,\frac{(k-1)\ln(k-1)}{k-2}
\]
and hence
\begin{equation}
\label{c2k}
  c_{2,k}\, =\, \frac{(k-1)^2}{k}\, \eta(1/k)\,=\,
   \frac {(k-1)^3\ln(k-1)}{k(k-2)}.
\end{equation}
Now $(k-1)^3/k(k-2)=(k-1)(1+1/k(k-2))$ which
lies strictly between $k-1$ and $k$, for $k\geq3$. Thus
\begin{equation}
  u_{2,k-1} = (k-1)\ln(k-1)\, < \, c_{2,k}\,
   < \,k\ln k\, = \, u_{2,k}     \label{opt:eqn055}
\end{equation}
and moreover
\begin{equation}
c_{2,k}  \leq \frac{(k-1)^2}{2},\label{opt:c2k-bound}
\end{equation}
as required for Lemma~\ref{moment:lem002}.

\subsection{The general case}\label{sec:general}
We now consider the remaining cases $k\geq 3$ or $k=2,\,r\geq 5$. We will do this by
finding values $w,y\in(0,1-\kfr)$ such that $\eta'(w)\leq 0$ and
$\eta'(y)>0$.  That is, $w$ satisfies~\eqref{opt:eqn050}, but $y$ does not. The
uniqueness of $\xi$ then implies that $w\leq\xi<y$, and we will use this to
place a lower bound on $\eta(\xi)$. We will achieve this for all pairs
$r,\,k$ except for a small number, and we will solve
these few remaining cases numerically.

To simplify the analysis, we will exclude some cases initially.
Thus we assume below that
\begin{equation}\label{opt:eqn056}
  k=2,\ r\geq 9\ \,\, \text{ or }\  \,\, k=3,\,  r\geq4 \ \,\, \text{ or } \ \,\,  k\geq 4,\ r\geq 3.
\end{equation}
By Lemma~\ref{app:lem013}, the inequality
\begin{equation}
\label{opt:lem002}
r^2(k+2)/k^r<1
\end{equation}
holds whenever \eqref{opt:eqn056} holds.

First we set $x=w$ in \eqref{opt:eqn050}, where $w=(k-1)/k^r$. Note
that $w<1/r^2$, from \eqref{opt:lem002}.
Using Lemmas~\ref{app:lem001}, \ref{app:lem002} and~\ref{app:lem005}, we have
\begin{align}
  \frac{r\halfpt\ln\big(k\halfpt(1-w)\big)}{\ln\big((k-1)(1-w)/w\big)}\,
    &<\, \frac{r\ln k-rw}{r\ln k-3w/2}\,=\,\frac{1-w/\ln k}{1-3w/(2r\ln k)} \notag\\
    &<\, \Big(1-\frac{w}{\ln k}\Big)\Big(1+\frac{3w}{r\ln k}\Big)\,
    <\,1-\Big(1-\frac{3}{r}\Big)\frac{w}{\ln k}\,\leq\,1, \label{opt:eqn057}
\end{align}
since $r\geq 3$.

Using Lemma~\ref{app:lem006}, we have
\begin{align*}
  g(w)\,&=\,(1-w)^r+\frac{w^r}{(k-1)^{r-1}}-\frac{1}{k^{r-1}}\,
\geq\,1-rw-\frac{kw}{k-1}
 \,=\,1-\frac{(k-1)r+k}{k-1}w, \\
 \frac{-g'(w)}{r}\,&=\,(1-w)^{r-1}-\frac{w^{r-1}}{(k-1)^{r-1}}\,\leq\,(1-w)^{r-1}\,\leq\,\frac{1}{1+(r-1)w}.
\end{align*}
So we have
\begin{align*}
rw-\frac{r g(w)}{g'(w)}\,&\geq\,rw+\Big(1-\frac{(k-1)r+k}{k-1}w\Big)\big(1+(r-1)w\big)\\
&>\,1+\Big(r-\frac {2k-1}{k-1}-(r-1)\frac{k(r+1)}{k-1}\frac{k-1}{k^r}\Big)w\\
&>\,1+\Big(r-\frac{2k-1}{k-1}-\frac{r^2}{k^{r-1}}\Big)w,\\
&=\,1+\Big(r-2-\frac{1}{k-1}-\frac{r^2}{k^{r-1}}\Big)w,
\end{align*}
and the right hand side is bounded below by $1$ whenever
\begin{equation}\label{opt:eqn058}
  \frac{1}{k-1}+\frac{r^2}{k^{r-1}}\leq r-2.
\end{equation}
We may easily show that the left hand side of~\eqref{opt:eqn058} is decreasing with $r$ for
$r\geq 3$, and it is clearly decreasing with $k\geq 2$. The right hand side is independent
of $k$ and increasing with $r$.  Now \eqref{opt:eqn058} holds by calculation when $(k,r) \in
\{(2,5),\, (3,4),\, (4,3)\}$.  Therefore \eqref{opt:eqn058} holds for
all $(k,r)$ which satisfy \eqref{opt:eqn056}, and combining this
with~\eqref{opt:eqn057} shows that $w$
satisfies \eqref{opt:eqn050}, as desired.

We now set $x=y$ in \eqref{opt:eqn050}, where $y=(k+2)/k^r$.
We have $ry<1/r$ from \eqref{opt:lem002}.
Then, using Lemmas~\ref{app:lem001} and~\ref{app:lem002}, we have
\begin{align*}
 \frac{r\halfpt\ln\big(k\halfpt(1-y)\big)}{\ln\big((k-1)(1-y)/y\big)}
&>\,\frac{\ln k^r-3ry/2}{\ln k^r+\ln\big(1-3/(k+2)\big)+\ln(1-y)}\,
>\,\frac{\ln k^r-3ry/2}{\ln k^r-3/(k+2)}\\[0.5ex]
&=\,1+\frac{3/(k+2)-3ry/2}{\ln k^r-3/(k+2)} \,
>\,1+\frac{3/(k+2)-3ry/2}{r\ln k}\,.
\end{align*}

Using Lemma~\ref{app:lem006},
\begin{align*}
  g(y)\, &=\,(1-y)^r+\frac{y^r}{(k-1)^{r-1}}-\frac{1}{k^{r-1}}\,
  \leq\,1-ry+\half(ry)^2+\frac{y^r}{(k-1)^{r-1}}-\frac{ky}{k+2}\\
  &=\,1-\Big(r-\half r^2y-\frac{y^{r-1}}{(k-1)^{r-1}}+\frac{k}{k+2}\Big)y,\\
 \frac{-g'(y)}{r}\,  &=\,(1-y)^{r-1}-\frac{y^{r-1}}{(k-1)^{r-1}}\,
\geq\,1-(r-1)y-\frac{y^{r-1}}{(k-1)^{r-1}}\\
&=\,1-\Big(r-1+\frac{y^{r-2}}{(k-1)^{r-1}}\Big)y.
\end{align*}
Now $y^{r-2}/(k-1)^{r-1}<1$ for $r, k\geq 2$ and $ry<1/r<\nicefrac12$, using Lemma~\ref{app:lem005}.
Therefore
\begin{align*}
\frac{r}{-g'(y)}\,&\leq\,1+\Big(r-1+\frac{y^{r-2}}{(k-1)^{r-1}}\Big)y +2\Big(r-1+\frac{y^{r-2}}{(k-1)^{r-1}}\Big)^2y^2\\
&<\,1+\Big(r-1+\frac{y^{r-2}}{(k-1)^{r-1}}+2r^2y\Big)y.
\end{align*}
Thus
\begin{align*}
ry-\frac{r g(y)}{g'(y)}\,&\leq\,ry+\Big(1-\Big(r-\half r^2y-\frac{y^{r-1}}{(k-1)^{r-1}}+\frac{k}{k+2}\Big)y\Big)
\Big(1+\Big(r-1+\frac{y^{r-2}}{(k-1)^{r-1}}+2r^2y\Big)y\Big)\\
&\leq\,1+\Big(r-1+\frac{5r^2y}{2}+\frac{(1+y)y^{r-2}}{(k-1)^{r-1}}-\frac{k}{k+2}\Big)y.
\end{align*}
So $p'(y)>0$ if $y$ does not satisfy~\eqref{opt:eqn050}; that is, if
\[ \frac{3/(k+2)-3ry/2}{r\ln k}\,>\, \Big(r-1+\frac{5r^2y}{2}+\frac{(1+y)y^{r-2}}{(k-1)^{r-1}}-\frac{k}{k+2}\Big)y.
\]
Dividing by $y$ and rearranging gives the equivalent condition
\begin{equation}
\frac{3k^r}{r(k+2)^2\ln k}\,>\,r-2+\frac{2}{k+2}+\frac{3}{2\ln k}+\frac{5r^2y}{2}+\frac{(1+y)y^{r-2}}{(k-1)^{r-1}}\,.
\label{opt:eqn059}
\end{equation}
From Lemma~\ref{app:lem013}, we have $r^2y\leq 1$ and that $y=(r^2y)/r^2$ is
decreasing with both $r$ and $k$. Since $y<1$, it follows easily
that $(1+y)y^{r-2}/(k-1)^{r-1}$ is decreasing with $r$ and $k$. We may now
check numerically that $(1+y)y^{r-2}/(k-1)^{r-1}\leq \nicefrac{1}{50}$
for all $k,\,r$ satisfying \eqref{opt:eqn056}. It follows that \eqref{opt:eqn059} is implied by the inequality
\begin{equation}\label{opt:eqn060}
  \frac{3k^r}{r^2(k+2)^2\ln k}\,\geq\,1+\frac{0.52}{r}+\frac{2}{r(k+2)}+\frac{3}{2r\ln k}.
\end{equation}
We show in Lemma~\ref{app:lem014} that, if \eqref{opt:eqn060} holds for
some $r\geq 3$, $k\geq 2$, then it holds for any $r',k'$ such that
$r'\geq r$, $k'\geq k$. We may verify numerically that \eqref{opt:eqn060}
holds for the following pairs  $r,k$.
\[ k=2,\ r=9 ,\quad k=3,\ r=6 ,\quad k=4,\ r=5 ,
\quad k =5,\ r=4 ,\quad k=15,\ r=3.\]

Thus it holds for all pairs $r,\,k$ such that
\[ k=2,\ r\geq 9 ,\quad k=3,\ r\geq 6 ,\quad k=4,\ r\geq 5 ,
\quad k\in \{ 5,\ldots, 14\},\ r\geq 4 ,\quad k\geq 15,\ r\geq 3.\]
Let us call these the pairs $(k,r)$ \emph{regular}, with the remaining nineteen pairs being \emph{irregular}.
We deal with the irregular pairs below by numerical methods.

First we continue our focus on regular pairs.
For such pairs we have argued that $(k-1)/k^r\leq \xi<(k+2)/k^r$ and hence,
 using Lemmas~\ref{app:lem004} and~\ref{app:lem006},
\begin{align*}
  f(\xi)\,&=\,\ln k -\xi\ln(k-1)+\xi\ln\xi+(1-\xi)\ln(1-\xi)\,>\,\ln k-(r\ln k+1)\xi\,,\\
  g(\xi)\,&=\, (1-\xi)^r+\xi^r/(k-1)^{r-1}-1/k^{r-1}\,<\, 1/(1+r\xi)\,.
  \intertext{Hence, using Lemma~\ref{app:lem015},}
  \eta(\xi)\,&>\,\big(\ln k-(r\ln k+1)\xi\big)(1+r\xi)\,=\, \ln k -\xi -r(r\ln k+1)\xi^2\,\geq\,\ln k - 2\xi\,.
\end{align*}
From \eqref{opt:eqn041}, we can now determine
\begin{equation*}
  c_{r,k}\,\geq\,\frac{(k^{r-1}-1)^2}{k^{r-1}}\left(\ln k -\frac{2(k+2)}{k^r}\right)\,>\,(k^{r-1}-2)\ln (k-1),
\end{equation*}
using Lemma~\ref{app:lem016}.
Now  $k^{r-1}>(k-1)^{r-1}+(r-1)(k-1)^{r-2}\geq (k-1)^{r-1}+2$ for $r\geq 3,\,k\geq 2$, which shows that
\begin{equation}
\label{opt:eqn061a}
 u_{r,k-1}=(k-1)^{r-1}\ln(k-1)\,<\,c_{r,k}
\end{equation}
for all regular pairs.

We also have
\begin{equation}\label{opt:eqn061}
 c_{r,k}\,<\,\frac{(k^{r-1}-1)^2}{k^{r-1}}\, \eta(0)\,=\,
\frac{(k^{r-1}-1)^2}{k^{r-1}}\frac{k^{r-1}\ln k}{k^{r-1}-1}
\,=\,(k^{r-1}-1)\ln k\,<\,k^{r-1}\ln k\,=\,u_{r,k},
\end{equation}
as required, and this holds for all $r,\,k\geq 2$.

Next we consider irregular pairs and use \eqref{opt:eqn050} to bound $\xi$
numerically, by bisection. This is quite straightforward, since
we know that $\xi\in(0,1-\kfr)$ is unique. The resulting values of
$c_{r,k}$ are shown below, along with the corresponding values of
$u_{r,k-1}$ and $u_{r,k}$.\vspace{\baselineskip}
\[ \renewcommand{\arraystretch}{1.0}
\begin{array}{|r|r|r|r|r|}\hline
\rule[-6pt]{0pt}{18pt}
k & r & u_{r,k-1}\ \  & c_{r,k}\ \ \  & u_{r,k}\ \ \ \\\hline
\multirow{4}{*}{2} &  5 &   0.0000 &   9.8771&  11.0904\\
 &  6 &   0.0000 &  21.2990&  22.1807\\
 &  7 &   0.0000 &  43.7678&  44.3614\\
 &  8 &   0.0000 &  88.3486&  88.7228\\
\hline
\multirow{3}{*}{3} &  3 &   2.7726 &   8.1566&   9.8875\\
 &  4 &   5.5452 &  27.9595&  29.6625\\
 &  5 &  11.0904 &  87.4703&  88.9876\\
\hline
\multirow{2}{*}{4} &  3 &   9.8875 &  20.0491&  22.1807\\
 &  4 &  29.6625 &  86.6829&  88.7228\\
\hline
\multirow{1}{*}{5} &  3 &  22.1807 &  37.8417&  40.2359\\
\hline
\multirow{1}{*}{6} &  3 &  40.2359 &  61.8958&  64.5033\\
\hline
\multirow{1}{*}{7} &  3 &  64.5033 &  92.5637&  95.3496\\
\hline
\multirow{1}{*}{8} &  3 &  95.3496 & 130.1457& 133.0843\\
\hline
\multirow{1}{*}{9} &  3 & 133.0843 & 174.9034& 177.9752\\
\hline
\multirow{1}{*}{10} &  3 & 177.9752 & 227.0688& 230.2585\\
\hline
\multirow{1}{*}{11} &  3 & 230.2585 & 286.8499& 290.1453\\
\hline
\multirow{1}{*}{12} &  3 & 290.1453 & 354.4353& 357.8266\\
\hline
\multirow{1}{*}{13} &  3 & 357.8266 & 429.9977& 433.4764\\
\hline
\multirow{1}{*}{14} &  3 & 433.4764 & 513.6960& 517.2552\\
\hline
\end{array}\vspace{5mm}
\]
By inspection, $u_{r,k-1}<c_{r,k}<u_{r,k}$ for all irregular pairs.

We have already proved most of Lemma~\ref{moment:lem002}, and
we complete the task below.

\begin{proof}[Proof of Lemma~\ref{moment:lem002}]\
The above analysis shows the existence of constants $c_{r,k}$
for all $r,k\geq 2$ such that $F$ has a unique maximum at
$\vecJ$ whenever $c < c_{r,k}$.
Combining the numerical results for irregular pairs
with
\eqref{opt:eqn054}, \eqref{opt:eqn055}, \eqref{opt:eqn061a}
and~\eqref{opt:eqn061}
shows that $c_{r,k}\in ( u_{r,k-1},\, u_{r,k})$ for all $r,k\geq 2$.

It remains to prove that for all $r, k\geq 2$ we have
\[ c_{r,k} \leq \frac{(k^{r-1}-1)^2}{r(r-1)}.\]
This follows from~\eqref{opt:eqn051} if $k=2$ and $r=2,3,4$,
or from~\eqref{opt:c2k-bound} if $r=2$ and $k\geq 3$.
In all other cases we have $c_{r,k}<(k^{r-1}-1)\ln k$,
from \eqref{opt:eqn061}.  Furthermore, it follows from
Lemma~\ref{app:lem019} that $r(r-1)\ln k/(k^{r-1}-1)<1$ whenever
$k\geq 3$, $r\geq 2$, or $k=2$, $r\geq 5$.
This completes the proof of Lemma~\ref{moment:lem002}.
\end{proof}

Combining this result with the conclusion of Section~\ref{sec:moments},
we see that Theorem~\ref{thm:chrom001} is established.

\subsection{Asymptotics}\label{sec:asymptotics}

We have given precise bounds on $c_{r,k}$, but if we require only
asymptotic estimates as $r\to\infty$ and/or $k\to\infty$ then the following
simplified analysis suffices.
\begin{remark}\label{chrom:rem003}
When we write ``$r\to\infty$ and/or $k\to\infty$'', this is not to be
interpreted as ``$r(n)\to\infty$ and/or $k(n)\to\infty$'', but merely
as ``$r$ and/or $k$ are arbitrarily large constants''. Otherwise, we
cannot use Theorem~\ref{thm:HatMol08} to establish the existence of a
sharp threshold between $c_{r,k}$ and $u_{r,k}$. This is the approach
to asymptotic estimates taken, for example, in~\cite{AchMoo06}.
\end{remark}
We will use~\eqref{opt:eqn037} to improve the estimate of $c_{r,k}$
asymptotically, as discussed in Remarks~\ref{chrom:rem002} and~\ref{chrom:rem005}.
First let us consider the maximum possible improvement that we might
be able to achieve.

From Remark~\ref{chrom:rem005}, we know that the maximum value of $z$ in~\eqref{opt:eqn001} cannot be smaller than that given by taking $t=k$ in~\eqref{opt:eqn037}. Thus
we may bound the possible increase in $c_{r,k}$ as follows.
Since $g(\beta)\leq 1-1/k^{r-1}$ and $t\leq k$, it follows
from~\eqref{opt:eqn038}, using Lemma~\ref{app:lem001},
that $\vartheta\leq 1/(k^{r-1}-1)$ in~\eqref{opt:eqn039}.
Thus $\sum_{i=0}^{\infty}\vartheta^i/(i+1)!\leq 1+\vartheta/2+O(\vartheta^2)$
in~\eqref{opt:eqn039}. Therefore we can increase $c_{r,k}$ asymptotically
by a factor at most $1+1/(2k^{r-1})+ O(1/k^{2r-2})$. Since $c_{r,k}<u_{r,k}=k^{r-1}\ln k$,
the additive improvement to $c_{r,k}$ from fully
optimising~\eqref{opt:eqn001} is at most $\nicefrac12\ln k+O(\ln k/k^{r-1})$. Hence we cannot improve $c_{r,k}$ asymptotically
by more than an additive term $\nicefrac12\ln k$.

Now let us consider what improvement we can rigorously justify.
From Remark~\ref{chrom:rem002}, we know that we can take $t=2$ in~\eqref{opt:eqn037}. Let
$\kappa=4(k-1)/k^r$, and $\calR\,=\,\set{x\in\reals:\,
  1/k^r\, \leq x  \leq \kappa}$.
We proved in Sections~\ref{sec:k=2} -- \ref{sec:general} that the minimum of $\eta(x)$ for $x\in[0,1-\kfr]$ lies in $[(k-1)/k^r,(k+2)/k^r]\subset\calR$ for
all $r, k\geq 2$. However, all we require here is the fact that $\eta$ has a unique minimum in $[0,1-\kfr]$, as shown in Section~\ref{sec:min eta}.
Now we may approximate
\begin{align*}
  f(x)\,&=\,\ln k-x\ln(k-1)+x\ln x-x+O(x^2)\,,\\
  g(x)\,&=\,1-rx-1/k^{r-1}+O(r^2x^2).\,
\intertext{Hence, using Lemma~\ref{app:lem005}, and noting that $-\ln x=O(r\ln k)$ since $x\geq 1/k^r$, in $\calR$ we have}
\eta(x)\,&=\,\big(\ln k-x\ln(k-1)+x\ln x-x+O(x^2)\big)\big(1+rx+1/k^{r-1}+O(r^2x^2)\big)\\
&=\,(1+1/k^{r-1})\ln k+x(r\ln k-\ln(k-1)-1)+x\ln x+O(r^2x^2 \ln k)\,.
\end{align*}
Therefore, let $\varphi$ be the function defined by
\begin{align*}
\varphi(x)\,&=\,(1+1/k^{r-1})\ln k+x(r\ln k-\ln(k-1)-1)+x\ln x\,.
\intertext{We have seen that $\varphi$ approximates $\eta$.  Now}
\varphi'(x)\,&=\,(r\ln k-\ln(k-1)-1)+1+\ln x\,=\,\ln x - \ln\big((k-1)/k^r\big)\,,\\
\varphi''(x)\,&=\,1/x\,>0\,.
\end{align*}
Thus $\varphi(x)$ is minimised at $\hat{\xi}=(k-1)/k^r\in\calR$, as expected.
We can write
\begin{align}
\varphi(x)\,&=\,(1+1/k^{r-1})\ln k-x+x\ln(x/\hat{\xi})\,.\label{opt:eqn062}
\intertext{In particular,}
\varphi(\hat{\xi})\,&=\,(1+1/k^{r-1})\ln k-\hat{\xi}\,.\notag
\intertext{Hence, reinstating the error term in \eqref{opt:eqn040}, we may take}
  c_{r,k}\,&=\,\frac{(k^{r-1}-1)^2}{k^{r-1}}\bigg(\Big(1+\frac{1}{k^{r-1}}\Big)
\ln k-\frac{k-1}{k^r}-O\Big(\frac{r^2\ln k}{k^{2r-2}}\Big)\bigg)\notag\\
&=\,(k^{r-1}-1)\ln k -\frac{k-1}{k} - O\Big(\frac{r^2\ln k}{k^{r-1}}\Big)\,.\label{opt:eqn063}
\end{align}
Since $\kappa=4\hat{\xi}$, using~\eqref{opt:eqn062} we have,
\begin{equation*}
     \varphi(\kappa)-\varphi(\hat{\xi})\,= \,(\hat{\xi}-\kappa)+\kappa\ln(\kappa/\hat{\xi})\,
     =\,-3\hat{\xi}+4\hat{\xi}\ln 4\,>\,2.5\,\hat{\xi}\,=\,2.5\, (k-1)/k^r\,.
\end{equation*}
Therefore, since $\eta$ has a unique minimum in $[0,1-\kfr]$, we have
\begin{align*}
\eta(x)\,&\geq\,\varphi(\hat{\xi})- O\big(r^2\ln k/k^{2r-2}\big)&&(x\leq \kappa)\,\\
\omega(x)\,&\geq\,\varphi(\hat{\xi})+2.5(k-1)/k^r- O\big(r^2\ln k/k^{2r-2}\big)&&(x\geq \kappa)\,.
\end{align*}
We have $g(x)=1-O(r/k^{r-1})$ when $x\leq \kappa$, and hence
$\vartheta=2/k^r-O(r/k^{2r-1})$, taking $t=2$ in~\eqref{opt:eqn038}.
Thus the factor $(e^\vartheta-1)/\vartheta$
in~\eqref{opt:eqn039} is $1+1/k^r-O(r/k^{2r-1})$. This is effectively
the maximum value of $(e^\vartheta-1)/\vartheta$ for $x\in[0,1-\kfr]$
and $(e^\vartheta-1)/\vartheta$ is effectively constant for $x\leq\kappa$.  Thus
\begin{align}
\min_{x\leq \kappa}\Big(\eta(x)\halfpt\frac{e^\vartheta-1}{\vartheta}\Big)\,
&\geq\,\varphi(\hat{\xi})(1+1/k^r)-O(r^2\ln k/k^{2r-2})\notag\\
 &=\,\varphi(\hat{\xi})+\ln k/k^r-O(r^2\ln k/k^{2r-2})\,,\notag\\
\intertext{since $\varphi(\hat{\xi})=\ln k +O(\ln k/k^{r-1})$. Also, $\vartheta>0$ in $[0,1-\kfr]$, so}
\min_{x>\kappa}\Big(\eta(x)\halfpt\frac{e^\vartheta-1}{\vartheta}\Big)\,&\geq\,
\varphi(\hat{\xi})+2.5(k-1)/k^r- O\big(r^2\ln k/k^{2r-2}\big)\label{opt:eqn064}\\
 &>\,\varphi(\hat{\xi})+\ln k/k^r-O(r^2\ln k/k^{2r-2})\,,\notag
\end{align}
for any $k,\,r\geq 2$, provided $k^r$ is large enough. Thus, after
multiplying the right side of~\eqref{opt:eqn064} by $(k^{r-1}-1)^2/k^{r-1}$,
the additive improvement in $c_{r,k}$ is $\ln k/k - O(r^2\ln k/k^{r-1})$. Applying this
to~\eqref{opt:eqn063}, we have
\begin{equation}\label{opt:eqn065}
  c_{r,k}\,=\,k^{r-1}\ln k-\frac{k-1}{k}(1+\ln k)- O\Big(\frac{r^2\ln k}{k^{r-1}}\Big)\,.
\end{equation}
Substituting $k=2$ in~\eqref{opt:eqn065},
\[ c_{r,2}\,=\,2^{r-1}\ln 2-\nicefrac{1}{2}(1+\ln 2)- O(r^2/2^r)\,,\]
the result obtained by Achlioptas and Moore~\cite{AchMoo06} for $2$-colouring $r$-uniform hypergraphs. The case $r=2$ (colouring random graphs), studied by Achlioptas and Naor~\cite{AchNao05}, is discussed further below.

\bigskip

\begin{remark}\label{chrom:rem004}
The best lower bound on $u_{r,k}$ is $\widetilde{u}_{r,k}=u_{r,k}-\nicefrac12\ln k$ from Remark~\ref{chrom:rem001}, so there is a gap
\[ \widetilde{u}_{r,k}-c_{r,k}\,=\,\frac{k-1}{k}+\frac{k-2}{2k}\ln k+O\Big(\frac{r^2\ln k}{k^{r-1}}\Big)\,.\]
Asymptotically, this gap is always nonzero, though extremely small compared
to $c_{r,k}$ or $u_{r,k}$. It is independent of $r$
(up to the error term), and grows slowly
with $k$. It is minimised when $k=2$ and $r\to\infty$. The existence of this
gap merely indicates that the second moment method is not powerful enough to
pinpoint the sharp threshold. We know from Theorem~\ref{thm:HatMol08}
that the threshold lies in $[c_{r,k}, \widetilde{u}_{r,k}]$, although it is possible that it does not converge to a constant as $n\to\infty$. Note that
if we could obtain the maximum
possible correction $\nicefrac12\ln k$, as discussed above,
then the gap would be
approximately $(k-1)/k$, and hence uniformly bounded for all $k,\,r\geq 2$ except $k=r=2$.
\end{remark}

Observe that the asymptotic estimate of $c_{r,k}$
given in~\eqref{opt:eqn065} is not sharp in one case,
namely when $r=2$ and $k\to\infty$. Here the error in~\eqref{opt:eqn065}
is $O(\ln k/k)$, so we have not improved~\eqref{opt:eqn063}. Since this
is the important case of colouring random graphs,
we will examine it separately.

From~\eqref{c2k}, we know that the bound on $c_{2,k}$ from minimising $\eta$ is precisely
\begin{equation}\label{opt:eqn066}
\frac{(k-1)^3}{k(k-2)}\ln(k-1)\,=\,k\ln k -\frac{k-1}{k}(1+\ln k)-\frac{1}{2k}-O\Big(\frac{\ln k}{k^2}\Big)\,.
\end{equation}
The right side of~\eqref{opt:eqn065} is $\varphi(\hat{\xi})+O(\ln k/k)$, so
\eqref{opt:eqn064} still implies that, when $k$ is large enough, we need only
consider $\vartheta(x)$ for $x\in\calR$.
It follows, as above, that the factor
$(e^\vartheta-1)/\vartheta=1+1/k^2-O(1/k^3)$. Thus the additive improvement
in $c_{2,k}$ is $\ln k/k-O(\ln k/k^2)$. Adding this to~\eqref{opt:eqn066}, we have
\begin{equation}\label{opt:eqn067}
  c_{2,k}\,=\,k\ln k-\frac{k-2}{k}\ln k +\frac{2k-1}{2k}-O\Big(\frac{\ln k}{k^2}\Big)\,,
\end{equation}
which marginally improves~\eqref{opt:eqn055} asymptotically. Note that,
taken together,~\eqref{opt:eqn065} and~\eqref{opt:eqn067} exhaust the possibilities
for the manner in which $r$ and/or $k$ can grow large.

\newpage

\section{Appendix: Technical lemmas}\label{sec:tech}

\begin{lemma}\label{app:lem007}
$G\in \calG^*(n,r,cn)$ has at least $(k-1)$ isolated vertices \aas .
\end{lemma}
\begin{proof}
Define $m=\lfloor cn\rfloor$ and let $Y(\vecv)$ be the number of isolated vertices in $G$, determined by $\vecv$. The $mr$ entries of $\vecv$ are uniform on $[n]$, from which it follows that $\E[Y]=n(1-1/n)^{mr}\sim ne^{-cr}$. Also, the entries of $\vecv$ are independent, and arbitrarily changing any single entry can only change $Y(\vecv)$ by $\pm 1$. Thus we may apply a standard martingale inequality~\cite[Corollary 2.27]{JaLuRu00} to give
\[ \Pr\big( Y\leq \half ne^{-cr}\big)\,\leq\,e^{-ne^{-2cr}/12cr}\,,\]
for large $n$. Thus $G$ has $\Omega(n)$ isolated vertices \aas, from which the result follows easily.
\end{proof}

\begin{lemma}\label{app:lem018}
Suppose that $M$ is a $p\times p$ matrix of $q\times q$ blocks, such that
\[ M=\begin{bmatrix} 2B & B &\cdots & B\\ B & 2B &\cdots & B\\
 B & B &\ddots & B\\B & B &\cdots & 2B
 \end{bmatrix}\,\in\,\reals^{pq\times pq},\ \ \,\mbox{where}\ \
  B=\begin{bmatrix} 2 & 1 &\cdots & 1\\ 1 & 2 &\cdots & 1\\
 1 & 1 &\ddots & 1\\1 & 1 &\cdots & 2
 \end{bmatrix}\,\in\, \reals^{q\times q}\,.\]
 Then $\det(M)=(p+1)^q(q+1)^p$.
\end{lemma}
\begin{proof}
We have, by adding and subtracting rows and columns of $M$,
\begin{align*}
\det M &=\det\begin{bmatrix} 2B & -B &\cdots & -B\\ B & B &\cdots & 0\\
 B & 0 &\ddots & 0\\B & 0 &\cdots & B  \end{bmatrix}
=\det\begin{bmatrix} (p+1)B & 0 &\cdots & 0\\ B & B &\cdots & 0\\
 B & 0 &\ddots & 0\\B & 0 &\cdots & B  \end{bmatrix}\\[1ex]
&= \det\big( (p+1)B\big)(\det B)^{p-1}\, =\, (p+1)^q(\det B)^q.
\end{align*}
We can use the same transformations to compute $\det B$, replacing
$B$ by the $1\times 1$ unit matrix in the argument.
We obtain $\det B = (q+1)\,1^{q-1}=q+1$. Hence $\det M=(p+1)^q(q+1)^p$.
\end{proof}

(We are grateful to Brendan McKay for pointing out that $(p+1)^q(q+1)^p$
is the number of spanning trees in the complete bipartite graph
$K_{p+1,q+1}$.  This suggests that an alternative proof of the above lemma
may be possible using Kirchhoff's Matrix Tree Theorem, but we do not
explore this here.)

\begin{lemma}
\label{app:feasible}
If $\rho\in [1,k^{r-1}]$ then the system defined by
\eqref{opt:eqn003}--\eqref{opt:eqn006} is feasible.
\end{lemma}
\begin{proof}
Firstly, note that the system \eqref{opt:eqn004}--\eqref{opt:eqn006} defines a convex set.
The $k\times k$ matrix $\vecJ$ with all entries
equal to $1/k$ is feasible when $\rho=1$, while any $k\times k$ permutation
matrix $\vecDelta$ is  feasible when $\rho = k^{r-1}$. Now define
the $k\times k$ matrices $\vecA(\epsilon)=(1-\epsilon)\vecJ_0+\epsilon\vecDelta$
for all $\epsilon\in [0,1]$. Then $\vecA(\epsilon)$ satisfies
\eqref{opt:eqn004}--\eqref{opt:eqn006} by convexity, while
\eqref{opt:eqn003} becomes
{\setlength{\abovedisplayskip}{5pt}
\[ \Psi(\epsilon)=k^{r-2}\sum_{i=1}^k \sum_{j=1}^k \big((1-\epsilon)\vecJ_{ij} + \epsilon\vecDelta_{ij}\big)^r\,=\,\rho.\]}%
Now  $\Psi(\epsilon)$ is a polynomial function of $\epsilon$, and hence continuous. Also $\Psi(0)=1$ and $\Psi(1)=k^{r-1}$.  Therefore, by the Intermediate Value Theorem, for any $\rho\in [1,k^{r-1}]$ there
is some $\epsilon^*\in[0,1]$ such that $\Psi(\epsilon^*)=\rho$, and hence $\vecA(\epsilon^*)$ is a feasible solution.
\end{proof}

\begin{lemma}\label{app:boundary}
Let $\ell \in \{ 1,2,\ldots, k\}$ and let $\widehat{\rho}\halfpt \leq\halfpt \ell$ be
a fixed positive constant. Consider the maximisation problem
\begin{subequations}
\label{appopt:eqn001}
\begin{align}
    \textrm{maximise} \ z(\veca)\halfpt=\halfpt -\frac{1}{k}\sum_{i=1}^\ell \sum_{j=1}^ka_{ij}& \ln a_{ij}
\label{appopt:eqn002}\\
\textrm{subject to}\hspace*{14mm} \sum_{i=1}^{\ell} \sum_{j=1}^k a_{ij}^r \,&=\, \rhohat, \label{appopt:eqn003}\\
    \sum_{j=1}^k a_{ij} \,&=1\,\qquad(i\in[\ell]),\label{appopt:eqn004}\\
    a_{ij}&\geq 0\qquad(i\in[\ell],j\in[k])\, .\label{appopt:eqn005}
\end{align}
\end{subequations}
If $\widehat{\rho} < \ell$ 
then no boundary point of \eqref{appopt:eqn004}--\eqref{appopt:eqn005} can be a local maximum of $z$.
\end{lemma}
\begin{proof}
Suppose that $\vecb$ is a local optimum to \eqref{appopt:eqn001} on the boundary of \eqref{appopt:eqn004}--\eqref{appopt:eqn005}.  Without loss of generality, we may
assume that $1\geq b_{i1}\geq b_{i2}\geq \cdots\geq b_{ik}\geq 0$ for all $i\in[\ell]$.
If $\rhohat<\ell$ then there exist $i\in[\ell],\,j_1,j_2\in [k]$ such that $j_1\neq j_2$  and $0<b_{ij_1},b_{ij_2}<1$.
Without loss of generality, suppose that $i=1$, $j_1=1,j_2=2$.
Since $\vecb$ is on the boundary, there exist $t\in[\ell],j\in [k]$ such that $b_{tj}=0$. Without loss of generality, we may assume that $t\in\{1,2\}$ and $j=k$.

Let $\calS$ denote the region given by \eqref{appopt:eqn004}--\eqref{appopt:eqn005}, such that $a_{ij}=b_{ij}$ for
$(i,j)\notin\{(1,1),(1,2),(1,k)\}$ if $t=1$, and for
$(i,j)\notin\{(1,1),(1,2),(2,1),(2,k)\}$ if $t=2$. Then let $\calS'$ be the
subset of $\calS$ determined by \eqref{appopt:eqn003}, and let $\calS^o$ denote the interior of $\calS$. Let $\Phi(\veca)=\sum_{i=1}^\ell \sum_{j=1}^k a_{ij}^r$. Thus \eqref{appopt:eqn003} is equivalent to $\Phi(\veca)=\Phi(\vecb)=\rhohat$.

At the point $\vecb$, note that $\partial z/\partial a_{ij}$ is finite for all $b_{ij}>0$ and $+\infty$ for all $b_{ij}=0$. Thus, for all small enough $\delta>0$, there is a ball $B$ with centre $\vecb$ and radius $\delta$, such that $z(\veca)>z(\vecb)$ for every point $\veca\in B'$, where $B'=  B\cap\calS^o$.  Note that  $B'$ is a convex set. So, we need only show that there is a point in $S'\cap B'$, since this will contradict the assumption that $\vecb$ is a local maximum of $z$.

Let us write the points in $\calS$ as $\vecu=(a_{11},a_{12},a_{1k})$  if $t=1$,
or as $\vecu=(a_{11},a_{12},a_{21},a_{2k})$ if $t=2$. Let
\begin{equation*}
\begin{array}{lll}
\vecu_0=  (b_{11}+\theta-\theta^3,b_{12}-\theta,\theta^3),\quad&\quad\vecu_1=(b_{11}-\theta,b_{12},\theta)\quad&\mbox{if } t=1,\\[0.5ex]
 \vecu_0=(b_{11}+\theta,b_{12}-\theta,b_{21}-\theta^3,\theta^3),\quad&\quad\vecu_1=
   (b_{11},b_{12},b_{21}-\theta,\theta)\quad&\mbox{if } t=2.
\end{array}
\end{equation*}
Then $\vecu_i\in\calS^o$ and
$\|\vecu_i-\vecb\|\leq 3\theta$ for $\theta\in (0,1)$ and $i=1,2$. Thus, for small enough $\theta$,  $\vecu_i\in B'$ ($i=1,2$). Also
\begin{equation*}
  \Phi(\vecu_0)-\Phi(\vecb)\,=\,r(b_{11}^{r-1}-b_{12}^{r-1})\theta+
\tfrac12r(r-1)(b_{11}^{r-2}+b_{12}^{r-2})\theta^2+O(\theta^3)>0\quad (t=1,2)
\end{equation*}
for small enough (positive) $\theta$, since $b_{11}\geq b_{12}>0$ and $r\geq2$. Thus $\Phi(\vecu_0)>\Phi(\vecb) = \widehat{\rho}$. Similarly
\begin{equation*}
  \Phi(\vecu_1)-\Phi(\vecb)\,=\,-r b_{t1}^{r-1}\theta+O(\theta^2)\,< 0\,\quad (t=1,2),
\end{equation*}
for small enough $\theta$, since $b_{t1}>0$ $(t=1,2)$ and $r\geq 2$. Thus $\Phi(\vecu_1)<\Phi(\vecb) = \widehat{\rho}$.

Now consider the points $\vecu_\epsilon=(1-\epsilon)\vecu_0+\epsilon\vecu_1$, for $\epsilon\in[0,1]$. By convexity, $\vecu_\epsilon\in B'$ for all $\epsilon\in [0,1]$.
Also $\Phi(\vecu_\epsilon)$ is a polynomial function of $\epsilon$ with $\Phi(\vecu_0)>\rhohat$ and $\Phi(\vecu_1)<\rhohat$. Hence, by the Intermediate Value Theorem,
there exists $\epsilon^*\in[0,1]$ such that $\Phi(\vecu_{\epsilon^*})=\rhohat$. Then $\vecu_{\epsilon^*}$ is the required point in $\calS'\cap B'$.
\end{proof}

\begin{lemma}\label{app:lem001}
  $\ln(1+z)\leq z$ for all $z>-1$.
\end{lemma}
\begin{proof}
Let $\phi(z)=z-\ln(1+z)$, which is strictly convex on $z>-1$, since $\ln(1+z)$
is strictly concave. Also $\phi'(z)=1-1/(1+z)$, so $\phi$ is stationary at $z=0$,
and this must be its unique minimum. Since $\phi(0)=0$, we have $\phi(z)\geq 0$
for all $z>-1$, and $\phi(z)> 0$ if $z\neq 0$.
\end{proof}
\begin{lemma}\label{app:lem002}
  $\ln(1-z)\geq -3z/2$ for all $0\leq z \leq \nicefrac12$.
\end{lemma}
\begin{proof}
Let $\phi(z)=\ln(1-z)+3z/2$. Then $\phi$ is strictly concave on $[0,1)$,
since $\ln(1-z)$ is strictly concave. Also $\phi'(z)=-1/(1-z)+\nicefrac{3}{2}$,
so $\phi$ is stationary at $z=\nicefrac{1}{3}$, and this must be its unique
maximum. Now $\phi(0)=0$, and we may calculate $\phi(\nicefrac12)>0$,
so $\phi(z)>0$ for $0< z \leq \nicefrac12$.
\end{proof}

\begin{lemma}\label{app:lem004}
  For all $z\in(0,1)$, $(1-z)\ln(1-z)>-z$ and $(1-\tfrac12 z)\ln(1-z)< -z$.
\end{lemma}
\begin{proof}\ \vspace{-\baselineskip}
We have
\begin{align*}
  (1-z)\ln(1-z)\,& =\,-z+\sum_{i=2}^\infty\frac{z^i}{i(i-1)}\,>\,-z, \\
  (1-\tfrac12 z)\ln(1-z)\,& =\,-z-\sum_{i=3}^\infty\frac{(i-2)z^i}{2i(i-1)}\,<\,-z.&&\qedhere
\end{align*}
\end{proof}

\begin{lemma}\label{app:lem005}
  $1+z\leq 1/(1-z)\leq 1+z+2z^2\leq 1+2z$ for all $0\leq z \leq \nicefrac12$.
\end{lemma}
\begin{proof}
The first inequality is equivalent to $z^2\geq 0$ if $z<1$. The second
inequality is equivalent to $z\leq \nicefrac12$. The third follows trivially from the second.
\end{proof}

\begin{lemma}\label{app:lem006}
  For $p\in\nats$, $z\in[0,1]$, $1-p z\leq (1-z)^p \leq 1-p z+\half(p z)^2$. Also $(1-z)^p \leq 1/(1+p z)$.
\end{lemma}
\begin{proof}
Let $\phi_1(z)=(1-z)^p -1+p z$. Then $\phi_1(0)=0$ and $\phi'_1(z)=p (1-(1-z)^{p -1})\geq0$ if $z\in[0,1]$,
giving the first inequality. Let $\phi_2(z)=1-p z+\half(p z)^2-(1-z)^p$.
Then $\phi_2(0)=0$ and $\phi'_2(z)=-p + p^2z+ p(1-z)^{p -1}\geq - p + p^2z+ p(1-(p-1)z)= pz\geq0$,
by the first inequality, giving the second. For the third inequality,
using Lemma~\ref{app:lem001}, we have $(1-z)^p  \leq e^{-p z} = 1/e^{p z} \leq 1/(1+ pz)$.
\end{proof}

\begin{lemma}\label{app:lem008}
Let $\eta(x) = f(x)/g(x)$ for $x\in [0,1-1/k)$, where
\begin{equation*}
f(x) = \ln k - x\ln(k-1)+(1-x)\ln(1-x) +x\ln x,
\qquad
g(x) = (1-x)^r+x^r/(k-1)^{r-1}-1/k^{r-1}.
\end{equation*}
Then
\begin{alignat*}{3}
\eta(0)\,&=\,\frac{k^{r-1}\ln k}{k^{r-1}-1},&\qquad\quad
  \lim_{x\to 1-1/k}\eta(x)\,=\,\lim_{x\to 1 - 1/k}\omega(x)
&= \frac{k^{r-1}}{r(r-1)},\\
\lim_{x\to0}\eta'(x)\,&=\,-\infty,&\qquad\quad
 \lim_{x\to 1-1/k}\eta'(x)\,&=\,\frac{(k-2)k^r}{r(k-1)}.
\end{alignat*}
\end{lemma}
Furthermore, if $\eta'(x)=0$ and $g(x)\neq 0$ then $\omega(x) = \eta(x)$.

\begin{proof}
The stated value of $\eta(0)$ follows from the definition.
Recall the calculations of Section~\ref{sec:min eta}.
Using L'H\^{o}pital's rule~\cite{Spivak06},
\begin{equation}\label{app:eqn001}
\lim_{x\to 1-1/k}\eta(x)\,=\,\left[\frac{f''(x)}{g''(x)}\right]_{1-1/k}\,=\,
\frac{k^{r-1}}{r(r-1)}.
\end{equation}
The same calculations prove that $\lim_{x\to 1-1/k}\omega(x)$ also
takes this value.
Next,
\begin{equation}\label{app:eqn002}
  \eta'(x)\,=\,\frac{g(x)f'(x)-f(x)g'(x)}{g(x)^2}\,=
  \,\frac{f'(x)-\eta(x)\, g'(x)}{g(x)}.
\end{equation}
As $x\to0$, all quantities in \eqref{app:eqn002} are finite, except $f'(x)\to-\infty$. Since $g(0)>0$, we have
$\eta'(x)\to-\infty$ as $x\to0$.
Note also that the last statement of the lemma follows from \eqref{app:eqn002}.

For the final calculation note that when $x=1-\kfr$, the
numerator and denominator in the
expression for $\eta'(x)$ are both zero. Hence
applying L'H\^{o}pital's rule again gives
\begin{align*}
  \lim_{x\to 1-1/k}\eta'(x)\,
   =\,\lim_{x\to 1-1/k}\frac{f'(x)-\eta(x)\,g'(x)}{g(x)}
  \,&=\,\lim_{x\to 1-1/k}\frac{f'(x)-\eta''(1-\kfr)\, g'(x)}{g(x)}\\
&= \lim_{x\to 1-1/k}\frac{f'''(x)-\eta(1-\kfr)\, g'''(x)}{g''(x)}\\
  &=\,\frac{(k-2)k^r}{r(k-1)},
\end{align*}
using the values of $f''(1-\kfr)$, $f'''(1-\kfr)$, $g''(1-\kfr)$ and $g'''(1-\kfr)$ calculated in Section~\ref{sec:min eta}.
\end{proof}

\begin{lemma}\label{app:lem009}
Let $k\geq 2$ and $r\geq 2$ be integers, and let $\lambda$ be a positive real number.
If $r\leq 2k$ then the equation
\begin{equation*}
 x(1-x)^{r-1}+\frac{(1-x)x^{r-1}}{(k-1)^{r-1}}\,=\,\frac{1}{\lambda r(r-1)}
\end{equation*}
has at most two solutions for $x$ in $[0,1-\kfr]$.
Otherwise $r\geq 2k+1$ and the above equation has at most
two solutions for $x$ in $[0,1-\kfr]$
whenever $\lambda < \lambda_0$, where
\begin{equation*}
\frac{1}{\lambda_0}\ =\ r(r-1)\left(\frac{(r-2)2^{r-1}}{r^r}+\frac{1}{k^r}\right).
\end{equation*}
\end{lemma}

\begin{proof}
Let $\theta(x)=x(1-x)^{r-1}$ and define $\kappa$, $\ell$
by $1/\kappa=(k-1)^{r-1}$ and
$1/\ell=\lambda r(r-1)$.  We wish to investigate the number of
solutions of $\phi(x)= \ell$, where
\begin{align*}
  \phi(x)\, & =\,\theta(x)+\kappa\halfpt\theta(1-x).
\intertext{Differentiating gives}
  \phi'(x)\,&=\,\theta'(x)-\kappa\halfpt\theta'(1-x),
 &\phi''(x)\,&=\,\theta''(x) +\kappa\halfpt\theta''(1-x).
\intertext{Thus the stationary points of $\phi$ are the solutions of $\theta'(x)=\kappa\halfpt\theta'(1-x)$.
We may calculate}
  \theta'(x) \,& =\,(1-x)^{r-2}(1-rx),& \theta'(1-x)\,& =\,-x^{r-2}\big((r-1)-rx\big),\\
\theta''(x)\,&=\,-(r-1)(1-x)^{r-3}(2-rx), & \theta''(1-x)\,&=\,(r-1)x^{r-3}\big((r-2)-rx\big).
\end{align*}

We summarise the behaviour of $\phi$ in $[0,1]$ in the following table.
Here $\downarrow$ means ``decreasing'', $\uparrow$ means ``increasing''. The final
column gives the maximum number of stationary points of $\phi$ in the corresponding subinterval of $[0,1]$.

{\renewcommand{\arraystretch}{1.2}%
\begin{tabular}{@{\qquad$x\,\in\,$}l|@{\quad}c@{\quad}|@{\quad}c@{\quad}|@{\quad}c}
  $[0,1/r)$  & $\theta'(x)>0$,\ \,$\theta'(1-x)<0$ & $\phi(x)$ $\uparrow$ & 0 \\
  $[1/r,2/r)$ & $\theta''(x)\leq0$,\ \, $\theta''(1-x)>0$ & $\theta'(x)$ $\downarrow$ \,$\kappa\theta'(1-x)$ $\uparrow$ & $1$ \\
   $(2/r,1-2/r]$ & $\theta''(x)>0$,\ \, $\theta''(1-x)\geq0$ & $\phi(x)$ strictly convex & $1$ \\
  $(1-2/r,1-1/r]$ & $\theta''(x)>0$,\ \, $\theta''(1-x)\leq0$ & $\theta'(x)$
$\uparrow$ \,$\kappa\theta'(1-x)$ $\downarrow$ & $1$ \\
  $(1-1/r,1]$  & $\theta'(x)<0$,\ \,$\theta'(1-x)>0$ & $\phi(x)$ $\downarrow$ & 0
\end{tabular}}

Now $\phi$ is analytic on $[0,1]$,
with $\phi'(0)=1$ and $\phi'(1)=-\kappa$.
Therefore $\phi'$ changes sign an odd number of times in $[0,1]$,
which implies that
$\phi$ has an odd number of stationary points
in $[0,1]$. From the table it follows that $\phi$ has either
one or three stationary points.
Hence $\phi(x)= \ell$ has at most four solutions in $[0,1]$,
for any $\ell$.

We first consider small values of $r$. When $r=2$ the union of
the first and
last subinterval is $[0,1] \setminus \{ \nicefrac12 \}$, which contains no
stationary point.  Hence $\phi$ has at most one stationary point
in $[0,1]$ (and it can only occur at $x=\nicefrac12$).
When $r=3$ the union of the first, second and last
subinterval equals $[0,1] \setminus \{ \nicefrac23\}$ and contains at most
one stationary point.   Hence $\phi$ has at most two stationary
points in $[0,1]$.

When $r=4$, the central subinterval is empty, so $\phi$ has at most
two stationary points in $[0,1]$.
However, we know that an even number of stationary points is
impossible, from above.
Therefore when $r=2,\,3,\,4$ the function $\phi$ has at most one stationary point in
$[0,1]$, and hence at most two solutions to $\phi(x)= \ell$
in $[0,1]$, for any fixed $\ell$.

Next we assume that $r\geq 5$, which implies that all five subintervals are
nonempty.  Either $\phi$ has one stationary point
which is a local maximum, or it has three stationary points:
a local maximum $\mu_1$, a local minimum $\mu_2$, and a local maximum $\mu_3$,
with $\mu_1<\mu_2<\mu_3$.

Let
\[ L_1 = \sup\{ \phi(y) : y\in [1/r,2/r)\},\qquad
  L_2 = \sup\{ \phi(z) :  z\in (1-2/r,1-\kfr]. \]
(We take $L_2=-\infty$ if there is only one stationary point.)
First we show that
\begin{equation}
\label{app:eqn003}
L_1 \geq L_2.
\end{equation}
We readily see that
\begin{equation*}
  L_1 \geq \phi(1/r) \,>\, \theta(1/r) \,=\, \frac{(r-1)^{r-1}}{r^r}.
\end{equation*}
Next we calculate an upper bound on $L_2$ by considering two cases.
First, if $2\leq r\leq k$ then
\begin{align*}
  L_2\ &\leq \ \theta(1-2/r)+\kappa\halfpt\theta(1/r)\ =\
  \frac{1}{r^r}\Big((r-2)2^{r-1}+\Big(\frac{r-1}{k-1}\Big)^{r-1}\Big)\\
     &\leq\  \frac{(r-2)2^{r-1}+1}{r^r}\ \leq\  \frac{(r-2)2^{r-1}+(r/2)^r}{r^r},
\end{align*}
 since $\theta(1-x)$ is maximised at $x=1-1/r$ in $[0,1]$.
Next, if $r>k\geq 2$ then
$\theta(1-x)$ is maximised when
$x=1/k$ in $[0,1-\kfr]$. Therefore
\begin{align}
  L_2 \ &\leq \ \theta(1-2/r)+\kappa\halfpt\theta(1/k) \
   =\ \frac{(r-2)2^{r-1}}{r^r} +\frac{1}{k^r}
\leq\ \frac{(r-2)2^{r-1}+(r/2)^r}{r^r}.\label{app:eqn004}
\end{align}
Thus \eqref{app:eqn003} holds  if $(r-2)2^{r-1}+(r/2)^r<(r-1)^{r-1}$.
We show in Lemma~\ref{app:lem010} that this is true for all $r\geq 5$, so
\eqref{app:eqn003} holds for $r\geq 5$.
Now $\phi$ has at least one local
maximum, so we have established that $\phi$ has a local maximum
$\mu_1\in [1/r,2/r)$  whenever $r\geq 5$.

We now consider whether $\phi$ has a local minimum $\mu_2\in (2/r,1-2/r)$.
Since there is a local maximum $\mu_1\in [1/r,2/r)$ we know that
$\phi'(2/r)<0$. Thus $\phi$ has a local minimum $\mu_2\in (2/r,1-2/r]$
 if and only if $\phi'(1-2/r)>0$.  Now
\begin{equation*}
  \phi'(1-2/r)\,=\,\theta'(1-2/r)-\kappa\halfpt\theta'(2/r) \,
=\, -\frac{r-3}{r^{r-2}}\Big(2^{r-2}-\frac{1}{(r-3)(k-1)}
           \Big(\frac{r-2}{k-1}\Big)^{r-2}\Big).
\end{equation*}
This expression is certainly nonpositive if $2\geq (r-2)/(k-1)$;
that is, if $r\leq 2k$.  So, if $r\leq 2k$, there is no local minimum
in $(2/r,1-2/r)$ and it follows that  $\mu_1$ is the only stationary point
of $\phi$.  In this case, the equation $\phi(x)= \ell$ has at most two
solutions on $[0,1]$.

When $r\geq 2k+1\geq 5$ we know that
$\phi(x)= \ell$ has at most two solutions for all $\ell > L_2$,
using \eqref{app:eqn003}.
From \eqref{app:eqn004} we have
\begin{equation*}
  L_2\,<\, \frac{(r-2)2^{r-1}}{r^r}+\frac{1}{k^r}.
\end{equation*}
Substituting $\ell = (\lambda r(r-1))^{-1}$ we find that
$\phi(x)= \ell$ has at most two solutions in $[0,1-1/k]$ so long as
\[ \lambda\,>\,\frac{1}{r(r-1)\halfpt L_2}\,>\,
    \frac{1}{r(r-1)}\left(\frac{(r-2)2^{r-1}}{r^r}+\frac{1}{k^r}\right)^{-1}\,=\,\lambda_0.\qedhere \]
\end{proof}

\begin{lemma}\label{app:lem010}
For all $r\geq 5$ the inequality $(r-2)2^{r-1}+(r/2)^r<(r-1)^{r-1}$ holds.
\end{lemma}

\begin{proof}
We will show $(r-2)2^{r-1}<(r-1)^{r-1}/2$ and $(r/2)^r<(r-1)^{r-1}/2$.

To show $(r-2)2^{r-1}<(r-1)^{r-1}/2$,
let $\gamma_1(r)=2(r-2)2^{r-1}/(r-1)^{r-1}$. Then
\[ \frac{\gamma_1(r+1)}{\gamma_1(r)}\,
=\,\frac{2}{r-2}\Big(\frac{r-1}{r}\Big)^r\,<\,1\]
if $r\geq 4$. Thus $\gamma_1(r)$ is decreasing for
$r\geq 4$. Since $\gamma_1(5)=\nicefrac38<1$, the inequality follows.

To show $(r/2)^r<(r-1)^{r-1}/2$, let $\gamma_2(r)=r^r/(2r-2)^{r-1}$. Then
\[ \frac{\gamma_2(r+1)}{\gamma_2(r)}\,
 =\,
   \frac{r+1}{2r-2}\, \Big(\frac{r^2-1}{r^2}\Big)^r
\leq \,\Big(\frac{r^2-1}{r^2}\Big)^r<\,1\]
if $r\geq 4$. Thus $\gamma_2(r)$ is decreasing for
$r\geq 4$. Since $\gamma_2(5)=5^5/2^{12}<1$, the inequality follows.
\end{proof}

\begin{lemma}\label{app:lem011}
For $k\geq 2$ and $r\geq 2k+1$, we have
\[ \eta(0)  <\,\eta(1-\kfr) \,<\,\lambda_0.
\]
(The values of $\eta(0)$ and $\eta(1-\kfr)$ are stated in
Lemma~\ref{app:lem008} while $\lambda_0$ is defined in Lemma~\ref{app:lem009}.)
\end{lemma}

\begin{proof}
The left hand inequality reduces to $r(r-1)\ln k/(k^{r-1}-1)<1$.
In Lemma~\ref{app:lem019} we show that
$r(r-1)\ln k/(k^{r-1}-1)<1$ for all $k\geq 3$, $r\geq 2$, or $k=2$, $r\geq 5$.
Clearly this includes
all $r\geq 2k+1$, and so establishes the left hand inequality.

The right hand inequality is
\[\frac{(r-2)2^{r-1}}{r^r}+\frac{1}{k^r}\ <\ \frac{1}{k^{r-1}},\]
which is equivalent to $\gamma(r,k) < 1$, where
\[\gamma(r,k)\ =\ \frac{r-2}{2k-2}\left(\frac{2k}{r}\right)^r\ .\]
For fixed $k\geq 2$, if $r>2k$ then
\[ \frac{\gamma(r+1,k)}{\gamma(r,k)}\,=
\,\frac{2k(r-1)r^r}{(r-2)(r+1)^{r+1}}\,\leq
\,\frac{2kr(r-1)}{(r-2)(r+1)^2}\,\leq\,\frac{2k}{r}\,<\,1,\]
if $r^2(r-1)\leq (r-2)(r+1)^2$. This is equivalent to $r^2-3r-2\geq 0$, which is true for all $r\geq 4$.
Thus $\gamma(r,k)$ is decreasing in $r$, so we need only establish the critical case $r=2k+1$. We have
\[\gamma(2k+1,k)\ =\ \frac{2k-1}{2k-2}\left(\frac{2k}{2k+1}\right)^{2k+1}\
\leq\ \frac{2k-1}{2k-2}\left(\frac{2k}{2k+1}\right)^2\ <\ 1,\]
if $(2k-2)(2k+1)^2- (2k-1)(2k)^2> 0$, which is $2k^2-3k-1>0$. This holds for all $k\geq 2$.
\end{proof}

\begin{lemma}\label{app:lem012}
If $k=2$ then $x=\nicefrac12$ is a local minimum of $\eta$ for $r=2,\,3,\,4$,
and a local maximum if $r\geq 5$.
\end{lemma}

\begin{proof}
 We have
\[\eta(x)\,=\,\frac{\ln 2+x\ln x+(1-x)\ln(1-x)}{(1-x)^r+x^r-1/2^{r-1}}.\]
Substituting $x=(1-z)/2$, we find
\[\frac{\eta(z)}{2^{r-1}}\,=\,\frac{(1-z)\ln(1-z)+(1+z)\ln(1+z)}{(1-z)^r+(1+z)^r-2}.\]
We may compute Taylor expansions, giving
\begin{align*}
  \frac{r(r-1)\eta(z)}{2^{r-1}} \,& =\,\frac{z^2+z^4/6+O(z^6)}{z^2+(r-2)(r-3)z^4/12+O(z^6)} \\
   & =\,1+\frac{2-(r-2)(r-3)}{12}\halfpt z^2+O(z^4).
\end{align*}
If $r=2,3$ then the coefficient of $z^2$ is positive, so $z=0$ is a local minimum. If $r\geq 5$ then the coefficient
of $z^2$ is negative, so $z=0$ is a local maximum. However, if $r=4$, the
coefficient of $z^2$ is zero, so we need a higher order approximation. We compute
\begin{align*}
  \frac{3\eta(z)}{2} \,& =\,\frac{z^2+z^4/6+z^6/15+O(z^8)}{z^2+z^4/6} \,
    =\,\frac{1+z^2/6+z^4/15+O(z^6)}{1+z^2/6}\,
    =\,1+\frac{z^4}{15}+O(z^6).
\end{align*}
The coefficient of $z^4$ is positive, and hence $z=0$ is a local minimum.
\end{proof}

\begin{lemma}\label{app:lem013}
The function $r^2(k+2)/k^r$ is decreasing in both $r$ and $k$ for all
$r\geq 3,\,k\geq 2$. Hence $r^2(k+2)/k^r < 1$ if
\[ k=2,\ r\geq 9,\quad k=3,\ r\geq4,\quad k\geq4,\ r\geq 3.\]
\end{lemma}
\begin{proof}
Let $\phi(r,k)=r^2(k+2)/k^r$. Then
\[ \frac{\phi(r+1,k)}{\phi(r,k)}\,=\,
\frac{(r+1)^2}{kr^2}\,<\,1,\]
if $k\geq (1+1/r)^2$. Since $(1+1/r)^2\leq \nicefrac{16}{9}$ for $r\geq 3$, this is satisfied for all $k\geq 2$. Also
\[ \frac{\phi(r,k+1)}{\phi(r,k)}\,=\,
\frac{(k+3)k^r}{(k+2)(k+1)^r}\,<\,\frac{(k+3)k}{(k+2)(k+1)}\,=\,\frac{k^2+3k}{k^2+3k+2}\,<\,1.\]
We can now check numerically that $r^2(k+2)/k^r < 1$ for
$(k,r)\in \{(2,9),\, (3,4),\, (4,3)\}$.
\end{proof}

\begin{lemma}\label{app:lem014}
If the inequality
\[ \frac{3k^r}{r^2(k+2)^2\ln k}\,\geq\,1+\frac{0.52}{r}+\frac{2}{r(k+2)}+\frac{3}{2r\ln k}\]
holds for some $(r,k)$ with $r\geq 3$, $k\geq 2$,
then it holds for all $(r',k')$ such that $r'\geq r$, $k'\geq k$.
\end{lemma}
\begin{proof}
The right side of this inequality is decreasing with $r$ and $k$, so it
suffices to show that the function $\phi(r,k)$ on the left side is increasing.
This follows since, if $k\geq 2,\,r\geq 3$,
\[ \frac{\phi(r+1,k)}{\phi(r,k)}\,=\,
\frac{kr^2}{(r+1)^2}\,\geq\,1.\]
Also, if $r\geq 3$ then
\[ \frac{\phi(r,k+1)}{\phi(r,k)}\,=\,
\frac{(k+1)^r (k+2)^2\ln k}{k^r(k+3)^2\ln(k+1)}
\geq\,\frac{(k+1)^3(k+2)^2\ln k}{k^3(k+3)^2\ln(k+1)}\,>\,
\frac{(k+1)\ln k}{k\ln(k+1)},\]
since $(k+1)(k+2)>k(k+3)$ for all $k\geq 0$. Now we will have
\[ \frac{(k+1)\ln k}{k\ln(k+1)}\,=\,\frac{(k+1)/\ln(k+1)}{k/\ln k}\,>\,1\]
if the function $\gamma(x)=x/\ln x$ is increasing for $x\geq k$. Since
$\gamma'(x)=(\ln x-1)/(\ln x)^2>0$ for $x>e$, we have $\phi(r,k+1)/\phi(r,k)>1$
for $k\geq 3$. For $k=2$ and $r\geq 3$, we may verify that
\[ \frac{\phi(r,3)}{\phi(r,2)}\,=\,
   \frac{16\ln 2}{25\ln 3}\Big(\frac{3}{2}\Big)^r\,
\geq\,\frac{54\ln 2}{25\ln 3}\,>\,1.\]
Thus $\phi(r,k)$ is increasing in $k$ and $r$ for all $k\geq 2,\,r\geq 3$, and the conclusion follows.
\end{proof}

\begin{lemma}\label{app:lem015}
For all regular pairs, $r(r\ln k+1)\xi \leq 1$.
\end{lemma}
\begin{proof}
We have $\xi \leq (k+2)/k^r$ for all regular pairs.
Thus the inequality is true if $\phi(r,k)\leq 1$, where
$\phi(r,k) = r(r\ln k+1)(k+2)/k^r$.
Now
\[\phi(r,k)\,=\,\frac{r(r\ln k+1)}{k^2}\cdot\frac{k+2}{k^{r-2}}
\,=\,r\brac{\frac{r\ln k}{k^2}+\frac{1}{k^2}}\brac{\frac{1}{k^{r-3}}+\frac{2}{k^{r-2}}}\]
is decreasing with $k\geq 2$ for all $r\geq 3$, since $\ln k/k^2$ is
decreasing for $k\geq 2$. Also
\[\frac{\phi(r+1,k)}{\phi(r,k)}\,=\,\frac{1}{k}\brac{1+\frac{1}{r}}\brac{1+\frac{\ln k}{r\ln k+1}}\,<\,\frac1k\brac{1+\frac{1}{r}}^2\,\leq\,
\frac89\,<\,1\]
if $r\geq 3,\,k\geq 2$.
Thus $\phi(r,k)$ is decreasing with $r\geq 3$ for $k\geq 2$.
Direct calculation shows that
$\phi(9,2)$,\, $\phi(6,3)$,\, $\phi(5,4)$,\, $\phi(4,5)$,\, $\phi(3,15)$
are all less than 1.
Thus $r(r\ln k+1)(k+2)/k^r < 1$ for all regular pairs.
\end{proof}

\begin{lemma}\label{app:lem016}
For all regular pairs,  $\ln k -2(k+2)/k^r>\ln(k-1)$.
\end{lemma}
\begin{proof}
Using Lemma~\ref{app:lem001},  $\ln k -\ln(k-1)=-\ln(1-1/k)>1/k>2(k+2)/k^r$, provided
$2+4/k<k^{r-2}$. The left hand side of $2+4/k<k^{r-2}$ is decreasing,
and the right hand side increasing, for all $r,\,k$. Thus we need only determine the
smallest pairs $r\geq 3,\,k\geq 2$ which satisfy it. These are
$k=2,\,r=5$,\ \, $k=3,\,r=4$ and $k=4,\,r=3$, which are not regular pairs.
\end{proof}

\begin{lemma}\label{app:lem017}
For all $k\geq 1$,  $4(k-1)\geq \sqrt{k}\ln k$.
\end{lemma}
\begin{proof}
Using Lemma~\ref{app:lem001}, $\ln k=2\ln\sqrt{k}\leq 2(\sqrt{k}-1)$.
So the conclusion is implied by
$2(k-1)\geq k-\sqrt{k}$, which follows from $2(k-1)\geq (k-1)$ for all $k\geq 1$.
\end{proof}

\begin{lemma}\label{app:lem019}
$r(r-1)\ln k/(k^{r-1}-1)<1$ for all $k\geq 3$, $r\geq 2$, or $k=2$, $r\geq 5$.
\end{lemma}
\begin{proof}
 Let $\phi(r,k)=r(r-1)\ln k/(k^{r-1}-1)$. Then, for $k\geq 3$, $r\geq 2$, or $k=2$, $r\geq 5$ we have
\[\frac{\phi(r+1,k)}{\phi(r,k)}\,=\,\frac{(r+1)(k^{r-1}-1)}{(r-1) (k^r-1)}
\,<\,\frac{r+1}{(r-1)k}\,\leq\,1. \]
Furthermore
\[\frac{\phi(r,k+1)}{\phi(r,k)}\,=\,
\frac{(k^{r-1}-1)\,\ln(k+1)}{((k+1)^{r-1}-1)\,\ln k}
\,<\,\frac{\ln(k+1)}{\ln k}\Big(\frac{k}{k+1}\Big)^{r-1}\,\leq\,
\frac{k/\ln k}{(k+1)/\ln(k+1)}\,<\,1 \]
for $k\geq 3$, from the proof of Lemma~\ref{app:lem014}. If $k=2$, $r\geq 5$
then
\[ \frac{\phi(r,3)}{\phi(r,2)}\,=\,\frac{(2^{r-1}-1)\ln 3}{(3^{r-1}-1)\ln 2}
\,<\,\frac{16\, \ln 3}{81\,\ln 2} <\, 1. \]
So $\phi$ is decreasing with both $r$ and $k$. Now we may calculate
that
\[ \phi(3,3)\,=\,3\ln 3/4\,<\,1,\qquad \phi(5,2)\,=\,4\ln 2/3\,<\,1.\qedhere\]
\end{proof}


\begin{thebibliography}{10}

\bibitem{AchFri99}
{\sc D.~Achlioptas and E.~Friedgut}, {\em A sharp threshold for
  $k$-colorability}, Random Struct. Algorithms, 14 (1999), pp.~63--70.

\bibitem{AcKiKT02}
{\sc D.~Achlioptas, J.~H. Kim, M.~Krivelevich, and P.~Tetali}, {\em
  Two-coloring random hypergraphs}, Random Struct. Algorithms, 20 (2002),
  pp.~249--259.

\bibitem{AchMoo06}
{\sc D.~Achlioptas and C.~Moore}, {\em Random k-SAT: Two moments suffice to
  cross a sharp threshold}, SIAM J. Comput., 36 (2006), pp.~740--762.

\bibitem{AchNao05}
{\sc D.~Achlioptas and A.~Naor}, {\em The two possible values of the chromatic
  number of a random graph}, Ann. of Math., 162 (2005), pp.~1333--1349.

\bibitem{Coj13}
{\sc A.~Coja-Oghlan},
{\em Upper-bounding the $k$-colorability threshold by counting covers},
Electron. J.~Combin., 20(3) (2013), P32.

\bibitem{CojEftHet13}
{\sc A.~Coja-Oghlan, C.~Efthymiou and S.~Hetterich},
{\em On the chromatic number of random regular graphs},
arXiv preprint,
\texttt{arxiv:1308:4287}, 2013.

\bibitem{CojVil13}
{\sc A.~Coja-Oghlan and D.~Vilenchik},
{\em Chasing the $k$-colorability threshold},
in Proc. 54th Annual IEEE Symposium on Foundations of Computer
Science, IEEE, 2013, pp.~380--389.  Full preprint at
\texttt{arxiv:1304:1063}, April 2013.

\bibitem{CojZde12}
{\sc A.~Coja-Oghlan and L.~Zdeborov{\'a}}, {\em The condensation transition in
  random hypergraph 2-coloring}, in Proc. 23rd Annual
  ACM-SIAM Symposium on Discrete Algorithms, SIAM, 2012, pp.~241--250.

\bibitem{FlKnPi89}
 {\sc P.~Flajolet, D.~E.~Knuth and B.~Pittel}, {\em The first cycles in an evolving graph},
 Discrete Math., 75 (1989), pp.~167--215.


\bibitem{GrJaRu10}
{\sc C.~Greenhill, S.~Janson, and A.~Ruci\'{n}ski}, {\em On the number of
  perfect matchings in random lifts}, Comb. Probab. Comput., 19 (2010),
  pp.~791--817.

\bibitem{HaLiPo88}
{\sc G.~H. Hardy, J.~E. Littlewood, and G.~P\'{o}lya}, {\em Inequalities},
  Cambridge University Press, 2nd~ed., 1988.

\bibitem{HatMol08}
{\sc H.~Hatami and M.~Molloy}, {\em Sharp thresholds for constraint
  satisfaction problems and homomorphisms}, Random Struct. Algorithms, 33
  (2008), pp.~310--332.

\bibitem{JaLuRu00}
{\sc S.~Janson, T.~{\L}uczak, and A.~Ruci\'nski}, {\em Random graphs},
  Wiley-Interscience, New York, 2000.

\bibitem{KriSud98}
{\sc M.~Krivelevich and B.~Sudakov}, {\em The chromatic numbers of random
  hypergraphs}, Random Struct. Algorithms, 12 (1998), pp.~381--403.

\bibitem{KupSha11}
{\sc A.~Kupavskii and D.~Shabanov}, {\em On $r$-colorability of random
  hypergraphs}, arXiv preprint,
\texttt{arXiv:1110.1249},  2011.

\bibitem{Shutle95}
{\sc P.~Shutler}, {\em Constrained critical points}, Amer. Math.
  Monthly, 102 (1995), pp.~49--52.

\bibitem{Spivak06}
{\sc M.~Spivak}, {\em Calculus}, Cambridge University Press, 3rd~ed., 2006.

\end{thebibliography}
\end{document}